\documentclass[aps,pra,onecolumn,superscriptaddress,nofootinbib]{revtex4}
\pdfoutput=1

\usepackage[page,toc,titletoc,title]{appendix}

\usepackage{footmisc}
\usepackage{hyperref}
\usepackage{mathtools}
\usepackage{amsmath,amssymb,graphicx,xcolor,braket,hyphenat,makeidx,comment,amsthm}
\usepackage{bbm}

\theoremstyle{theorem}

\newtheorem{theorem}{Theorem}[section]
\newtheorem{corollary}{Corollary}[theorem]
\newtheorem{lemma}{Lemma}[theorem]
\newtheorem{definition}{Definition}
\theoremstyle{remark}
\newtheorem{remark}{Remark}[section]

\newcommand{\red}{\color{red}}

\newcounter{bound}
\newcommand{\bound}{\stepcounter{bound} \bigskip {\bf Bound E\arabic{bound}.}}
\newcommand{\ketbra}[2]{\ket{#1}\!\!\bra{#2}}
\newcommand{\abs}[1]{\left| #1 \right|}
\newcommand{\norm}[1]{\left| \left| #1 \right| \right|_2}

\newcommand{\ltwo}[1]{\left|\!\left|#1\right|\!\right|_2}

\newcommand{\sumt}[1]{\sum_{#1 = -\frac{d-1}{2}}^{\frac{d-1}{2}} #1 \frac{T_0}{d} \ket{\theta_{#1}}\bra{\theta_{#1}}}
\newcommand{\sumS}[2]{\sum_{#1 \in \mathcal{S}_d(#2)}}
\newcommand{\errort}{e^{-\frac{\pi d}{4} (1-\beta)^2}}
\newcommand{\errortd}{e^{-\pi(1-\beta)}}
\newcommand{\sumnsym}{\sum_{n=-\frac{d-1}{2}}^{+\frac{d-1}{2}}}

\newcommand{\Mspace}{\vspace{0.2cm}} 

\newcommand{\gClock}{Quasi-Ideal clock}
\newcommand{\gClocks}{Quasi-Ideal clocks}

\newcommand{\supp}{appendix}

\newcommand{\cft}[1]{\mathcal{F}_d \left[ #1 \right]}
\newcommand{\cfti}[1]{\mathcal{F}^{-1}_d \left[ #1 \right]}
\newcommand{\be}{\begin{equation}}
\newcommand{\ee}{\end{equation}}
\newcommand{\ba}{\begin{align}}
\newcommand{\ea}{\end{align}}
\newcommand{\tr}{\mathrm{tr}}

\newcommand{\nn}{{\mathbbm{N}}}
\newcommand{\rr}{{\mathbbm{R}}}
\newcommand{\cc}{{\mathbbm{C}}}
\newcommand{\hh}{{\mathbbm{H}}}
\newcommand{\zz}{{\mathbbm{Z}}}
\newcommand{\id}{{\mathbbm{1}}}

\newcommand{\me}{\mathrm{e}}
\newcommand{\mi}{\mathrm{i}}

\newcommand{\bo}{\mathcal{O}}

\newcommand{\suppl}{\text{main text}}
\newcommand{\comm}{\text{manuscript}}

\usepackage[normalem]{ulem} 
\theoremstyle{definition}
\newtheorem{result}{Result}

\begin{document}

\title{Autonomous quantum machines and the finite sized Quasi-Ideal clock}

\begin{abstract}
  
Processes such as quantum computation, or the evolution of quantum cellular automata are typically described by a unitary operation implemented by an external observer. In particular, an interaction is generally turned on for a precise amount of time, using a classical clock. A fully quantum mechanical description of such a device would include a quantum description of the clock whose state is generally disturbed because of the backreaction on it. Such a description is needed if we wish to consider finite sized autonomous quantum machines requiring no external control. The extent of the backreaction has implications on how small the device can be, on the length of time the device can run, and is required if we want to understand what a fully quantum mechanical treatment of an observer would look like. Here, we consider the implementation of a unitary by a finite sized device which we call the ``Quasi-Ideal clock", and show that the back-reaction on it can be made exponentially small in the device's dimension with only a linear increase in energy. As a result, an autonomous quantum machine need only be of modest size and energy. We are also able to solve a long-standing open problem  by using a finite sized quantum clock to approximate the continuous evolution of an idealised clock. The result has implications for how well quantum devices can be controlled and on the equivalence of different paradigms of control.

\end{abstract}

\author{Mischa P. Woods}
\affiliation{Institute for Theoretical Physics, ETH Zurich, Switzerland}
\author{Ralph Silva}
\affiliation{Group of Applied Physics, University of Geneva, Geneva, Switzerland}
\author{Jonathan Oppenheim}
\affiliation{University College of London, Department of Physics \& Astronomy, London, United Kingdom}
\maketitle

\tableofcontents
\section{Introduction}

Many recent advances in quantum theory are due to our ability to manipulate small systems. Witness on the experimental front, the progress in quantum computation, quantum memory, non-locality, quantum thermodynamics, and randomness generation. The quantum machinery involved in each of these typically requires very precise external control | for example, in a quantum cellular automata, or a quantum computation, a unitary (gate) is applied at each time step.\Mspace

This is reasonable for modeling a large machine which is controlled by a classical system, but what if we wish to consider a fully quantum machine? This could be an autonomous quantum device which interacts with its surroundings, or a way of modelling a fully quantum observer. While the latter is mostly of foundational interest, the former is needed to understand and optimise current quantum technologies, or understand important physical processes. Take for example, molecular machines or nanomachines such as molecular motors \cite{howard1997molecular}, which are important in biological processes \cite{molecularBio}, or distant technologies such as nanorobots \cite{nanobots}, where quantum effects on the control mechanism, and the back-reaction they incur, are likely to be significant.\Mspace

Understanding autonomous machines is particularly important in thermodynamics, where one is interested in devices which can be used for tasks such as energy harvesting or erasing a memory \cite{Scovil1959masers,Geusic1967quatum,linden2010small,brask2015autonomous,brandao2013resource,malabarba2015clock,tonner2005autonomous,gelbwaser2014heat,correa2014quantum,tonner2007quantum}. In the thermodynamics literature, there tend to be two sorts of processes, those which are fully autonomous, and those which allow a certain level of external control at no cost to the agent. A canonical example of the former is the brownian ratchet, popularised by Feynman \cite{FeynamnLecs}, which simply sits between two thermal baths and extracts work in situ. There are a number of autonomous quantum thermal machines built on this principle \cite{linden10,brask2015autonomous,MarcusPauli}. However, there are a number of processes, such as quantum Carnot cycles \cite{geusic,gelbwaser2014heat}, and thermal operations, on which a number of resource theories are based, and from which one can derive the quantum version of the second law, that require external control. While an autonomous thermal machine can be implemented via a fixed time-independent interaction Hamiltonian, an externally controlled machine requires a time-dependent Hamiltonian. For example, in order to implement a unitary operation, an interaction Hamiltonian must be switched on and allowed to run for a specific amount of time.\Mspace


Allowing such external control is highly contentious as the hidden cost of such fine-tuned control may often dwarf the perceived costs of such theories. If one insists (as one should) that the true cost of quantum process can only be observed reliable in the absence of external control, a number of questions present themselves - Is there a fundamental disadvantage to autonomous machines as opposed to those with external control? How sensitive are the conclusions of the various ``resource theories" to the removal of external control?\Mspace

Note that if one has access to a system of infinite dimensions and a Hamiltonian unbounded from below, that we may understand as an ``idealized clock", then it turns out that the paradigm of allowing external control vs the paradigm of autonomous machines are equivalent, for instance, see \cite{brandao2013resource,malabarba2015clock}. However, while illustrative as a conceptual proof of principle that controlled processes can be turned autonomous, these analyses are misleading for two reasons.\Mspace

First of all, the Hamiltonian being unbounded from below implies that the clock has infinite energy, which is unphysical. Furthermore, accounting for changes in energy, which is an important part in the analysis of quantum operations, is rendered meaningless in the presence of a control device of infinite energy.\Mspace

Also, a careful analysis requires one to consider a finite dimensional clock, because it is important to show 
that an infinite dimensional clock cannot be used to {\it embezzle} work from it, an issue covered at length in \cite{second}. Embezzling of work, based on the notion of entanglement embezzling \cite{Hayden-embezzling}, is the process of transferring work from a system while only changing the state of the system by an amount (w.r.t. trace distance) which vanishes in the limit of increasing dimension.\Mspace




Thus, in order to reliably consider the changes in energy and entropy (the two central quantities in thermodynamics) in a controlled quantum process, one requires a physically reasonable control of both finite size and energy, and take into account the back-reaction on it. In fact, there is a wealth of interesting physics that presents itself when we recognize the finiteness of control systems, such as the advantage of coherence in control \cite{HesenbergLim}, the tradeoff between accuracy and power in thermodynamics \cite{MarcusPauli}, as well as fundamental bounds on the synchronization time of clocks \cite{sandraAlternative}.\Mspace


There are two important limitations of finite clocks. The first is that they can only record the precise time at discrete intervals \cite{Peres}, and can be very inaccurate in between. Furthermore, any attempt to use the clock to measure time or as a control system disturbs the clock \cite{Peres,buzek,allcock1969time}, leading to the performance of the clock degrading on further use.\Mspace


In this paper we are able to circumvent these two difficulties. We present a finite size quantum clock, based upon the Hamiltonian of the Wigner clock \cite{SaleckerWigner}, but whose initial state is a coherent superposition w.r.t. the basis that Peres used \cite{Peres}. To demonstrate the clock's utility, we describe how to convert two of the most ubiquitous externally controlled operations in quantum theory: the unitary, and the time-dependent interaction Hamiltonian, into an operation performed by an autonomous device.  We compute the back-reaction on the clock, and find analytic bounds on the errors developed in the clock and target system, thus explicitly accounting for the cost of these operations that form the basis of so many theoretical paradigms. Our main result is that we find that the disturbance in the clock can be made exponentially small in the dimension of the clock, a calculation which requires going beyond perturbation theory. We thus see that the back-reaction can be made negligible, and one can make a device autonomous using a control of modest size as well as energy. One significance of this result, is that it allows an autonomous machine, even one of small size, to run for a significant length of time, before becoming too degraded.\Mspace

In fact, we demonstrate that the evolution of the clock mimics that of the idealized clock (up to the exponentially small error). As such, if one wishes to convert any quantum operation with external control into an autonomous process, one may do so by using the idealized momentum clock (as in \cite{malabarba2015clock}), whose description is simple, and keep track of the real error by using the results presented in this manuscript. Thus one can account for the cost of turning a quantum process autonomous \emph{without having to explicitly describe the control}, an immense advantage for the cases wherein to do so would be either analytically intractable, or computationally intensive \cite{maxclock}.\Mspace


On a foundational note, the behaviour of an idealized quantum clock is equivalent to it obeying the canonical commutation relation. Given a suitable definition of a time operator we demonstrate that our finite clock states also approximate the canonical commutator relation | a property which is absent in the clock proposed by \cite{SaleckerWigner,Peres}. Importantly, this property, together with the quasi-ideal evolution of our clock, are both consequences of the coherent nature of the clock state. This highlights the importance of quantum coherence in quantum clocks and control.\Mspace


{\bf\emph{Organisation of this \comm. - }} Given the high volume of material presented, we summarize the main results and discussion first; including a conclusion, and provide the full theorems and technical proofs thereafter. To begin with, Section \ref{sec:ideal} introduces the infinite dimensional clock highlighting its relevant properties that we wish to mimic. This is followed by an introduction to finite clocks in Section \ref{sec:The finite clock main text}, in particular to the complex Gaussian superposition that we study. The main results of our work are then stated and explained in Section \ref{sec:RESULTS}, including a discussion on the implications for quantum autonomous control in Subsection \ref{sec:Quasi-Canonical commutator}. This is followed by a general discussion and conclusions in Sections \ref{sec:discussion} and \ref{sec:conclusions} respectively. In the \suppl~we prove the results presented in this \comm. The results themselves are presented with more generality in the form of theorems. The organization of the \suppl~is described directly after the conclusions (Section \ref{sec:conclusions}).
\section{The idealised quantum clock and its properties}\label{sec:ideal}

The notion of an ideal clock is closely related to whether there exists a time operator in quantum mechanics (i.e. time is an observable). Wolfgang Pauli \cite{pauli1,pauli2} argued that if there exists an ideal clock with Hamiltonian $\hat H$ and ideal observable of time $\hat{t}$; both self-adjoint on some suitably defined domains, then in the Heisenberg picture the pair must obey
\begin{equation}\label{eq:div hat t}
\frac{d}{dt} \hat t(t)= \id,\quad \forall\, t \in \rr,
\end{equation}
which in turn implies the canonical commutation relation $-\mi [\hat t(t), \hat H]=\id, \quad \forall t\in \rr$ (we use units so that $\hbar =1$) on some suitably defined domain. Pauli further argued that the only pair of such operators (up to unitary equivalence) are $\hat t=\hat x$ and $\hat H=\hat p$, where $\hat x$, $\hat p$ are the canonically conjugate position-momentum operators of a free particle in one dimension. However, all such representations of $\hat p$ for which Eq. \eqref{eq:div hat t} is satisfied have spectra unbounded from below. One can thus conclude that no perfect time operator exists in quantum mechanics since such clock Hamiltonians would require infinite energy to construct due to the lack of a ground state. The question of whether a physically realizable perfect time operator exists in quantum mechanics is still a contentious issue, see Remark \ref{rem:Paulidiscussion}. We will not dwell upon this issue here, but rather summarize the characteristic properties of such a system that allow for the precise timing of events,which we then mimic using a finite sized clock.\Mspace

For the simple case in which $\hat t=\hat x_c$, $\hat H=\hat p_c$ are the position and momentum operators of a free particle in one dimension\footnote{The domain of all of the operators is taken to be $D_0$, the space of infinitely differentiable functions of compact support on $L^2(\rr)$.\label{domainxp}}, the dynamics are easily solvable. We refer to this as the \textit{idealised clock}. It has three critical properties (that we proceed to discuss in detail), namely:
\begin{itemize}
	\item [1)] The clock possesses a distinguishable basis of ``time states",
	\item [2)] it demonstrates ``continuity", and 
	\item [3)] it allows for perfect continuous autonomous control on an external system.
\end{itemize}

To be more precise, given that the Hamiltonian of the clock is $\hat{H} = \hat{p}_c$, the generalised eigenvectors of the position operator $\ket x$ are a \emph{distinguishable basis of time states} by which we mean $\braket{x|x^\prime} = \delta(x-x^\prime)$, and that given any initial generalized eigenvector $\ket{x}$, the natural evolution of the clock will past through all of the positions $x'>x$,
\begin{equation}
	e^{-\mi t\hat{H} } \ket{x} = \ket{x+t},
\end{equation}
i.e. a time translation is equivalent to a spatial translation.\Mspace

For the second property, note that the equivalence between time and space translations hold for any state of the clock, 
\begin{equation}\label{idealregularity_main}
	\braket{x|e^{-\mi t\hat{H} }|	\Psi} = \braket{x-t|\Psi}.
\end{equation}

The fact that this statement holds for all $x,t\in\rr$, and in particular, for arbitrarily small $t$ is what we refer to as \emph{continuity}, or by referring to the clock as \emph{continuous}.\Mspace

To demonstrate the third property, that of perfect control, observe that if one adds a position-dependent potential to the clock, it still remains continuous, while its state is only modified by a phase that depends on the potential,
\begin{equation}\label{idealpotential}
	\braket{x | \me^{-\mi t\left(\hat{H} + V(\hat{x}_c)\right)} | \Psi} = \me^{-\mi\int_{x-t}^x V(x^\prime) dx^\prime} \braket{x-t | \Psi},\quad
	 x,t\in\rr, \,V\in D_0.\footref{domainxp}
\end{equation}

Notice that the phase integrates over the potential in the region that the state passes through ($[x-t,x]$). The clock can therefore be turned into a control device by simply having the potential be an interaction on an external system, whose strength is a function of the clock's position. Rather than an observer having to switch on and off an interaction on a system, here the clock does so autonomously by passing through the region of the potential.\Mspace

A detailed proof of the above statements may be found in Section \ref{idealizedclock}, while the consequences for autonomous control of external systems is discussed in \ref{sec:Consequences of Quasi-Autonomous control}.

\section{A finite clock to mimic the idealised clock}\label{sec:The finite clock main text}

The clock we propose is based upon a quantum system that has been discussed before in \cite{SaleckerWigner,Peres}. The system has dimension $d$ and equally spaced (normalised) energy eigenstates $\ket{E_n}$, i.e. its Hamiltonian is
\begin{equation}\label{finiteHamiltonian}
	\hat{H}_c = \sum_{n=0}^{d-1} n \omega \ketbra{E_n}{E_n}.
\end{equation}

The frequency $\omega$ determines both the energy spacing as well as the time of recurrence of the clock, $T_0 = 2\pi/\omega$, as $e^{-i \hat{H}_c T_0} = \id_c$. This system possesses a distinguishable basis of time states $\left\{\ket{\theta_k}\right\}_{k=0}^{d-1}$, that is mutually unbiased w.r.t. the energy eigenstates,
\begin{align}\label{finitetimestates_main}
	\ket{\theta_k} &= \frac{1}{\sqrt{d}} \sum_{n=0}^{d-1} e^{-i2\pi n k/d} \ket{E_n}.
\end{align}

It will also be useful later to have the range of $k$ extended to $\zz$.\footnote{Note that $k$ will belong to a set of only $d$ consecutive integers so that $\{\ket{\theta_k}\}$ form a complete orthonormal basis without repetition.} Extending the range of $k$ in Eq.  \eqref{finitetimestates_main} it follows that $\ket{\theta_k}=\ket{\theta_{k \textup{ mod. } d}}$ for $k\in\zz$. The $\ket{\theta_k}$ are referred to as time states because they rotate into each other in regular time intervals of $T_0/d$, i.e. $e^{-i \hat{H}_c T_0/d} \ket{\theta_k} = \ket{\theta_{k+1}}$. Since they also form an orthonormal basis, this property is true for any state $\ket{\Psi}$ of the clock Hilbert space,
\begin{equation}\label{finiteregularity_main}
	\braket{\theta_k | e^{-i \hat{H}_c m\, T_0 /d} |\Psi} = \braket{\theta_{k-m} | \Psi},\quad k,m\in\zz.
\end{equation}

This is reflective of the idealised case Eq. \eqref{idealregularity_main}, but the key difference is that for the finite clock, Eq. \eqref{finiteregularity_main} only holds at \emph{regular} intervals ($m \in \mathbb{Z}$). Thus while every state of the clock is regular with respect to time, a general clock state does not demonstrate continuity. In fact the time-states $\ket{\theta_k}$ themselves are considerably discontinuous. Specifically, as a time-state evolves, it spreads out considerably in the time basis for non-integer intervals \cite{Gross2012,sergedft}. Additionally, the time-states fail to even approximate the canonical commutator relationship between a (suitably defined) time operator and Hamiltonian in any limit \cite{Peres}. See Section \ref{SWPclock} for a more precise discussion of the behaviour of time-states $\ket{\theta_k}$.\Mspace

The clock that we work with henceforth, rather than being a time-state, is instead a coherent complex Gaussian superposition of time-states,
\begin{align}\label{gaussianclock_main}
	\ket{\Psi_\textup{nor}(k_0)} &= \sum_{\mathclap{\substack{k\in \mathcal{S}_d(k_0)}}} \psi_\textup{nor}(k_0;k) \ket{\theta_k}, \\
	\text{where} \quad \psi_\textup{nor}(k_0;x) &= A e^{-\frac{\pi}{\sigma^2}(x-k_0)^2} e^{i 2\pi n_0(x-k_0)/d}.
\end{align}
For reasons which will become clear, we refer to these states as \textit{\gClock}~states.
The precise definition and basic properties of these states are detailed in Section \ref{sec:Definition of Gaussian clock}. Roughly, $k_0\in\rr$ represents the mean position of the clock about which the Gaussian is centred ($\mathcal{S}_d(k_0)$ is a set of $d$ consecutive integers centered about $k_0$, such that the set $\{ \ket{\theta_k}: k\in \mathcal{S}_d(k_0) \}$ forms a complete orthonormal basis for the clock Hilbert space). $\sigma$ denotes the width of the state in the time basis, ranging from $\sigma\approx 0$ (approximately a time-state) to $\sigma \approx  d$ (almost an energy eigenstate). $n_0\in (0,d-1)$ represents the \textit{mean energy number} of the clock, so that the average energy $\omega n_0$ ranges between $0$ and $\omega(d-1)$. $A$ is a normalisation constant.\Mspace

In the next section we quantify how closely a clock state of the above form mimics the behaviour of an idealised clock, and the consequences for using the clock as a control device.

\section{Results (Overview)}\label{sec:RESULTS}

We now briefly state and explain our main results, reserving the full theorems for later. For simplicity, we state here the results for the simplest case of the clock state, where $\sigma = \sqrt{d}$, and $n_0 = (d-1)/2$. This corresponds to a state, that when expressed in the energy eigenbasis, has a width (i.e. the standard deviation w.r.t. the energy spacing $\omega$) that is also approximately $\sqrt{d}$, and whose mean energy is at about the middle of the energy spectrum. We discuss the more general case in Section \ref{sec:discussion}. In the following, we use $poly(d)$ to refer to a polynomial in $d$ and use $\bo$ for Big-O notation.

\subsection{Quasi-Continuity}
Our first result is to recover the \emph{continuity} of a quantum clock for the class of complex Gaussian superpositions of time-states introduced above.

\begin{result}[See Theorem \ref{gaussiancontinuity} for the most general version] If the clock begins in a complex Gaussian superposition centred about $k_0$, i.e. $\braket{\theta_k|\Psi_\textup{nor}(k_0)} = \psi_\textup{nor} (k_0;k)$, then for all $k\in\mathcal{S}_d(k_0+t \,d/T_0)$ and $t\in\rr$,
\begin{align}\label{eq:result 1 main text}
\begin{split}
 \braket{\theta_k|e^{-\mi t \hat H_c}|\Psi_\textup{nor}(k_0)} &= \psi_\textup{nor}(k_0 - d \frac{t}{T_0} ; k) + \braket{\theta_k|\epsilon_c},\\
 \text{where} \quad |\!\braket{\theta_k|\epsilon_c}\!|&\leq\varepsilon_c(t,d) = \bo \left( t\; poly(d)\; e^{-\frac{\pi}{4} d} \right) \,\textup{as } d\rightarrow \infty.
 \end{split}
\end{align}
In other words, the evolution of the clock state, which is composed of the discrete coefficients $\braket{\theta_k|e^{-\mi t \hat H_c}|\Psi_\textup{nor}(k_0)}$, may be approximated by the continuous movement of the background Gaussian function $\psi_\textup{nor}(k_0;x)$ (see Fig. \ref{fig:clocks main text} a). This mimics the equivalence of time and space translations of the idealised clock, Eq. \eqref{idealregularity_main}, and crucially holds for arbitrarily small time intervals $t$, in contrast to a clock that is a single time state $\ket{\theta_k}$. The error in the approximation is bounded to be linear in time, but exponentially small in the dimension.
\end{result}

From another perspective, our set-up can be viewed as a specific example of a continuous time quantum walk. The authors of \cite{Gross2012} studied such walks when the initial state has support on a fixed subset of $\{\ket{\theta_k}\}_{k=0}^{d-1}$, and showed that the support of the state necessary spreads out to occupy the entire space $\{\ket{\theta_k}\}_{k=0}^{d-1}$ in finite time. The above result demonstrates, that these no-go theorems can be circumvented for some states whose support is exponentially suppressed in regions away from a central region.

\subsection{Autonomous Quasi-Control}

For the clock to serve as a quantum control unit, it must be continuous not only under its own evolution, but also in the presence of an appropriate potential, as in the case of the idealised clock \eqref{idealpotential}. Our main result is to prove a similar behaviour for the finite clock.

\begin{result}[See Theorem \ref{movig through finite time} for the most general version]

Let $V_0: \rr\rightarrow \rr$ be an infinitely differentiable, periodic function with period $2\pi$, normalized so that its integral over a period is $\Omega$, with $-\pi\leq \Omega<\pi$. For the finite clock of dimension $d$, construct a potential from $V_0$ as follows,
\begin{align}
	\hat{V}_d &= \frac{d}{T_0} \sum_{k=0}^{d-1} {V}_d(k) \ketbra{\theta_k}{\theta_k}, &
	 \quad V_d(x) &= \frac{2\pi}{d} V_0 \left( \frac{2\pi x}{d} \right).
\end{align}

Then under the Hamiltonian $\hat H_c + \hat V_d$, if the clock begins in a complex Gaussian superposition $\braket{\theta_k|\Psi_\textup{nor}(k_0)} = \psi_\textup{nor} (k_0;k)$, then for all $k\in\mathcal{S}_d(k_0+t \,d/T_0)$ and $t\in\rr$,
\begin{align}\label{eq:result 2 main text}
\begin{split}
 \braket{\theta_k|e^{-\mi t ( \hat H_c + \hat V_d )}|\Psi_\textup{nor}(k_0)} =&\,\me^{-\mi \int_{k-d\,t/T_0}^k V_d(x)dx}\, \psi_\textup{nor}(k_0 - d \frac{t}{T_0};k)
  + \braket{\theta_k|\epsilon_v}, \\\;\;\;\;\; \text{where} \quad |\!\braket{\theta_k|\epsilon_v}\!|&\leq\varepsilon_v(t,d) =\, \bo\left( t\; poly(d)\; e^{-\frac{\pi}{4} \frac{d}{\zeta}} \right)\,\textup{as } d\rightarrow \infty,
 \end{split}
\end{align}
where $\zeta \geq 1$ is a measure of the size of the derivatives of $V_0(x)$,
\begin{align}\label{eq:b zeta main text}
\begin{split}
\zeta &=\left( 1+\frac{0.792\, \pi}{\ln(\pi d)}b \right)^2,\quad\text{for any } \\
b &\geq\; \sup_{k\in\nn^+}\left(2\max_{x\in[0,2\pi]} \left|  V_0^{(k-1)}(x) \right|\,\right)^{1/k},
\end{split}
\end{align}
where $ V_0^{(k)}(x)$ is the $k^\textup{th}$ derivative with respect to $x$ of $ V_0(x)$. Roughly speaking - the larger the derivatives, the larger is $\zeta$. [see Fig. \ref{fig:clocks main text} b].\footnote{We further require the derivatives $V_0$ to be such that the lower bound on $b$ in Eq. \eqref{eq:b zeta main text} is finite, in order for Eq. \eqref{eq:result 2 main text} to be non-trivial.}

Note the construction of the potential is analogous to the case of the idealised clock; there the potential is expressed in the conjugate variable ($\hat x_c$) to the Hamiltonian ($\hat p_c$); in the case of the finite clock the potential is diagonal in a basis (of time states) that is mutually unbiased w.r.t. the energy eigenbasis.\Mspace

As in the case of the idealised clock, this result allows one to implement a unitary on an external system, or any time-dependent interaction that commutes with the Hamiltonian of the system. We discuss this in greater detail in Section \ref{sec:Consequences of Quasi-Autonomous control}, and bound the errors in the system and the clock at the end of such processes (in comparison to the idealised case). In particular, we will show that the state of the clock is disturbed by only an exponentially small amount in the dimension of the clock.\Mspace

\emph{The clock as a switch.} - Note that in the case of the idealised clock, the ``steepness" of the potential (i.e. the magnitude of the derivatives) does not affect the continuity of the clock, i.e. do not enter in Eq. \eqref{idealpotential}. Thus there is no limit to how quickly the idealised clock could switch on and off an interaction on an external system. This is no longer the case for the finite clock, since the steeper the potential, the larger the error in the continuity of the clock. In the discussion of the clock as a control device in Section \ref{sec:Consequences of Quasi-Autonomous control}, we will see how the above result leads to a natural tradeoff between how sharply the clock can switch on and off interactions on the one hand, vs the back-reaction on the clock and accuracy of the process implemented on the external system on the other hand.\Mspace

In further work \cite{RMRenatoetal}, it is demonstrated that the result above (in it's more general form presented in Theorem \ref{movig through finite time}) also encapsulates continuous weak measurements on the state of the clock, that may be used to extract temporal information from the clock (i.e. to measure time).\Mspace

To conclude this section, we will briefly point out how this work is related to the so-called Pegg and Barnet phase operator, $\hat \Phi_\textup{PB}= \sum_{k=0}^{d-1} \left(\vartheta_0+\frac{2\pi}{d}k \right)\ketbra{\theta_k}{\theta_k}$, $\vartheta_0\in\rr$  \cite{PeggBarnet89}. In particular, we see that the potential $\hat V_d$ is a function of it. Specifically $\hat V_d=\omega V_0\left(\hat\Phi_\textup{PB}-\vartheta_0\id_d\right)$. Note that the Phase operator $\hat\Phi_\textup{PB}$ was introduced by A. Peres under the name of a time operator $\hat t_c$, approximately 10 years prior to Pegg and Barnet, as we will see in Section \ref{sec:Quasi-Canonical commutator}.

\begin{figure}
	\begin{center}
		\includegraphics[scale=0.2]{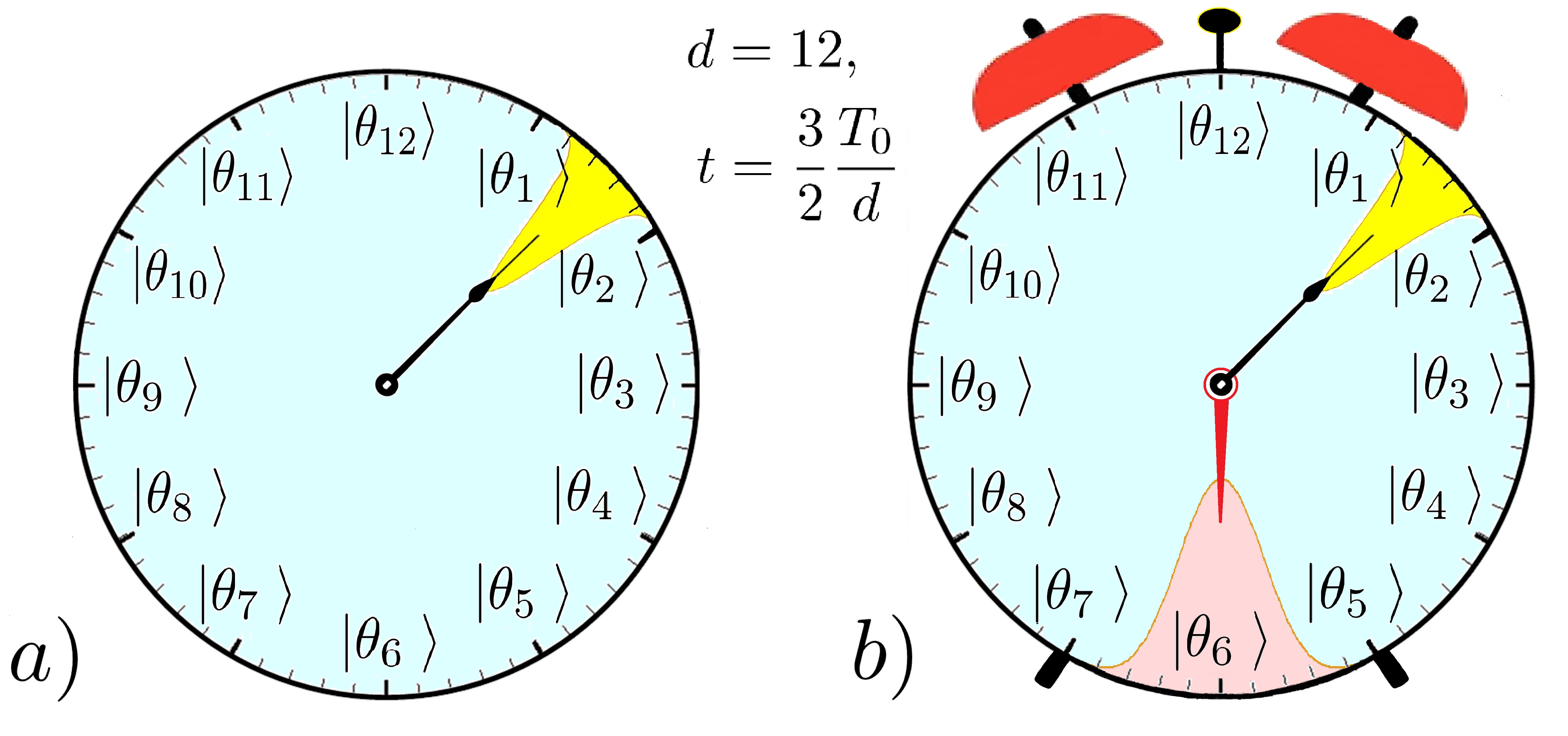}
		\caption{
			\textbf{Resemblance of a quantum clock to a classical clock.} \textbf{a)} Illustration of main result 1. A quantum clock state $\ket{\Psi_\textup{nor}(12)}$ initially centered at ``12 o'clock" time evolved to $t=3 T_0/(2 d)$. The clock state moves continuously around the clock face, maintaining an approximately constant width, in contrast to the behaviour of time states $\ket{\theta_k}$. \textbf{b)} Illustration of main result 2. Now a potential peaked around ``6 o'clock" is introduced. The clock dynamics will be, up to a small error the same as in a) until a time just before ``6 o'clock", at which point the clock will start to acquire a phase due to its passage over the potential. At a time just after ``6 o'clock", the dynamics will start to resemble those of a) again, but now with a global phase added. A direct consequence of main result 2, is that one can use such a potential to induce a time dependent phase that can control an auxiliary quantum system, resembling the working of an alarm clock.
		}\label{fig:clocks main text}
	\end{center}
	\end{figure}
\end{result}

\subsection{Quasi-Canonical Commutation}\label{sec:Quasi-Canonical commutator}

So far we have demonstrated that the finite dimensional \gClock~state, Eq. \eqref{gaussianclock_main}, satisfies the three properties of the idealized clock, up to a small quantifiable error. Recall that for the idealized clock, these desired properties were a direct consequence of the existence of a perfect time operator, Eq. \eqref{eq:div hat t}, which is satisfied if $[\hat t,\hat H]=\mi $ on some appropriately defined domain. We now discuss the issue of the commutator. For the finite clock, since the $\ket{\theta_k}$ rotate into each other in regular time intervals, an intuitive definition of a time operator is in the eigenbasis of $\ket{\theta_k}$. Therefore, in analogy to the energy operator, $\hat H_c$, we define the time operator \cite{Peres}
\begin{equation}\label{eq:t op main text def}
	\hat{t}_c = \sum_{k=0}^{d-1} k \frac{T_0}{d} \ketbra{\theta_k}{\theta_k}.
\end{equation}

However, as noted by A. Peres \cite{Peres}, this operator cannot obey $[\hat t_c, \hat H_c]= \mi$ for any time state $\ket{\theta_k}$. In fact,
\begin{equation}
	\bra{\theta_k} [\hat t_c, \hat H_c] \ket{\theta_k} = 0 \quad\text{for all } k\in\zz \text{ and } d\in\nn^+.
\end{equation}

This is intrinsically related to the time-states not being good clocks themselves. However, while the time and energy operators cannot obey the canonical commutation relation themselves, they do approximate it when applied to the subspace of clock states that we propose. Specifically, for complex Gaussian superpositions of time states $\ket{\Psi_\textup{nor}(k_0)}$,
\begin{result}[See Theorem \ref{Quasi-Canonical commutation} for the most general version]
\begin{equation}\label{eq:main result 3 amin text}
\begin{split}
	[\hat t_c,\hat H_c]\ket{\Psi_\textup{nor}(k_0)}&=\mi \ket{\Psi_\textup{nor}(k_0)} +\ket{\varepsilon_\textup{comm}},\\
	\text{where} \quad \|\ket{\varepsilon_\textup{comm}}\|_2 &=\bo\left(poly(d)\; \me^{-\frac{\pi}{4}d} \right)\,\textup{as } d\rightarrow \infty.
	\end{split}
\end{equation}
The implication of Eq. \eqref{eq:main result 3 amin text} is that our \gClock~states $\ket{\Psi(k_0)}$ can achieve the canonical commutation relation up to a exponentially small error in clock dimension. In addition, since the l.h.s. of Eq. \eqref{eq:main result 3 amin text} is $T_0$ independent, so is the error $\|\ket{\varepsilon_\textup{comm}}\|_2$.
\end{result}

See Remark \ref{rem:result3 Serge} that relates the above result to previous work on approximating the canonical commutator relation in finite dimension.

\subsection{Consequences of Autonomous Quasi-control}\label{sec:Consequences of Quasi-Autonomous control}
We now discuss the important consequences that Result 2 has for quantum control. Consider a quantum system of dimension $d_s$, upon which we require a unitary $\tilde{U}$ to be implemented within a time interval $t\in[t_1,t_2]$. More precisely, we would like that for any initial state $\rho_s\in\mathcal{S}(\mathcal{H}_s)$,
\be\label{eq:rho initial final}
\rho_s(t)=
\begin{cases}
\rho_s &\mbox{ if } 0\leq t< t_1\\
\tilde U \rho_s \tilde U^\dag   &\mbox{ if } t> t_2.
\end{cases}
\ee 
Such an operation could represent, for example, the implementation of a quantum gate in a quantum computer. The operation can be implemented by applying an interaction on the system in the given time interval, for instance via the time dependent Hamiltonian $\hat H(t)=\hat H_s^{int} g(t)$ acting on $\mathcal{H}_s,$
where $\tilde U=\me^{\mi \hat H^{int}_s},$
and $g\in L(\rr: \rr_{\geq 0})$ is a normalised pulse within the time interval, i.e. $\int_{t_1}^{t_2} dx\, g(x)=1,$ with support $[t_1,t_2]$.
Indeed, the state of the system is found to be
\begin{align}\label{eq:t dependent ham}
	\rho_s(t) &= U(t)\rho_s U^\dag(t), \quad\text{where} \quad U(t) = \me^{-\mi \hat H_s^{int} \int_0^t dx g(x)},
\end{align}
thus implementing the desired unitary Eq. \eqref{eq:rho initial final}. However, the interaction Hamiltonian $\hat H(t)$ is time dependent, and the pulse $g(t)$ represents an external observer modulating its strength, which makes the entire operation non-autonomous.\Mspace

One can describe the above autonomously by using the idealized (albeit unphysical) clock, as in \cite{malabarba2015clock}, by having the pulse $g(t)$ modulated by the clock itself, as a position dependent potential $g(x)$. In Section \ref{sec:Automation via the idealised clock}, we demonstrate that for the idealized clock there is no back-reaction due to the potential, and the unitary $\tilde{U}$ is implemented perfectly.\Mspace

We now show that a direct consequence of the autonomous quasi-control result, is that the unitary can be implemented autonomously with the aid of the finite clock, while only incurring a small error due to the finite nature of the clock. This is done analogously to the idealised case, by implementing the pulse $g(t)$ by a potential term $V_0$ that is added to the clock's own Hamiltonian. To be more precise,
\begin{result}[See Section \ref{sec:Implementing Energy preserving unitaries with the finite} for the proof and explicit construction]

There exists a $V_0$ (and an associated $\zeta<\infty$ given by Eq. \eqref{eq:b zeta main text}) such that for all unitaries $\tilde U$,  initial states $\rho_s\in\mathcal{S}(\mathcal{H}_s)$, and time intervals $t\in[0,t_1]\cup [t_2,T_0]$, for all $0<t_1<t_2<T_0$, the evolution of the initial state $\rho_{sc}'(0)=\rho_s\otimes\ketbra{\Psi_\textup{nor}(0)}{\Psi_\textup{nor}(0)}$ under the time \textit{independent} Hamiltonian $\hat H_{sc} = \id_s \otimes \hat{H}_c + \hat H^{int}_s \otimes \hat{V}_d,$
denoted by $\rho_{sc}'(t)=\me^{-\mi t \hat H_{sc}}\, \rho_{sc}'(0) \,\me^{\mi t \hat H_{sc}}$ satisfies
\be\label{eq:q-auto contrl trace bound}
\| \rho_s(t)-\rho_s'(t)\|_1\leq \sqrt{d_s\tr[\rho_s^2]} \,\Big(\varepsilon_s(t,d)+\varepsilon_v(t,d)\big(2+\varepsilon_v(t,d)\big)\Big),
\ee
where $\rho_s'(t)=\tr_c[\rho_{sc}'(t)]$ denotes the partial trace over the clock, $\rho_s(t)$ is given by Eq. \eqref{eq:t dependent ham}, and the error $\varepsilon_s\geq 0$ is independent of both $\rho_s$ and $d_s$. The explicit form of $\varepsilon_s$, specifying how it depends on the potential $V_0$, is provided in Section \ref{sec:Clocks as Quantum control}, see Corollary \ref{lem:new g}.
\end{result}


Before discussing Eq. \eqref{eq:q-auto contrl trace bound} in more detail, we also introduce a dynamical measure of disturbance for the clock. Since the clock state $\rho_c'(t)=\tr_s[\rho_{sc}'(t)]$ undergoes periodic dynamics with periodicity $T_0$ when no unitary is implemented, i.e. when $V_0(x)=0$ for all $x\in\rr$, the difference in trace distance between the initial and final state after one period is exactly zero. Moreover, when $V_0\neq 0$ any difference between the two states is \textit{solely} due to the back-reaction caused by the potential implementing the unitary on the system.\Mspace

A simple application of Result 2, allows for a direct characterization of this disturbance,
\be\label{eq:clock finite trace diff}
\|\rho_c'(0)-\rho_c'(T_0)\|_1\leq
 2\,\varepsilon_v(T_0,d).
\ee

Eqs. \eqref{eq:q-auto contrl trace bound} and \eqref{eq:clock finite trace diff}, represent a trade-off between the back-reaction on the clock dynamics on the one hand, and how well the clock acts as a switch on the other hand. The back-reaction $\varepsilon_v$ (introduced in Eq. \ref{eq:result 2 main text}) is minimized by having the potential $V_0$ be as ``gradual" as possible (i.e. small $\zeta$). However, the error $\varepsilon_s$ is minimized by having the unitary implemented as much as possible only within the time interval $[t_1,t_2]$, and this requires a narrower (and thus steeper) potential (i.e. large $\zeta$). See Fig. \ref{fig:tradeoff} for a visual discussion of the effect of the potential on the two error terms.\Mspace

We can also investigate how the tradeoff depends on $d$.  Depending on how one parametrizes the potential $V_0$ with $d$, the decay rates of $\varepsilon_v$ and $\varepsilon_s$ as a function of $d$ will be different. We highlight here two extremal cases. The first case is that of minimal clock disturbance. This corresponds to when we fix the tolerance error $\varepsilon_s$ to be independent of $d$, for all $\varepsilon_s>0$. In this case $V_0$ can be chosen to be $d$ independent from which it follows that (refer Eqs. \ref{eq:result 2 main text},\ref{eq:b zeta main text}) $b$ is constant and $\lim_{d\rightarrow \infty}\zeta=1$, meaning that the clock disturbance $\varepsilon_v$ is minimal and has a decay rate asymptotically equal to that given by the quasi-continuity bound $\varepsilon_c$ (See Eq. \eqref{eq:result 1 main text}). More generally, we show that one can even achieve a parametrisation such that $\lim_{d\rightarrow \infty}\varepsilon_s=0$ while still maintaining $\lim_{d\rightarrow \infty} \zeta=1$.

Alternatively, one can choose a potential with the aim of minimizing the r.h.s. of Eq. \eqref{eq:q-auto contrl trace bound}. This leads to both $\|\rho_s(t)-\rho_s'(t)\|_1$ and $\|\rho_c'(0)-\rho_c'(T_0)\|_1$ to be of the same order and to decay faster than any power of $d$,
specifically, exponential decay in $d^{1/4}\sqrt{\ln d}$.


\begin{figure}[!htb]
  \includegraphics[scale=0.32]
  {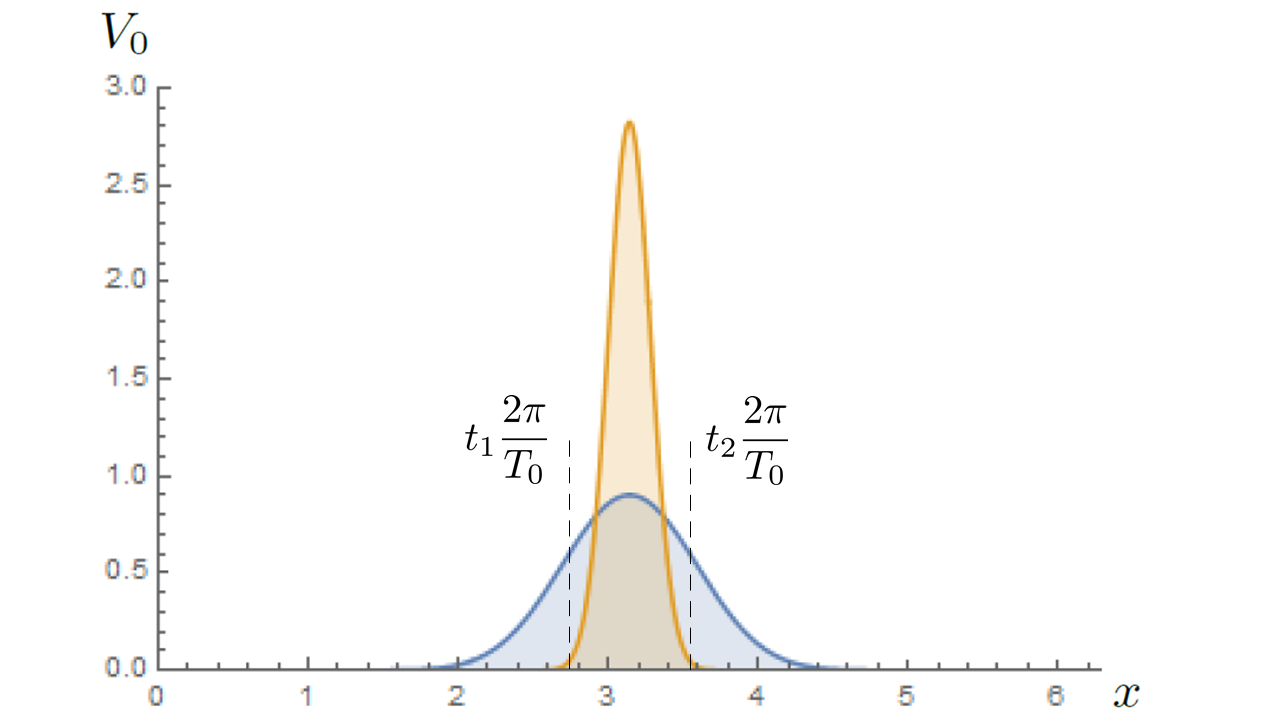}
  \caption{Plot of the potential $V_0(x)=A_c \cos^{2n}\left(\frac{x-\pi}{2} \right)$ over one period $[0,2\pi]$, with $A_c$ a normalisation constant, where $n=10$ (blue), and $n=100$ (orange). The blue potential is ``flatter" than the orange potential, and as such, $b$ in Eq. \eqref{eq:b zeta main text} is smaller than in the case of the orange potential. 
  The disturbance $\varepsilon_v$ in the clock's dynamics caused by the implementation of the unitary (see Eq. \eqref{eq:clock finite trace diff}), will be smaller for the blue potential than for the orange potential. However, the error term $\varepsilon_s$ involved in implementing the system unitary (see Eq. \eqref{eq:q-auto contrl trace bound}), will be larger for the blue potential than for the orange potential. Generally speaking, whenever $b$ is larger, $\varepsilon_s$ can be made smaller at the expense of a larger disturbance to the clock's dynamics (i.e. larger $\varepsilon_v$). This trade off is reminiscent of the information gain disturbance principle \cite{info_disturbace_Paul}, but here rather than gaining information, the unitary is implemented more accurately (i.e. smaller r.h.s. in Eq. \eqref{eq:q-auto contrl trace bound}).
  }\label{fig:tradeoff}
\end{figure}

\subsection{Results for \gClock~states of general width}\label{sec:generalwidth}

So far, the results stated have only been for when the initial complex Gaussian superposition of time states $\ket{\Psi_\textup{nor}(k_0)}$ has a width $\sigma$ equal to $\sqrt{d}$ and mean energy centered in the middle of the spectrum of $\hat H_c$. This particular choice of the width which we call \textit{symmetric}, has a particular physical significance.  The uncertainty in both the energy and time basis, denoted $\Delta E$ and $\Delta t$ respectively, are equal with $\Delta E\Delta t=1/2$.\footnote{up to an additive correction term which decays exponentially in $d$.\label{refnote}} The last equality is true regardless of the width $\sigma$, implying that all our \gClock~states are minimum uncertainty states.\footref{refnote} Whenever $\sigma >\sqrt{d}$, we have $\Delta E < \Delta t$ implying less uncertainty in energy. We call these \textit{energy squeezed} states in analogy with quantum optics terminology. Meanwhile, we call states  for which $\sigma <\sqrt{d}$, \textit{time squeezed} since $\Delta E > \Delta t$ in this case. It is likely that only time squeezed states can 
achieve the Heisenberg limit \cite{HesenbergLim}. Yet, from the viewpoint of quantum control, our results suggest that squeezed states (either in the time or energy basis) have larger errors $\varepsilon_c$ and $\varepsilon_v$, thus rendering them more fragile to back-reactions from the potential $V_0$. More precisely, whenever the initial clock state is energy or time squeezed, the error terms $\varepsilon_c$, $\varepsilon_v$ still maintain their linear scaling with time. However, w.r.t. the dimension of the clock, the error no longer decays exponentially in $d$, but rather in $d^\eta$, where $0<\eta<1$, and where $\eta$ decreases as the state is squeezed further in time or energy. (Full details are left to later sections.)

\section{Discussion and Outlook}\label{sec:discussion}

Our results have so far been framed in terms of the dimension of the clock. However, since the clock has a Hamiltonian of equally spaced energy levels, the energy of the clock is linear in its dimension. Thus all of our results are also statements on the efficacy of the clock w.r.t. the energy of the state. In particular, this implies that the back-reaction on the clock is exponentially small in its mean energy.\Mspace

The second question of interest is how does the mean energy of the initial state of the clock effect the errors induced by finite size? From the energy-time uncertainty relation \cite{time_energy_PBusch}, and indeed the idealised clock, one might expect the larger the mean energy of the state, the better it performs. Contrarily, we find that errors maintain the same exponential decay in $d^\eta$ with $\eta=1$ reserved for symmetric states, but now with a smaller prefactor | one has to replace the $\pi/4$ in Eq. \eqref{eq:result 1 main text} with a factor which approaches zero as the mean energy of the initial clock state approaches either end of the spectrum of $\hat H_c$. This suggests that, when the dimension is finite, it is the \textit{dimension} itself, rather than the energy of the state, as suggested by the energy-time uncertainty relation \cite{time_energy_PBusch}, which is the resource for improving the accuracy of a clock.\Mspace

Thus, an interesting follow-up question to consider is the case of infinite-dimensional Hamiltonians, such as a harmonic oscillator, and investigate the accuracy of the clock w.r.t. to its mean energy. In fact, an alternative interpretation of our results is that the clock Hamiltonian is that of a quantum harmonic oscillator with energy spacing $\omega$, and that we work only in a finite dimensional subspace. In this interpretation, in the error terms $\varepsilon_c$, $\varepsilon_v$, $\varepsilon_s$, one could make the substitution $d\mapsto \frac{2}{\omega} \langle \hat H_c \rangle$, where $\langle \hat H_c \rangle$ is the mean energy of the clock's initial state, $\ket{\Psi_\textup{nor}(k_0)}$. 
In this case, we can interpret the clock as living in an infinite dimensional Hilbert space and with error terms $\varepsilon_c$, $\varepsilon_v$ which are exponentially small in mean clock energy.\Mspace

Based on numerical evidence, we conjecture that the exponential decay of $\varepsilon_c$ in the clock dimension for the \gClock~states with $\sigma=\sqrt{d}$, is the best possible scaling with clock dimension which holds for all times in one time period $[0,T_0]$. Furthermore, it could be that no initial clock state can achieve a better scaling. Such a fundamental limitation, would have direct implications for how much work can be embezzled from the clock when it implements a unitary in the framework of \cite{second}. (See Section \ref{sec:conjectures} for a longer discussion.)\Mspace

In the main text we have assumed that the system has a trivial Hamiltonian. The results derived continue to hold in the case of the system having a non-zero $\hat{H}_s$, as long as it commutes with the unitary operation being applied, i.e. $[ \hat{H_s}, \tilde{U} ] = 0$, what is commonly referred to as an ``energy-preserving unitary" or ``covariant operation". The results also serve as an approximation for the non-commuting case, if the clock operates at a time-scale much quicker than that of the system (here the frequency of the clock $\omega$, which has so far not affected the results, would become important).\Mspace

For the general case of non-commuting unitaries, one would additionally need a source of energy and coherence in addition to the source of timing provided by the clock. The energy source can be modeled in our setup explicitly by including a quantum battery system $\mathcal{H}_b$ into the setup or taken to be the clock itself. In the former case, the clock would perform an energy preserving unitary over the $\mathcal{H}_s\otimes\mathcal{H}_b$ system. The battery Hamiltonian and initial state would be such that locally, on $\mathcal{H}_s$, an arbitrary unitary would have been performed. Such a battery Hamiltonian and initial battery state has been studied in \cite{Aberg}, and will be applied to the setup in this paper in an upcoming paper.

\section{Conclusions}\label{sec:conclusions}
We solve the dynamics of a finite dimensional clock when the initial state is a coherent complex Gaussian superposition of time states | a basis which is mutually unbiased with respect to the energy eigenstates of the finite dimensional Hamiltonian. We show that such superpositions evolve in time in ways which mimic idealised, infinite dimensional and energy clocks up to errors which decay exponentially fast in clock dimension. We demonstrate the consequences our  results have for autonomous quantum control. We show that the clock can implement a timed unitary on the system via a joint clock-system time independent Hamiltonian, with an error which decays faster than any polynomial in the clock dimension; or equivalently, faster than any polynomial in the clock's mean energy. The implementation of the unitary induces a back-reaction onto the clock's dynamics, which we prove is also exponentially small in the clock's dimension and energy. We discuss the trade-off between smaller clock disturbance and better temporal localization of the unitary's implementation, which our bounds address quantitatively. Our results single out states of equal uncertainty in time and energy, with a mean energy at the mid point of the energy spectrum, as being the most robust, and thus incurring the least disturbance due to its implementation of the unitary on the system.\Mspace

On another level, by demonstrating that up to small errors,  the three core paradigms in quantum control | the unitary (Eq. \eqref{eq:rho initial final}), the time dependent Hamiltonian (Eq. \eqref{eq:t dependent ham}), and time independent control (Eq. \eqref{eq:q-auto contrl trace bound}), are all equivalent up to small errors, our results represent a unification of the three paradigms. The implications of this are not only of a foundational nature, but also practical, since it implies that one can numerically simulate a time dependent Hamiltonian on $\mathcal{H}_s$, while being re-assured that in actual fact, the results of the simulation are equivalent to simulating a time independent Hamiltonian on the larger Hilbert space $\mathcal{H}_s\otimes\mathcal{H}_c$, which would be numerically intractable due to the increased dimension.\Mspace

Such results are important because the implementation of unitaries and time-dependent interactions is a ubiquitous operation encountered in almost every research field of theoretical quantum mechanics; perhaps the most prominent example being in the field of quantum computation. Yet very little was known about how well such control, normally accounted for by a classical field, can actually be implemented by a fully quantum autonomous setup. Our results provide analytical and insightful bounds on exactly how well this can be achieved, venturing into the fundamental issue of time in quantum mechanics in the process. We envisage that the techniques for solving this problem may be useful in solving other dynamical problems in many-body physics.

 \newpage
\onecolumngrid
\begin{center}
\textbf{{\large Autonomous quantum machines and finite sized clocks: detailed results and derivations}}
\end{center}

\onecolumngrid

In the remaining sections we provide the full details of our findings. The organization is as follows.

 Section \ref{sec:Definition of Gaussian clock} introduces the definitions of our clock states and discusses some of their basic properties.
Section \ref{sec:Continuity of the Gaussian clock state} is dedicated to the proof of our first main result. The main result is presented in Theorem \ref{gaussiancontinuity}. All other Lemmas in this section are technical lemmas used solely for the proof of Theorem \ref{gaussiancontinuity}. A sketch of the proof can be found at the beginning of the Section. If one wishes to understand the proof of the second main result, this section may be skipped, since the proof of the second result is a generalization of Theorem \ref{gaussiancontinuity}. Section \ref{sec:Proof of quasi-continuous control of Gaussian clock states} is concerned with the proof of our second main result. Readers solely interested in the result can go directly to Theorem \ref{movig through finite time}, followed by Corollary \ref{movig through finite time coro} and the example Section \ref{sec:Examples of Potential functions}. A sketch of the proof can be found at the beginning of the Section. The following section, \ref{sec:Clocks as Quantum control} is concerned with the immediate consequences that main result two has for autonomous quantum control. Unlike the with the previous two sections, a reader who wishes to obtain a deeper understanding of the results (not necessarily the proofs), is advised to read the entire section, omitting the proofs if desired. The highlights of this section are 
Lemma \ref{lem:trace dist bound t-Ham Vs d dim clock}, Corollary \ref{corr:explicit bound on system trace distance}, and the examples in Section \ref{sec: control exmaples}.

Section \ref{appendixcommutator}, is concerned with the proof of main result three. Theorem \ref{Quasi-Canonical commutation} represents the main result of this section.

Section \ref{sec:conjectures} proposes some conjectures, based on numerical studies about the tightness and generality of our bounds. Some open questions about the properties of the bounds are also discussed.

The remaining sections proved background information and technical results and definitions used throughout the proofs. Section \ref{idealizedclock} is concerned with describing the idealised momentum clock. It serves as a reference to the idealised properties we wish our finite dimension clock to mimic. Section \ref{SWPclock} explains previous results in the literature on finite clocks while pointing out their shortcomings which our clock will overcome. It also introduces some of the definitions which will be used  in the rest of the manuscript.

Sections \ref{mathidentities} and \ref{Error Bounds} are for reference, and do not contain any of the main results or this article. The first of these two Sections, Section \ref{mathidentities}, contains some of the essential mathematical ingredients which have been used repetitively throughout previous Sections. Here, Sections, \ref{Poisson summation formula} and \ref
{unitaryerroraddition} are simple yet crucially important for the main results of this article. The second of these two sections, Section \ref{Error Bounds}, contains error bounds for summations over Gaussian tails. 


\onecolumngrid

\section{Definition of \gClock~states and properties}\label{sec:Definition of Gaussian clock}
In this section we will introduce the class of \gClock~states, that are complex Gaussian superpositions of the time-states, and review some of their properties. In the following, we call the Hilbert space of the clock, $\mathcal{H}_c$, the Hilbert space formed by the span of the time basis $\{\ket{\theta_k}\}_{k=0}^{d-1}$, or equivalently the energy basis $\{\ket{E_n}\}_{n=0}^{d-1}$.

\begin{definition}\label{def:Gaussian clock states}\emph{(\gClock~states)}. Let $\Lambda_{\sigma,n_0}$ be the following space of states in the Hilbert space of the $d$ dimensional clock,
\be\label{eq:lambda set def} 
\Lambda_{\sigma,n_0}=\bigg\{ \ket{\Psi(k_0)}\in\mathcal{H}_{c},\quad  k_0\in\rr\bigg\},
\ee
where
\begin{align}\label{gaussianclock}
	\ket{\Psi(k_0)} = \sum_{\mathclap{\substack{k\in \mathcal{S}_d(k_0)}}} A e^{-\frac{\pi}{\sigma^2}(k-k_0)^2} e^{i 2\pi n_0(k-k_0)/d} \ket{\theta_k},
\end{align}
with $\sigma \in (0,d)$, $n_0 \in (0,d-1)$, $A \in \rr^+$, and $\mathcal{S}_d(k_0)$ is the set of $d$ integers closest to $k_0$, defined as:
\begin{align}\label{eq: mathcal S def}
	\mathcal{S}_d(k_0) = \left\{ k \; : \; k\in \mathbb{Z} \text{   and  }  -\frac{d}{2} \leq  k_0-k < \frac{d}{2} \right\}.
\end{align}
In the special case that $\ket{\Psi(k_0)}$ is normalized, it will be denoted by
\be\label{eq:Psi ket normalised}
\ket{\Psi_\textup{nor}(k_0)}=\ket{\Psi(k_0)},
\ee
and $A$ will take on the specific value
\be\label{eq:A normalised}
A=A(\sigma;k_0)=\frac{1}{\sqrt{\sum_{k\in \mathcal{S}_d(k_0)} \me^{-\frac{2\pi}{\sigma^2}(k-k_0)^2}}},
\ee
s.t. $\braket{\Psi_\textup{nor}(k_0)|\Psi_\textup{nor}(k_0)}=1$. Bounds for $A$ can be found in Section \ref{Normalizing the clock state}. 
\end{definition}
\begin{remark}[Technicality]
Sometimes we will use big O notation $\bo$, and $poly(x)$ to denote a generic polynomial in $x\in\rr$ of constant degree. When doing so, for simplicity, it will be assumed that $\sigma\in(0,d)$ becomes infinitely far away from the end points of its domain in the large $d$ limit, namely  $\lim_{d\rightarrow\infty}\sigma=+\infty$ and $\lim_{d\rightarrow\infty}\sigma /d =0$.
\end{remark}
\begin{remark}[Time, uncertainty, and energy of clock states]

The parameters $k_0$ and $\sigma$ may be identified with the mean and variance respectively in the basis of time-states. $n_0$ may be identified with the mean energy of the clock, as we shall discuss shortly in Remark \ref{rem:symme of clock} which motivates Def. \ref{def:clock stat classes}.

\end{remark}
\begin{remark}[Keeping the clock state centered]\label{clockcentered}

The choice of the set $\mathcal{S}_d(k_0)$ is to ensure that the Gaussian is always centered in the chosen basis of angle states. The reason it is possible to do this in the first place is because the basis of clock states is invariant under a translation by $d$, i.e. $\ket{\theta_k} = \ket{\theta_{k+d}}$ due to Eq.  \eqref{finitetimestates_main}. Instead of the set $\{0,1,...,d-1\}$ we can choose to express the state w.r.t. to any set of $d$ consecutive integers, and we choose the specific set in which the state above is centered.
If $d$ is even, then $\mathcal{S}_d(k_0)$ changes at integer $k_0$, if $d$ is odd, then it changes at half-integer values of $k_0$. Equivalently, in terms of the floor function $\lfloor \cdot \rfloor$,
\begin{equation}\label{eq:S set def}
	\mathcal{S}_d(k_0) = \begin{cases}
		\mathcal{S}_d(\lfloor k_0 \rfloor ) & \text{if $d$ is even}, \\
		\mathcal{S}_d(\lfloor k_0 + 1/2 \rfloor) & \text{if $d$ is odd}.
			\end{cases}
\end{equation}
\end{remark}
\begin{definition}\label{Distance of the mean energy from the edge}(Distance of the mean energy from the edge of the spectrum)
We define the parameter $\alpha_0\in(0,1]$ as a measure of how close $n_0\in(0,d-1)$ is to the edge of the energy spectrum,
namely
\begin{align}\label{eq:alpha_0 def}
\alpha_0&=\left(\frac{2}{d-1}\right) \min\{n_0,(d-1)-n_0\}\\
&=1-\left|1-n_0\,\left(\frac{2}{d-1}\right)\right|\in(0,1].
\end{align}
The maximum value $\alpha_0=1$ is obtained for $n_0=(d-1)/2$ when the mean energy is at the mid point of the energy spectrum, while $\alpha_0\rightarrow 0$ as $n_0$ approaches the edge values $0$ or $d-1$. 
c.f. similar measure, Eq. \eqref{eq:alpha_c def eq}.
\end{definition}

\begin{remark}[Comparison to time-states]

As a first comparison, we repeat the analysis of the behaviour of the time-states (see Fig. \ref{peresbehaviour}), this time for the \gClock~state, i.e. we plot the expectation value and variance of the time operator, given the \gClock~initial state centered about $\ket{\theta_0}$, see Fig. \ref{gaussbehaviour}.

\begin{figure}[h]
\includegraphics[width=0.4\linewidth]{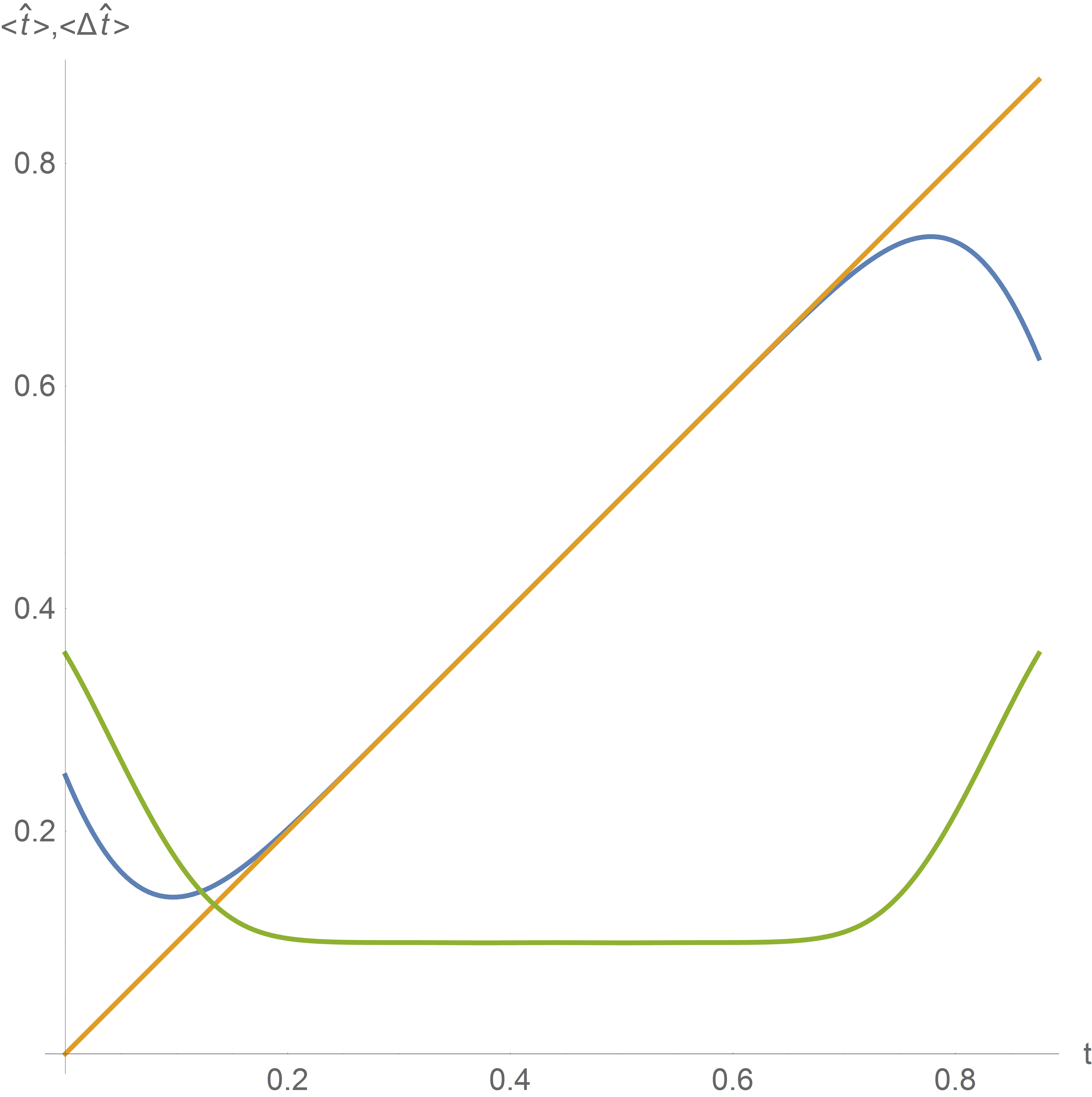}
\caption{The expectation value (blue) and variance (green) of the time operator Eq. \eqref{eq:t op main text def} 
	 given the initial state $\ket{\Psi(0)} \in \Lambda_{\sqrt{d},(d-1)/2}$, (referred to as symmetric, see Def. \ref{def:clock stat classes} and Remark \ref{rem:symme of clock}) for $d=8$, $T_0 = 1$. The ideal case $\braket{\hat{t}} = t$ is in orange. If one compares this to the analogous plot for time states rather than the state used here (See Fig. \ref{peresbehaviour}, Section \ref{Shortcomings of the time-states} in the appendix), it is apparent that time states produce a mean and standard deviation which is much worse.}\label{gaussbehaviour}
\end{figure}
\end{remark}

\begin{definition}\label{def:analytic ext. of clock states} \emph{(Analytic extension of clock states).} Corresponding to every $\ket{\Psi(k_0)} \in \Lambda_{\sigma,n_0}$, we define $\psi : \rr \rightarrow \mathbb{C}$ to be the analytic Gaussian function
\begin{equation}\label{analyticposition}
	\psi(k_0;x) = A e^{-\frac{\pi}{\sigma^2}(x-k_0)^2} e^{i2\pi n_0(x-k_0)/d}.
\end{equation}

By definition, $\psi(k_0;k) = \braket{\theta_k| \Psi(k_0)}$ for $k \in \mathcal{S}_d(k_0)$, and thus $\psi$ is an analytic extension of the discrete coefficients  $\braket{\theta_k|\Psi(k_0)}$.

In the special case that the corresponding state $\ket{\Psi(k_0)} \in \Lambda_{\sigma,n_0}$ is normalised, $\psi$ will be denoted accordingly, namely
\be\label{eq:psi nor def}
\psi_\textup{nor}(k_0;k):=\psi(k_0;k)\quad \text{iff } A \text{ satisfies Eq. \eqref{eq:A normalised}}
\ee
\end{definition}

\begin{remark}\label{interchangearguments} $\psi(k_0; k+y) = \psi(k_0-y; k)$.
\end{remark}

\begin{definition}\label{Continuous Fourier Transform no pot} \emph{(Continuous Fourier Transform as a function of dimension $d$ of clock state).} Let $\tilde{\psi} : \rr \rightarrow \mathbb{C}$ be defined as the continuous Fourier transform of $\psi$,
\begin{align}\label{cftd}
	\tilde{\psi}(k_0;p) &= \frac{1}{\sqrt{d}} \int_{-\infty}^\infty \psi(k_0;x) \me^{-\mi 2\pi px/d} dx = A \frac{\sigma}{\sqrt{d}} \me^{-\frac{\pi \sigma^2}{d^2} (p-n_0)^2} \me^{-i 2\pi n k_0/d}.
\end{align}
Similarly to above, we denote
\be\label{eq:psi f.t. nor def}
\tilde\psi_\textup{nor}(k_0;p):=\tilde\psi(k_0;p)\quad \text{iff } A \text{ satisfies Eq. \eqref{eq:A normalised}}.
\ee
\end{definition}

\begin{lemma}[The clock state in the energy basis] 
This Lemma states that the continuous Fourier transform $\tilde{\psi}(k_0;p)$ is an exponentially good approximation (w.r.t. dimension $d$) of the clock state in the energy basis. Mathematically this is a statement of the closeness of the discrete Fourier transform (D.F.T.) to the continuous Fourier transform (C.F.T.) for \gClock~states. For simplicity, we state the result for the special case $\sigma=\sqrt{d}$, $n_0=(d-1)/2$.
\begin{equation}
	\abs{ \braket{E_n | \Psi_\textup{nor}(k_0)} - \tilde{\psi}_\textup{nor}(k_0;n) } < \left( \frac{2^{\frac{9}{4}}}{1 - e^{-\pi}} \right) d^{-\frac{1}{4}} e^{-\frac{\pi d}{4}}, \quad n=0,1,2,\ldots, d-1.
\end{equation}
\end{lemma}

\begin{proof}

From the definition of the state, Eq. \eqref{def:Gaussian clock states}, and the relation between the time states and energy states Eq.  \eqref{finitetimestates_main},
\begin{align}
	\braket{E_n | \Psi(k_0)}	&= \frac{1}{\sqrt{d}} \;\;\;  \sum_{k \in \mathcal{S}_d(k_0)} \psi(k_0;k) e^{-i 2\pi nk/d}.
\end{align}

By definition, this is the D.F.T. of the $\psi(k_0;k)$, state in the time basis. To prove that this approximates the C.F.T. of the state $\psi(k_0,k)$, we first convert the finite sum above to the infinite sum $k\in\mathbb{Z}$, and bound the difference using Lemma \ref{G0},
\begin{equation}
	\left| \frac{1}{\sqrt{d}} \;\;\;  \sum_{k \in \mathcal{S}(d,k_0)} \psi_\textup{nor}(k_0;k)\, e^{-i 2\pi nk/d} - \frac{1}{\sqrt{d}} \;\;\;  \sum_{\mathclap{\substack{k \in \mathbb{Z}}}} \psi_\textup{nor}(k_0;k)\, e^{-i 2\pi nk/d} \right| < \left( \frac{2^{\frac{5}{4}}}{1 - e^{-\pi}} \right) d^{-\frac{3}{4}} e^{-\frac{\pi d}{4}},
\end{equation}
where we have used the bound for $A$ for normalized states derived in Section \ref{Normalizing the clock state}.
Applying the Poisson summation formula (Corollary \ref{poissonsummation}) on the infinite sum,
\begin{equation}
	\frac{1}{\sqrt{d}} \;\;\;  \sum_{k \in \mathbb{Z}} \psi_\textup{nor}(k_0;k) e^{-i 2\pi nk/d} =  \sum_{m \in \mathbb{Z}} \tilde{\psi}_\textup{nor}(k_0;n + md).
\end{equation}

Using Lemma \ref{G0}, we can approximate the sum by the $m=0$ term,
\begin{equation}
	\left| \tilde{\psi}_\textup{nor}(k_0;n) - \sum_{m \in \mathbb{Z}} \tilde{\psi}_\textup{nor}(k_0;n + md) \right| < \left( \frac{2^{\frac{5}{4}}}{1 - e^{-\pi}} \right) d^{-\frac{1}{4}} e^{-\frac{\pi d}{4}}.
\end{equation}

Adding the two error terms, we arrive at the Lemma statement.
\end{proof}

\begin{remark}[Symmetry of clock states]\label{rem:symme of clock}
Now that we know the C.F.T. $\psi$ to be a good approximation of the clock state in the energy basis, we can note the two most important features. First, the width of the state in the energy basis is of the same order ($\sqrt{d}$) as that in the time basis, Eq. \eqref{analyticposition}. The states are thus symmetric w.r.t. the uncertainty in either basis. In general when $\sigma\neq \sqrt{d}$, the widths of the states in the time and energy bases are of different order.

Second, the state $\ket{\Psi_\textup{nor}(k_0)}$ is centered about the value $n_0=(d-1)/2$. Considering that the range of values for the energy number is from $0$ to $d-1$, the energy of the clock state is exactly at the middle of the spectrum. These observations motivate the following definitions.
\end{remark}

\begin{definition}(\gClock~state classes)\label{def:clock stat classes}
We give the following names to states $\ket{\Psi(k_0)}\in\Lambda_{\sigma,n_0}$ depending on the relationship between their width $\sigma\in(0,d)$, and the clock dimension $d$:
\begin{itemize}
\item[1)] \textup{symmetric} if $\sigma^2=d$.\\\vspace{-0.4cm}
\item[2)] \textup{time squeezed} if $\sigma^2<d$.\\\vspace{-0.4cm}
\item[3)] \textup{energy squeezed} if $\sigma^2>d$.
\end{itemize}
Furthermore, we add the adverb, \textup{completely} when in addition $n_0=(d-1)/2$, i.e.  $\ket{\Psi(k_0)}\in\Lambda_{\sigma,n_0}$ has mean energy centered at the middle of the spectrum of $\hat H_c$.
\end{definition}
\section{Proof of quasi-continuity of \gClock~states}
\label{sec:Continuity of the Gaussian clock state}

Recall the \emph{continuity} of the idealized clock, Eq. \eqref{idealregularity_main}. The time-translated state of the clock was the same as the space-translated state, for arbitrary translation size. Our first major result is to derive an analogous statement for the finite clock. This is the subject of this section, with the Theorem stated at the end (Theorem \ref{gaussiancontinuity}). 

{\bf Sketch of proof:}

The mean position $k_0$ of the clock state appears twice in the expression for the state, first as the mean of the Gaussian $\psi(k_0;k)$, and second as determining the set of integers $k$ over which the time states $\ket{\theta_k}$ are defined.

Therefore, in order to prove that time translations are (approximately) equivalent to translations in position, we prove that an arbitrary time translation of the clock by a value $\Delta$ shifts both $k_0$ to $k_0+\Delta$ and the set $\mathcal{S}(k_0)$ to $\mathcal{S}(k_0+\Delta)$.

By the properties of $\mathcal{S}(k_0)$, we see that it changes at definite values of $k_0$ (integer values if $d$ is even, and half-integer values if $d$ is odd).

We proceed in a number of steps. First we prove that an infinitesimal time translation by $\delta$ is approximately equivalent to shifting only the mean $k_0$ by $\delta$. We then use the Lie-Product formula to extend the statement to finite time-translations of $\Delta$ that are small enough so that $\mathcal{S}(k_0+\Delta)$ is the same as $\mathcal{S}(k_0)$.

We then bound the error involved in switching the set $\mathcal{S}(k_0)$ to $\mathcal{S}(k_0+1)$ (without any other change in the state).

Together the two results then provide a manner of calculating the state for arbitrary time translations: first move only the mean $k_0$ for one unit, then switch $\mathcal{S}(k_0)$ by one integer, and repeat.

We proceed by deriving all the necessary technical lemmas which are necessary for Theorem \ref{gaussiancontinuity}. We first prove the quasi-continuity of the \gClock~states under infinitesimal time-translations, followed by extending the proof to arbitrary times.

\begin{lemma}\label{infinitesimaltimetranslation}
The action of the clock Hamiltonian on a \gClock~state for infinitesimal time of order $\delta$ may be approximated by an infinitesimal translation of the same order on the analytic extension of the clock state. Precisely speaking, for $\ket{\Psi(k_0)} \in \Lambda_{\sigma,n_0}$,
\begin{align}\label{infshift}
	e^{-i \delta \frac{T_0}{d} \hat{H}_c} \ket{\Psi(k_0)}  = e^{-i \delta \frac{T_0}{d} \hat{H}_c} \sum_{k\in\mathcal{S}_d(k_0)} \psi(k_0; k) \ket{\theta_k}  = \sum_{k\in\mathcal{S}_d(k_0)} \psi (k_0+\delta; k) \ket{\theta_k} + \ket{\epsilon},
\end{align}
where the $l_2$ norm of the error $\ket{\epsilon}$ is bounded by 
\begin{align}
	\ltwo{\ket{\epsilon}} &\leq  \delta \epsilon_{total} + \delta^2 C^\prime, \quad \text{where $C^\prime$ is $\delta$ independent, and } \\
	\frac{\epsilon_{total}}{2\pi Ad} &<
	\begin{cases}\label{continuityerror}\displaystyle
 \left( 2 \sqrt{d} \left( \frac{1}{2} + \frac{1}{2\pi d} + \frac{1}{1 - e^{-\pi}} \right) e^{-\frac{\pi d}{4}} + \frac{2}{1 - e^{-\pi}} + \frac{1}{2} + \frac{1}{2\pi d} \right) e^{-\frac{\pi d}{4}} &\mbox{if } \sigma=\sqrt{d} \\[10pt]\displaystyle
 \left( 2 \sigma \left( \frac{\alpha_0}{2} + \frac{1}{2\pi\sigma^2} + \frac{1}{1 - e^{-\pi\sigma^2\alpha_0}} \right) e^{-\frac{\pi\sigma^2}{4}\alpha_0^2} + \left( \frac{1}{1 - e^{-\frac{\pi d}{\sigma^2}}} + \frac{1}{1 - e^{-\frac{\pi d^2}{\sigma^2}}} + \frac{d}{2\sigma^2} + \frac{1}{2\pi d} \right) e^{-\frac{\pi d^2}{4\sigma^2}} \right) &\mbox{otherwise}
\end{cases}
\end{align}
\end{lemma}
\begin{remark}
	Before we begin the proof, note that the R.H.S. of Eq. \eqref{infshift} is not necessarily $\ket{\Psi (k_0+\delta)}$, since the summation is still over the set of integers $\mathcal{S}_d(k_0)$, and not $\mathcal{S}_d(k_0+\delta)$ which may be different, see Eq. \eqref{gaussianclock} and the remark \ref{clockcentered} that follows.
\end{remark}
\begin{proof}
\begin{align}
	e^{-i \delta \frac{T_0}{d} \hat{H}_c} \ket{\Psi(k_0)} &= e^{-i \delta \frac{T_0}{d} \sum_{m=0}^{d-1} \frac{2\pi}{T_0} m \ket{E_m}\bra{E_m}} \sum_{k\in \mathcal{S}_d(k_0)} \psi(k_0;k) \ket{\theta_k}.
\end{align}

Switching to the basis of energy states so as to apply the Hamiltonian, and back to the basis of time states, via Eq. \eqref{finitetimestates_main},
\begin{align}
	e^{-i \delta \frac{T_0}{d} \hat{H}_c} \ket{\Psi(k_0)} &= \sum_{k,l\in\mathcal{S}_d(k_0)} \psi(k_0,k) \left( \frac{1}{d} \sum_{n=0}^{d-1} e^{-i2\pi n(k+\delta-l)/d} \right) \ket{\theta_l}
\end{align}

We label the coefficient in the time basis of the exact state and the approximation as
\begin{align}
	c_l(\delta) &= \braket{\theta_l | e^{-i \delta \frac{T_0}{d} \hat{H}_c} | \Psi(k_0)} = \sum_{k\in\mathcal{S}_d(k_0)} \psi(k_0,k) \left( \frac{1}{d} \sum_{n=0}^{d-1} e^{-i2\pi n(k+ \delta-l)/d} \right), \\
	c^\prime_l(\delta) &= \bra{\theta_l} \sum_{k\in\mathcal{S}_d(k_0)} \psi (k_0+\delta; k) \ket{\theta_k}=  \psi\left( k_0 + \delta; l \right). \label{idealcoeff}
\end{align}

By the properties of the analytic extension $\psi$, both of the above are analytic functions w.r.t. $\delta$, and we can express the difference between the coefficients via the Taylor series expansion about $\delta=0$,
\begin{align}\label{firstTaylor}
	c_l(\delta) - c^\prime_l(\delta) &= c_l(0) - c^\prime_l(0) + \delta \left[ \frac{\partial c_l(\delta)}{\partial \delta} - \frac{\partial c^\prime_l(\delta)}{\partial \delta} \right]\bigg|_{\delta=0} + C \delta^2, \\
	\text{where   } C & = \frac{1}{2} \left(  \text{max}_{|t|\leq \delta} \left|  \frac{\partial^2 c_l(t)}{\partial t^2} - \frac{\partial^2 c^\prime_l(t)}{\partial t^2} \right| \right).
\end{align}

$C$ is upper bounded by a $\delta$ independent constant  because the second derivatives w.r.t. $\delta$ of both $c_l(\delta)$ and $c^\prime_l(\delta)$ can be verified to be finite sums of bounded quantities, and therefore bounded.

We now simplify Eq. \eqref{firstTaylor}. By direct substitution, $c_l(0) = c^\prime_l(0)$. For the first derivatives,
\begin{align}
	\frac{\partial c_l(\delta)}{\partial \delta} \bigg|_{\delta=0} &= \left[ \frac{\partial}{\partial \delta} \sum_{k \in \mathcal{S}_d(k_0)} \psi(k_0;k) \left( \frac{1}{d} \sum_{n=0}^{d-1} e^{-i2\pi n(k+\delta-l)/d} \right) \right]_{\delta=0} \\
	&= \left( \frac{-i 2\pi}{d^2} \right) \sum_{k \in \mathcal{S}_d(k_0)} \psi(k_0;k) \sum_{n=0}^{d-1} n e^{-i2\pi n(k-l)/d}
\end{align}

One can replace the finite sum over $k$ by an infinite sum, and bound the difference via Lemma \ref{G0},
\begin{align}
	\frac{\partial c_l(\delta)}{\partial \delta} \bigg|_{\delta=0} &= \left( \frac{-i 2\pi}{d^2} \right) \sum_{k\in\mathbb{Z}} \psi(k_0;k) e^{-i 2\pi nk/d} \sum_{n=0}^{d-1} n e^{i2\pi nl/d} + \epsilon_1, \\
	\text{where} \;\;\;\;\; \left| \epsilon_1 \right| &< 
	\begin{cases}\displaystyle
  2\pi A \frac{e^{-\frac{\pi d}{4}}}{1 - e^{-\pi}} &\mbox{if } \sigma=\sqrt{d} \\[5pt]\displaystyle
 2\pi A \frac{e^{-\frac{\pi d^2}{4\sigma^2}}}{1 - e^{-\frac{\pi d}{\sigma^2}}}  &\mbox{otherwise}
\end{cases}\label{eq:epsilon1 for continuity lema 1}
\end{align}

Applying the Poisson summation formula \ref{poissonsummation} upon the infinite sum,
\begin{align}
	\frac{\partial c_l(\delta)}{\partial \delta} \bigg|_{\delta=0} &= \left( \frac{-i 2\pi}{d\sqrt{d}} \right) \sum_{n=0}^{d-1} \sum_{m\in\mathbb{Z}} \tilde{\psi}(k_0;n+md) \; n e^{i2\pi nl/d} + \epsilon_1
\end{align}

Since $\sum_{n=0}^{d-1} \sum_{m\in\mathbb{Z}} f(n+md) = \sum_{s\in\mathbb{Z}} f(s)$, after some arithmetic manipulation,
\begin{align}
	\frac{\partial c_l(\delta)}{\partial \delta} \bigg|_{\delta=0} &= \left( \frac{-i 2\pi}{d\sqrt{d}} \right) \left( \sum_{s\in\mathbb{Z}} \tilde{\psi}(k_0;s) \; s e^{i2\pi sl/d} - \sum_{n=0}^{d-1} \sum_{m=-\infty}^{\infty} \tilde{\psi}(k_0;n+md) \;md  e^{i2\pi nl/d} \right) + \epsilon_1
\end{align}

The second summation is a small contribution and may be bound in a similar way to $\epsilon_1$, via Lemmas \ref{G0}-\ref{G1}. Therefore,
\begin{align}
	\frac{\partial c_l(\delta)}{\partial \delta} \bigg|_{\delta=0} &= \left( \frac{-i 2\pi}{d\sqrt{d}} \right) \sum_{s\in\mathbb{Z}} \tilde{\psi}(k_0;s) \; s e^{i2\pi sl/d} + \epsilon_2 + \epsilon_1, \\
	\text{where} \;\;\;\;\;  & \label{eq:epsilon2 for continuity lema 1}\\
	\abs{\epsilon_2} &< 
\begin{cases}\displaystyle
 4\pi A \sqrt{d} \left( \frac{1}{2} + \frac{1}{2\pi d} + \frac{1}{1 - e^{-\pi d}} \right) e^{-\frac{\pi d}{4}} &\mbox{if } \sigma=\sqrt{d} \\[10pt]\displaystyle
 4\pi A \sigma \left( \frac{\alpha_0}{2} + \frac{1}{2\pi\sigma^2} + \frac{1}{1 - e^{-\pi\sigma^2\alpha_0}} \right) e^{-\frac{\pi\sigma^2}{4}\alpha_0^2}  &\mbox{otherwise}
\end{cases}
\end{align}

Apply the Poisson summation on the remaining sum, Corollaries \ref{poissonsummation},
\begin{align}
	\frac{\partial c_l(\delta)}{\partial \delta} \bigg|_{\delta=0} &= \left( \frac{-i 2\pi}{d\sqrt{d}} \right) \sqrt{d} \left( \frac{-id}{2\pi} \right) \sum_{m\in\mathbb{Z}} \frac{\partial}{\partial x} \psi \left( k_0; x + l \right) \bigg|_{x=md} + \epsilon_2 + \epsilon_1
\end{align}

Substituting $y=-x$, and using $\psi(k_0;l-y) = \psi(k_0+y,l)$, \ref{interchangearguments},
\begin{align}
	\frac{\partial c_l(\delta)}{\partial \delta} &= \sum_{m\in\mathbb{Z}} \frac{\partial}{\partial y} \psi (k_0+y; l) \bigg|_{y=md} +\epsilon_2 + \epsilon_1.
\end{align}

The $m=0$ term in the sum above is exactly the derivative $\partial c_l^\prime(\delta)/\partial \delta$ given by Eq. \eqref{idealcoeff}. We may bound the remainder of the sum using Lemmas \ref{G0}-\ref{G1},
\begin{align}
	\frac{\partial c_l(\delta)}{\partial \delta} &= \frac{\partial c_l^\prime(\delta)}{\partial \delta} + \epsilon_3 + \epsilon_2 + \epsilon_1 \\
	\text{where} \;\;\;\;\; \left| \epsilon_3 \right| &< 
\begin{cases}\displaystyle
 4\pi A \left( \frac{1}{2} + \frac{1}{2\pi d} + \frac{1}{1 - e^{-\pi}} \right) e^{-\frac{\pi d}{4}}  &\mbox{if } \sigma=\sqrt{d} \\[10pt]\displaystyle
 4\pi A \left( \frac{d}{2\sigma^2} + \frac{1}{2\pi d} + \frac{1}{1 - e^{-\frac{\pi d^2}{\sigma^2}}} \right) e^{-\frac{\pi d^2}{4\sigma^2}} &\mbox{otherwise}
\end{cases}
	 \label{eq:epsilon3 for continuity lema 1}
\end{align}

The Taylor series expansion Eq. \eqref{firstTaylor} thus ends up as
\begin{align}
	c_l(\delta) - c^\prime_l(\delta) &= \delta ( \epsilon_1 + \epsilon_2 + \epsilon_3 ) + C \delta^2,
\end{align}
where using Eqs. \eqref{eq:epsilon1 for continuity lema 1},\eqref{eq:epsilon2 for continuity lema 1},\eqref{eq:epsilon3 for continuity lema 1}, the sum of the errors is bounded by $\abs{\epsilon_1 + \epsilon_2 + \epsilon_3} \leq \epsilon_{4}$,
\begin{align}
	\epsilon_{4}<
	\begin{cases}\displaystyle
 2\pi A \left( 2 \sqrt{d} \left( \frac{1}{2} + \frac{1}{2\pi d} + \frac{1}{1 - e^{-\pi}} \right) e^{-\frac{\pi d}{4}} + \frac{2}{1 - e^{-\pi}} + \frac{1}{2} + \frac{1}{2\pi d} \right) e^{-\frac{\pi d}{4}} &\mbox{if } \sigma=\sqrt{d} \\[10pt]\displaystyle
 2\pi A \Bigg( 2 \sigma \left( \frac{\alpha_0}{2} + \frac{1}{2\pi\sigma^2} + \frac{1}{1 - e^{-\pi\sigma^2\alpha_0}} \right) e^{-\frac{\pi\sigma^2}{4}\alpha_0^2} \\ \quad\quad\quad\quad\quad\quad\quad\quad\quad\quad\quad\quad\quad+ \left( \frac{1}{1 - e^{-\frac{\pi d}{\sigma^2}}} + \frac{1}{1 - e^{-\frac{\pi d^2}{\sigma^2}}} + \frac{d}{2\sigma^2} + \frac{1}{2\pi d} \right) e^{-\frac{\pi d^2}{4\sigma^2}} \Bigg) &\mbox{otherwise}
\end{cases}
\end{align}

Note that $\epsilon_{total}$ is independent of the index $l$. Reconstructing the ideal and approximate states from the coefficients $c_l$ and $c_l^\prime$, we arrive at the error in the state,
\begin{align}
	\ltwo{\ket{\epsilon}} \leq \sum_{l\in\mathcal{S}_d(k_0)} \abs{c_l(\delta) - c_l^\prime(\delta)} \leq \delta d \epsilon_{4} + d C \delta^2.
\end{align}
\end{proof}

\begin{lemma}[Finite time translations within a single time step]\label{smalltimetranslation}

Given an initial mean position $k_0$ and a translation of $\Delta$, s.t. $\mathcal{S}_d(k_0) = \mathcal{S}_d(k_0+\Delta)$, then
\begin{equation}
	e^{-i \frac{T_0}{d} \Delta \hat{H}_c} \ket{\Psi(k_0)} = \ket{\Psi(k_0+\Delta)} + \ket{\epsilon}, \text{   where   } \ltwo{\ket{\epsilon}} < \Delta \epsilon_{total},
\end{equation}
where $\epsilon_{total}$ is defined in Eq. \eqref{continuityerror}.
\end{lemma}

\begin{proof}
Split the translation $\Delta$ into $M$ equal steps. Then from Lemma \ref{infinitesimaltimetranslation}, for $n\in\{0,1,...,M-1\}$,
\begin{equation}
	e^{-i \frac{T_0}{d} \frac{\Delta}{M} \hat{H}_c} \ket{\Psi(k_0 + n \Delta/M)} = \ket{\Psi(k_0 + (n+1) \Delta/M)} + \ket{\epsilon}, \text{   where   } \ltwo{\ket{\epsilon}} < \frac{\Delta}{M} \epsilon_{total} + C^\prime \frac{\Delta^2}{M^2}.
\end{equation}

Applying Lemma \ref{unitaryerroraddition}, when we add all the steps, we obtain,
\be
	\ltwo{ \ket{\Psi(k_0+\Delta)} - \left( e^{-i \frac{T_0}{d} \frac{\Delta}{M} \hat{H}_c} \right)^M \ket{\Psi(k_0)} } < M \left( \frac{\Delta}{M} \epsilon_{total} + C^\prime \frac{\Delta^2}{M^2} \right).
\ee
Therefore, 
\be 
\ltwo{ \ket{\Psi(k_0+\Delta)} - e^{-i \frac{T_0}{d} \Delta \hat{H}_c} \ket{\Psi(k_0)} } < \Delta \epsilon_{total} + C^\prime \frac{\Delta^2}{M}.
\ee

We are free to choose any positive integer $M$, so we take the limit $M\rightarrow\infty$, which recovers the statement of the Lemma.
\end{proof}

At this point, we have already proven the continuity of the clock state for time translations that are finite, but small (i.e. small enough that the range $\mathcal{S}_d(k_0)$ remains the same). In order to generalize the statement to arbitrary translations we need to be able to shift the range itself, which is the goal of the following Lemma.

\begin{lemma}[Shifting the range of the clock state]\label{rangeshift}

If $d$ is even, and the mean of the clock state $k_0$ is an integer, or alternatively, if $d$ is $odd$ and $k_0$ is a half integer, then
\begin{align}
	\ltwo{ \sum_{k \in \mathcal{S}_d(k_0-1)} \psi(k_0;k) \ket{\theta_k} - \sum_{k \in \mathcal{S}_d(k_0)} \psi(k_0;k) \ket{\theta_k} } &< \epsilon_{step}, \\
	\text{where} \quad \abs{\epsilon_{step}} &<  
	\begin{cases}\displaystyle
 2 A e^{-\frac{ \pi d}{4}}  &\mbox{if } \sigma=\sqrt{d} \\[10pt]\displaystyle
 2 A e^{-\frac{ \pi d^2}{4\sigma^2}}  &\mbox{otherwise}
\end{cases}
\end{align}
\end{lemma}

\begin{proof}
We prove the statement for even $d$, the proof for odd $d$ is analogous.

By definition (see Eq. \eqref{gaussianclock}), $\mathcal{S}_d(k_0)$ is a set of $d$ consecutive integers. Thus the only difference between $\mathcal{S}_d(k_0)$ and $\mathcal{S}_d(k_0-1)$ is the leftmost integer of $\mathcal{S}_d(k_0-1)$ and the rightmost integer of $\mathcal{S}_d(k_0)$, which differ by precisely $d$. By direct calculation, these correspond to the integers $k_0-d/2$ and $k_0+d/2$. These are the only two terms that do not cancel out in the statement of the Lemma,
\begin{equation}
	\ltwo{ \sum_{k \in \mathcal{S}_d(k_0-1)} \psi(k_0;k_0) \ket{\theta_k} - \sum_{k \in \mathcal{S}_d(k_0)} \psi(k_0;k) \ket{\theta_k} } = \ltwo{ \psi(k_0;k_0-d/2) \ket{\theta_{k_0-d/2}} - \psi(k_0;k_0+d/2) \ket{\theta_{k_0+d/2}} }.
\end{equation}

But $\ket{\theta_{k_0-d/2}} = \ket{\theta_{k_0+d/2}}$,
\begin{equation}
		\ltwo{ \sum_{k \in \mathcal{S}_d(k_0-1)} \psi(k_0;k_0) \ket{\theta_k} - \sum_{k \in \mathcal{S}_d(k_0)} \psi(k_0;k) \ket{\theta_k} } = \abs{ \psi(k_0;k_0-d/2)  - \psi(k_0;k_0+d/2) }.
\end{equation}

By direct substitution of $\psi(k_0;k_0-d/2)$ and $\psi(k_0;k_0+d/2)$ from Eq. \eqref{analyticposition}, we arrive at the Lemma statement.
\end{proof}

\begin{theorem}[Quasi-continuity of \gClock~states]\label{gaussiancontinuity} Let $k_0,t \in \rr$. Then the effect of the Hamiltonian $\hat{H}_c$ for the time $t$ on $\ket{{\Psi}_\textup{nor}(k_0)} \in\Lambda_{\sigma,n_0}$ is approximated by
\begin{equation}\label{eq:main eq theorem Q continuity}
	\me^{-\mi t  \hat{H}_c} \ket{\Psi_\textup{nor}(k_0)} = \ket{\Psi_\textup{nor}(k_0+t \frac{d}{T_0})} + \ket{\epsilon},
\end{equation}
where in the limits $d\rightarrow \infty$, $(0,d)\ni \sigma\rightarrow \infty$,
\begin{align}
\varepsilon_c:=\ltwo{ \ket{\epsilon} } =
\begin{cases}\displaystyle
\frac{|t|}{T_0} \, \bo\left( d^{9/4}\right)\me^{-\frac{\pi}{4}d\alpha_0^2} +\bo\left(d^{-1/4}\right)\me^{-\frac{\pi}{4}d}  &\mbox{if } \sigma=\sqrt{d} \\[10pt]\displaystyle
\frac{|t|}{T_0} \left(  \bo\left( d^2\sigma^{1/2}  \right)\me^{-\frac{\pi}{4}\sigma^2\alpha_0^2}+\bo\left( 1+\frac{d^3}{\sigma^{3/4}}  \right)\me^{-\frac{\pi}{4}\frac{d^2}{\sigma^2}}  \right) +\bo\left(\me^{-\frac{\pi}{4}\frac{d^2}{\sigma^2}}\right)+\bo\left(\me^{-\frac{\pi}{4}\sigma^2}\right) &\mbox{otherwise}
\end{cases}
\end{align}
More precisely,
\begin{align}\label{eq: thorem 1 tot error}
\varepsilon_c=\ltwo{ \ket{\epsilon} } &< |t| \frac{d}{T_0} \epsilon_{total} + \left( |t| \frac{d}{T_0} + 1 \right) \epsilon_{step} +\epsilon_\textup{nor}(t),
\end{align}
where,
\begin{align}
\epsilon_{total} =\!
\begin{cases}\displaystyle
2\pi A d \left( 2 \sqrt{d} \left( \frac{\alpha_0}{2} + \frac{1}{2\pi d} + \frac{1}{1 - e^{-\pi\alpha_0}} \right) e^{-\frac{\pi d}{4}\alpha_0^2} + \left(\frac{2}{1 - e^{-\pi}} + \frac{1}{2} + \frac{1}{2\pi d} \right) e^{-\frac{\pi d}{4}}\right)\! &\mbox{if } \sigma=\sqrt{d} \\[10pt]
 2\pi A d \left( 2 \sigma \left( \frac{\alpha_0}{2} + \frac{1}{2\pi\sigma^2} + \frac{1}{1 - e^{-\pi\sigma^2\alpha_0}} \right) e^{-\frac{\pi\sigma^2}{4}\alpha_0^2} + \left( \frac{1}{1 - e^{-\frac{\pi d}{\sigma^2}}} + \frac{1}{1 - e^{-\frac{\pi d^2}{\sigma^2}}} + \frac{d}{2\sigma^2} + \frac{1}{2\pi d} \right) e^{-\frac{\pi d^2}{4\sigma^2}} \right)\! &\mbox{otherwise},
\end{cases}
\end{align}
with $A=\bo(\sigma^{-1/2})$ and is upper bounded by Eq. \eqref{eq:up low bounds for A normalize}, and
\begin{align}
\epsilon_{step} &<
\begin{cases}
2 A e^{-\frac{ \pi d}{4}} &\mbox{if } \sigma=\sqrt{d}\\
2 A e^{-\frac{ \pi d^2}{4\sigma^2}} &\mbox{otherwise.}
\end{cases}\\
\epsilon_\textup{nor}(t)&=\left|\sqrt{\frac{\sum_{k\in \mathcal{S}_d(k_0+td/T_0)} \me^{-\frac{2\pi}{\sigma^2}(k-k_0-\frac{d}{T_0}t)^2}}{\sum_{l\in \mathcal{S}_d(k_0)} \me^{-\frac{2\pi}{\sigma^2}(l-k_0)^2}}}-1 \right| \leq 
\begin{cases}\displaystyle
\frac{40}{3}\frac{ \me^{-\frac{\pi}{2}d}}{1 - \me^{-\pi}}\quad &\forall\, t\in\rr \;\;\mbox{ if } \sigma=\sqrt{d}\\[10pt]\displaystyle
\frac{40\sqrt{2}}{3\sigma}\left(\frac{ \me^{-\frac{\pi d^2}{2\sigma^2}}}{1 - \me^{-\frac{2\pi d}{\sigma^2}}}+\frac{\sigma}{\sqrt{2}}\frac{ \me^{-\frac{\pi \sigma^2}{2}}}{1 - \me^{-\pi \sigma^2}}\right) \quad &\forall\, t\in\rr \;\;\mbox{ otherwise},
\end{cases}\label{eq:ep nor simple form above}
\end{align}
where on the r.h.s. of the inequality in Eq. \eqref{eq:ep nor simple form above} we have assumed $\sigma\geq1$, $d=2,3,4,\ldots$; (tighter bounds can be found in Section \ref{Re-normalizing the clock state}).
\end{theorem}

\emph{Intuition.} The discrete clock mimics the perfect clock, with an error that grows linearly with time, and scales better than any inverse polynomial w.r.t. the dimension of the clock. The optimal decay is when the state is completely symmetric, (see Def. \ref{def:clock stat classes}). This gives exponentially small error in $d$, the clock dimension, with a decay rate coefficient $\pi/4$. When the initial clock state is symmetric, but not completely symmetric, the error is still exponential decay but now with a modified coefficient which decreases as the clock's initial state's mean energy approaches either end of the spectrum of $\hat H_c$.

\begin{proof} For, $t\geq 0$, directly apply the previous two Lemmas (\ref{smalltimetranslation}, \ref{rangeshift}) in alternation, first to move $k_0$ from one integer to the next, then to switch from ${\lfloor k_0\rfloor}$ to ${\lfloor k_0+1\rfloor}$ if $d$ is even, and ${\lfloor k_0+1/2\rfloor}$ to ${\lfloor k_0+1/2+1\rfloor}$ if $d$ is odd, and finally arriving at 
\begin{align}
\me^{-it \hat{H}_c} \ket{{\Psi}(k_0)} &= \ket{{\Psi}(k_0+ \frac{d}{T_0} t)} + \ket{\epsilon},\quad\quad\ltwo{\ket{\epsilon}} \leq t \frac{d}{T_0} \epsilon_{total}+\left(t \frac{d}{T_0}+1\right)\epsilon_{step}. \end{align}
To conclude Eq. \eqref{eq: thorem 1 tot error} for $t>0$, one now has to normalize the states using bounds from Section \ref{Re-normalizing the clock state} and then use Lemma \ref{unitaryerroraddition} to upper bound the total error. For $t<0$, simply evaluate Eq. \eqref{eq:main eq theorem Q continuity} for a time $|t|$ followed by multiplying both sides of the equation by $\exp(\mi t\hat H_c)$, mapping $k_0\rightarrow k_0-t\;d/T_0$ and noting the unitary invariance of the $l_2$ norm of $\ket{\epsilon}$.
\end{proof}

\section{Autonomous quasi-control of \gClock~states}\label{sec:Proof of quasi-continuous control of Gaussian clock states}
\subsection{Theorem and proof of autonomous quasi-control}
Recall the \textit{control} of the idealized clock, Eq. \eqref{idealpotential} (or Eq. \ref{idealregularity} in Appendix for more details). Namely, the time translated state of the clock under its natural Hamiltonian plus a potential for a time $t$, was the same as the space-translated state multiplied by a phase factor, with the phase given by the integral over the potential up to time $t$. Our second major result is to derive an analogous statement for the finite clock. This is the subject of this Section, with the result in the form of a Theorem stated at the end of the Section,  Theorem (\ref{movig through finite time}). The proof will follow similar lines to the proof of the continuity of the \gClock~states of the previous Section (\ref{sec:Continuity of the Gaussian clock state}). The main extra complications will be two: 
\begin{itemize}
\item[1)]
Unlike in the previous Section (\ref{sec:Continuity of the Gaussian clock state}), the Hamiltonian will now include a potential term, i.e. $\hat H=\hat H_c+\hat V_d$ where $\hat H_c$ and $\hat V_d$ do not commute\footnote{For generality which will be useful in future work, $\hat V_d$ will not necessarily be self adjoint as the definition will make clear. However, with a small abuse of terminology, we will still refer to $\hat H=\hat H_c+\hat V_d$ as the \textit{Hamiltonian} since it will still be the generator of a semi-group parametrised by time parameter $t\geq 0$.}. To get around this difficulty, we will have to employ the Lie product formula to split the time evolution into consecutive infinitesimal time steps with separate contributions from $\hat H_c$ and $\hat V_d$. Bounding the infinitesimal time step contribution from $\hat H_c$ will be the subject of Lemma \ref{lemm:1}, while the contribution from $\hat V_d$ will be bounded in Lemma \ref{lem:2 infinitesimal pot}.  Combining both contributions will be the subject of Lemma \ref{Moving the clock through finite time, within unit angle}.
\item [2)] The Fourier Transform of the clock state (see Def. \ref{Continuous Fourier Transform no pot}), will be replaced by the Fourier Transform of the clock state multiplied by an exponentiated integral over the potential. Unlike in the previous case, one will not be able to perform the Fourier Transform analytically, thus requiring to upper bound how quickly it decays with clock dimension $d$. This turns out not to be so simple, and is the subject of Lemma \ref{lem:the crucial fourierbound with potential} which is the main mathematically challenging difference between this section and Section \ref{sec:Continuity of the Gaussian clock state}. 
\end{itemize}

\begin{definition}\label{def:continous pot} \emph{(Continuous Interaction Potential)}. Let ${V}_0 : \rr \rightarrow \rr\cup\hh^{-}$ (where $\hh^{-}:= \{ a_0+\mi b_0 : a_0\in\rr, b_0 < 0 \}$ denotes the lower-half complex plane) be an infinitely differentiable function of period $2\pi$, normalized so that
\begin{equation}\label{eq:theta normalisation}
	\int_x^{x+2\pi} {V}_0(x^\prime) dx^\prime =: \Omega\in\rr\cup\hh^{-},
\end{equation}
and define ${V}_d : \rr \rightarrow \rr\cup\hh^{-}$ as
\begin{equation}\label{potstretch}
	{V}_d(x) = \frac{2\pi}{d}\,{V}_0 \left( \frac{2\pi x}{d} \right).
\end{equation}
For convenience of notation, we also define the function
\be\label{eq: Theta def}
\Theta(\Delta;x)=\int_{x-\Delta}^x dy  V_d(y).
\ee
\end{definition}
\begin{definition}\label{def:interaction pot def} \emph{(Interaction Potential).} Let $\hat{V}_d$ be an operator on $\mathcal{H}_c$ defined by
\begin{equation}
	\hat{V}_d = \frac{d}{T_0} \sum_k {V}_d(k) \ketbra{\theta_k}{\theta_k}.
\end{equation}
\end{definition}
\begin{remark}
	Thus ${V}_d(x)$ is a continuous extension of the discrete elements $\braket{\theta_k|\hat{V}_d|\theta_k}T_0/d$. Note that because ${V}_d$ has a period of $d$, the summation in the above expression may run over any sequence of $d$ consecutive integers without affecting the operator $\hat{V}_d$. Note that for the special case $V_0: \rr\rightarrow\rr$, $\hat V_d$ is self adjoint. This is the most important case for this \comm, and the more general case will only be useful for future work \cite{RMRenatoetal}.
\end{remark}
\begin{definition}\label{def:stand alone def mathcal N}\emph{(Decay rate parametres).}
Let $b$ be any real number satisfying
\be\label{eq:b def eq}
b\geq\; \sup_{k\in\nn^+}\left(2\max_{x\in[0,2\pi]} \left|  V_0^{(k-1)}(x) \right|\,\right)^{1/k},
\ee
where $ V_0^{(p)}(x)$ is the $p^\textup{th}$ derivative with respect to $x$ of $ V_0(x)$ and $V_0^{(0)}:=V_0$.
We can use $b$ to define $\mathcal{N}\in\nn^0$ as follows
\be 
\mathcal{N}=
\begin{cases}\displaystyle
\left\lfloor \frac{\pi\alpha_0^2}{2\left(\frac{\pi\kappa\alpha_0}{\ln(\pi\alpha_0 d)}b+1\right)^2} \,d \right\rfloor, &\mbox{if } \sigma=\sqrt{d}\\[20pt]\displaystyle
\left\lfloor \frac{\pi\alpha_0^2}{2\left(\frac{\pi\kappa\alpha_0}{\ln(\pi\alpha_0\sigma^2)}b+\frac{d}{\sigma^2}\right)^2} \left(\frac{d}{\sigma}\right)^2 \right\rfloor, &\mbox{otherwise}
\end{cases}
\ee
where $\kappa=0.792$ and $\alpha_0\in(0,1)$ characterizes how close the mean energy of the initial state of the clock is to either extremum of the energy spectrum, defined in Def. \ref{Distance of the mean energy from the edge}.
\end{definition}

We can use the above definitions to define the following parametre.

\begin{definition}\label{def:decay rate params}\emph{(Exponential decay rate parameter)}.
We define the rate parameter as
\be 
\bar\upsilon= \frac{\pi\alpha_0\kappa}{\ln\left(\pi\alpha_0\sigma^2  \right)}b
\ee
where recall $\kappa=0.792$.
\end{definition}
We will often require that $\bar\upsilon\geq 0$. This is equivalent to requiring $\pi\alpha_0 \sigma^2>1$ if $b>0$, and is always satisfied in the special case $b=0$.\\

To proceed further, we will need a generalized definition of Def. \ref{def:analytic ext. of clock states} which incorporates a potential.

\begin{definition}\label{def:analyticmomentum control version} \emph{(Continuous Fourier Transform of the control analytic extension)} Let $\tilde{\psi}(k_0;\cdot) : \rr \rightarrow \mathbb{C}$ be defined as the continuous Fourier transform of $\psi$ multiplied by an exponentiated integral over the potential function,
\begin{align}\label{analyticmomentum control version}
	\tilde{\psi}(k_0,\Delta ;p) &= \frac{1}{\sqrt{d}} \int_{-\infty}^\infty \psi(k_0 ;x)\,\me^{-\mi \int_{x-\Delta}^x  V_d(x')dx'} \me^{-\mi 2\pi px/d} dx ,
\end{align}
where $\psi$ is defined in Def. \ref{def:analytic ext. of clock states} and $ V_d$ in Def. \ref{def:continous pot}.
\end{definition}
\begin{remark}
Note that Eq. \eqref{analyticmomentum control version} is well defined since $|\tilde{\psi}(k_0,\Delta ;p)| \leq  \int_{-\infty}^\infty |\psi(k_0 ;x)|dx/\sqrt{d} <\infty$.
\end{remark}

\begin{lemma}\label{lem:the crucial fourierbound with potential}\emph{(Bounding a sampled version of the Fourier Transform of the clock state with potential).}
Let $\bar\epsilon_2$ be defined by
\be 
\bar\epsilon_2=\frac{\mi 2\pi}{\sqrt{d}}\sum_{n=0}^{d-1}\sum_{k=-\infty}^\infty k\, \me^{\mi 2\pi n l/d}\, \tilde \psi(k_0,\Delta;n+k d).
\ee
We have the bound  
\be 
|\bar\epsilon_2|<
\begin{cases}
 (2\pi)^{5/4}\sigma^{3/2} A \left(1+\frac{\pi^{2}}{8}\right)\sqrt{\frac{\me}{2}\frac{\alpha_0}{(\bar\upsilon\sigma^2/d+1)}}\exp\left(-\frac{\pi}{4}\frac{\alpha_0^2}{\left(\frac{d}{\sigma^2}+\bar\upsilon\right)^2} \left(\frac{d}{\sigma}\right)^2 \right)\quad &\text{if } \mathcal{N}\geq 8 \text{ and } \bar\upsilon\geq 0\\
\frac{3^{7/4}}{\sqrt{2\pi}\me}\frac{A(8+\pi^{2})}{\alpha_0^3}\left(\frac{\kappa\sqrt{6\pi}}{\ln(3)}b+\frac{d}{\sigma}\right)^3 \left( \frac{\sigma}{d^3} \right)\quad &\text{otherwise} \label{eq:lemma sampled ft 2}
\end{cases}
\ee
where $b$, $\mathcal{N}$, and $\alpha_0$ are defined in Def. \ref{def:stand alone def mathcal N} while $\bar\upsilon$ in Def.  \ref{def:decay rate params}.
\end{lemma}

\begin{proof}
Outline: to bound $|\bar\epsilon_2|$, the main challenge is in bounding the Fourier transform $\tilde \psi(\cdot,\cdot;\cdot)$, and show that it is exponentially decaying in $d$ (for $d$ large enough such that the first if condition in Eq. \eqref{eq:lemma sampled ft 2}  is satisfied). In order to achieve this, we will integrate by parts the Fourier transform $N=1,2,3,\ldots$ times followed by taking absolute values, thus generating a different bound for every $N$. We will then choose to bound the Fourier transform using a different bound (i.e. different $N$), depending of the value of $d$, i.e. $N=N(d)$. We will have to bound the derivatives produced by integrating by parts $N$ times. We will use results from combinatorics for this. The proof will make essential use of: the binomial theorem, the generalized Leibniz rule, Rodriguez formulas for Hermite polynomials, orthogonality conditions of Hermite polynomials, the Cauchy-Schwarz inequality, Fa\a di Bruno's formula, Bell polynomials, Bell numbers, analytic upper bounds for Bell numbers, Sterling's formula, and the Fundamental theorem of calculus. Also note that definitions \ref{def:stand alone def mathcal N} and \ref{def:decay rate params} have been defined more generally in the proof. One could use these slightly more general definitions to tighten the bound in Eq. \eqref{eq:lemma sampled ft 2} for small $d$ if desired.

We have that
\be\label{eq:ep2 def}
|\bar\epsilon_2|\leq \frac{2 \pi}{\sqrt{d}} \sum_{n=0}^{d-1}\left( \sum_{s=-\infty} ^{-1} |s\,\tilde \psi( k_0,\Delta; n+sd)| +\sum_{s=1}^\infty |s\,\tilde \psi( k_0,\Delta; n+sd)|   \right),
\ee
where from definitions \ref{def:analyticmomentum control version},  \ref{def:analytic ext. of clock states} we have
\be\label{int tilde phi}
\tilde \psi( k_0,\Delta;p)=\frac{A}{\sqrt{d}} \int _{-\infty}^\infty dx\, \me^{-\frac{\pi}{\sigma^2}(x-k_0)^2}\me^{\mi 2\pi n_0(x-k_0)/d} \me^{-\mi \theta(x)} \me^{-\mi 2 \pi p x/d}.
\ee
where we denote
\be\label{eq:theta mischa def}
\theta(x)= \frac{2 \pi}{d}\int_{x-\Delta}^x dy  V_0(2\pi y/d),
\ee
where $ V_0(\cdot)$ is a smooth, periodic function with period $2\pi$ defined in Eq. \eqref{eq:theta normalisation}. 
Performing the change of variable $z=2 \pi(k_0-x)/d$, we find in Eq. \eqref{int tilde phi} 
\be\label{eq:U in int 1st time}
\tilde \psi( k_0,\Delta; p)=\frac{A}{\sqrt{d}}\frac{d}{2\pi} \me^{-\mi p 2\pi k_0/d}\int _{-\infty}^\infty dz\,U(z) \me^{\mi( p+\gamma) z},
\ee
where
\be 
U(z):= \me^{-\left(\frac{d }{2\sqrt{\pi}\sigma}\right)^2 z^2} \me^{-\mi \tilde\theta(z)}\me^{-\mi (n_0+\gamma)z},
\ee
with
\be\label{eq:theta tilde mischa def}
\tilde\theta(z):=\theta\left(-\frac{d z}{2 \pi}+k_0  \right),
\ee
and $\gamma\in \big(-(d-1)/2,(d-1)/2 \big).$ We now integrate by parts $N$ times the integral in Eq. \eqref{eq:U in int 1st time}, differentiating $U(\cdot)$ once in every iteration. This requires that $U(\cdot)$ is differentiable $N$ times and since we will require $N$ to be unbounded from above, hence the requirement that $ V_0(\cdot)$ is smooth. Taking this all into account, we have
\be
\tilde \psi(k_0,\Delta; p)=-\frac{A}{\sqrt{d}} \frac{d}{2\pi}\me^{-\mi p 2\pi k_0/d}\frac{1}{\big(-\mi\left(p+\gamma\right)\big)^{N}}\int _{-\infty}^\infty dz\,U^{(N)}(z) \me^{\mi\big( p+\gamma\big) z},\quad N\in\nn^+,
\ee
where $U^{(N)}(z)$ denotes the $N^\text{th}$ derivative of $U(z)$ w.r.t. $z$. Thus taking absolute values, we achieve
\be\label{eq:div pars ineq}
\tilde \psi( k_0,\Delta; p)\leq\frac{A}{\sqrt{d}}\frac{d}{2\pi} |p+\gamma|^{-N}\int _{-\infty}^\infty dz\,\left|U^{(N)}(z)\right|,\quad N\in\nn^+.
\ee
We now substitute Eq. \eqref{eq:div pars ineq} into  Eq. \eqref{eq:ep2 def} obtaining the upper bound
\begin{align}\label{eq:ep2 1sr up bound}
|\bar\epsilon_2|&\leq \frac{2\pi}{\sqrt{d}}\frac{A}{\sqrt{d}}\frac{d}{2\pi}\left(\int_{-\infty}^\infty dz \left| U^{(N)}(z)\right|\right) \sum_{n=0}^{d-1}\sum_{s=1}^\infty\left( \frac{|s|}{ \left|n+sd+\gamma\right|^{N}} + \frac{|s|}{ \left|n-sd+\gamma\right|^{N}}   \right)\\
&\leq \left( A G d\int_{-\infty}^\infty dz \left| U^{(N)}(z)\right|\right) \left(\frac{d (1-|\beta|)}{2}\right)^{-N} , \quad N=3,4,5,\ldots
\end{align}
where $\beta\in(-1,1)$ and
\begin{align}\label{def: G definition}
G:=  \frac{1}{d}\sum_{n=0}^{d-1}\sum_{s=1}^\infty\frac{|s|}{\left(\frac{2}{d(1-|\beta|)}\right)^N}
\left( \frac{1}{ \big|n+sd+\gamma\big|^{N}} + \frac{1}{  \big|n-sd+\gamma\big|^{N}}   \right).
\end{align}
The rational to defining $G$ in this way, is that we will soon parametrize $\gamma$ in terms of $\beta$ such that $G$ will be upper bounded by a $d$ independent constant. We will now find this constant before proceeding to bound Eq. \eqref{eq:ep2 1sr up bound}. First we will parametrize $\gamma$ is such a way that the symmetry in the summations of $G$ becomes transparent. Let $\gamma=-(d-1)/2-y_0$, $y_0\in \big[0,-(d-1)/2\big)$. Substituting into Eq. \eqref{def: G definition}, we find
\begin{align} 
G=&  \frac{1}{d}\sum_{n=0}^{d-1}\sum_{s=1}^\infty\frac{|s|}{\left(\frac{2}{d(1-|\beta|)}\right)^N}
\left( \frac{1}{ \big|n+sd-(d-1)/2-y_0\big|^{N}} + \frac{1}{  \big|n-sd-(d-1)/2-y_0\big|^{N}}   \right)\\
=&  \frac{1}{d}\sum_{n=0}^{d-1}\sum_{s=1}^\infty\frac{|s|}{\left(\frac{2}{d(1-|\beta|)}\right)^N}
\left( \frac{1}{ \big|n+sd-(d-1)/2-y_0\big|^{N}} + \frac{1}{  \big|-n-sd+(d-1)/2-y_0\big|^{N}}   \right)\\
=&  \frac{1}{d}\sum_{n=0}^{d-1}\sum_{s=1}^\infty\frac{|s|}{\left(\frac{2}{d(1-|\beta|)}\right)^N}
\left( \frac{1}{ \big|n+sd-(d-1)/2-|y_0|\big|^{N}} + \frac{1}{  \big|n+sd-(d-1)/2+|y_0|\big|^{N}}   \right),\label{eq:G intermediate}
\end{align}
and thus the sum only depends on the modulus of $y_0$. We will now make the change of variable $y_0=(d-1)\beta/2$ with $\beta\in(-1,1)$ leading to
\be\label{eq:gamma in terms of beta} 
\gamma= -\left(\frac{d-1}{2}\right) (1+\beta), \quad \beta\in(-1,1).
\ee

Plugging this into Eq. \eqref{eq:G intermediate} leads to 
\begin{align}
G=&  \frac{1}{d}\sum_{n=0}^{d-1}\sum_{s=1}^\infty
\left( \frac{|s|}{\left|\frac{(2n+1+|\beta|)}{(1-|\beta|)d}-\frac{(1+|\beta|)-2s}{(1-|\beta|)}\right|^{N}}   + \frac{|s|}{\left|\frac{(2n+1-|\beta|)}{(1-|\beta|)d}+\frac{(1-|\beta|)+2s}{(1-|\beta|)}\right|^{N}}   \right)\\
=&\frac{1}{d}\sum_{n=0}^{d-1}
\left(\sum_{s=0}^\infty \frac{s+1}{\left(\frac{(2n+1+|\beta|)}{(1-|\beta|)d}+1+\frac{2s}{1-|\beta|}\right)^{N}}   +\sum_{s=1}^\infty  \frac{s}{\left(\frac{(2n+1-|\beta|)}{(1-|\beta|)d}+1+\frac{2s}{(1-|\beta|)}\right)^{N}}   \right)
\\
\leq&\frac{1}{d}\sum_{n=0}^{d-1}
\left(\sum_{s=0}^\infty \frac{s+1}{\left(1+2s\right)^{N}}   +\sum_{s=1}^\infty  \frac{s}{\left(1+2s\right)^{N}}   \right)\leq 1+\sum_{s=0}^\infty
 \frac{1}{\left(2s+1\right)^{N-1}}\\
 \leq & 1+\sum_{s=0}^\infty \frac{1}{(2s+1)^2}=1+\frac{\pi^{2}}{8}\approx 2.234\quad \text{for } N=3,4,5,\ldots\,.
\end{align}
Our next task will be to bound $\int_{-\infty}^\infty dz \left| U^{(N)}(z)\right|$. We will start by dividing $U(z)$ into a product of unitaries: $U(z)=U_1(z)U_2(z)U_3(z)$ where 
\be 
U_1(z):= \me^{-d z^2/4\pi},\quad U_2(z):= \me^{-\mi \tilde \theta(z)},\quad U_3(z):= \me^{-\mi(n_0+\gamma)z}.
\ee
We will take advantage of the distinct properties of $U_1,U_2,U_3$. We will start by using the \textit{general Leibniz rule} for $n$ times differentiable functions $u$ and $v$: 
\be 
(uv)^{(n)}=\sum_{k=0}^n \binom{n}{k} u^{(k)}v^{(n-k)},
\ee
where $\binom{n}{k}$ is the \textit{binomial coefficient} we are using the standard superscript round bracket notation to indicated derivatives.
We thus find
\begin{align}
\int_{-\infty}^\infty dz\left| U^{(N)}(z)\right|
&=\int_{-\infty}^\infty dz\left| \sum_{k=0}^N \binom{N}{k} \left(U_1(z)U_2(z)\right)^{(k)}  U_3^{(N-k)}(z)\right|\\
&=\int_{-\infty}^\infty dz\left| \sum_{k=0}^N \binom{N}{k}\sum_{q=0}^k \binom{k}{q}U_1^{(q)}(z)U_2^{(k-q)}(z)\left(-\mi(n_0+\gamma)\right)^{N-k}  U_3(z)\right|\\
&\leq \sum_{k=0}^N \binom{N}{k}\sum_{q=0}^k \binom{k}{q}|n_0+\gamma|^{N-k}\int_{-\infty}^\infty dz\left|U_1^{(q)}(z)U_2^{(k-q)}(z)\right|\\\label{eq:U N dev in terms of U1 U2}
\end{align}
We now need to relate $U_1^{(q)}$ to the Hermite polynomials in order to bound the integral. The Rodriguez formula for the $n^\text{th}$ Hermite polynomial $H_n(x)$ is \cite{Rodriguez}
\be 
H_n(x)=(-1)^n \me^{x^2}\frac{d^n}{dx^n}\left(\me^{-x^2}\right), \quad n\in \nn^0.
\ee
By the change in variable $x=\frac{d}{2\sigma\sqrt{\pi}}z$ we can relate $U_1^{(q)}$ to the Hermite polynomials:
\be\label{eq:U in terms of Hermite}
U^{(n)}_1(z)=\left(\frac{d}{2\sigma\sqrt{\pi}}\right)^{n}(-1)^n\, \me^{-\left(\frac{d}{2\sigma\sqrt{\pi}}\right)^2 z^2} H_n\left(\frac{d}{2\sigma\sqrt{\pi}}z\right),\quad n\in \nn^0.
\ee
Due to the periodic nature of $ V_0(\cdot)$, we have that $|U_2^{(k)}(z)|$ is bounded in $z\in\rr$ for $k\in \nn^0$. We can thus use the Cauchy--Schwarz inequality in conjunction with Eqs. \eqref{eq:U N dev in terms of U1 U2}, \eqref{eq:U in terms of Hermite} to obtain
\begin{align}
&\int_{-\infty}^\infty  dz\left| U^{(N)}(z)\right| \\
&\leq \sum_{k=0}^N  \binom{N}{k}\sum_{q=0}^N \binom{N}{q}|n_0+\gamma|^{N-k}\left(\frac{d}{2\sigma\sqrt{\pi}}\right)^{q}\sqrt{ \int_{-\infty}^\infty dy\, \me^{-\left(\frac{d}{2\sigma\sqrt{\pi}}\right)^2 y^2}\left|U_2^{(k-q)}(z) \right|^2  \int_{-\infty}^\infty dz H_q^2\left(\frac{d}{2\sigma\sqrt{\pi}}\,z\right)\me^{-\left(\frac{d}{2\sigma\sqrt{\pi}}\,z\right)^2 z^2} }\\
&\leq \sum_{k=0}^N  \binom{N}{k}\sum_{q=0}^N \binom{k}{q}|n_0+\gamma|^{N-k}\left(\frac{d}{2\sigma\sqrt{\pi}}\right)^{q-1/2} C_{k-q}\,\sqrt{ \int_{-\infty}^\infty dy\, \me^{-\left(\frac{d}{2\sigma\sqrt{\pi}}\right)^2 y^2}  \int_{-\infty}^\infty dx\, H_q^2\left(x\right)\me^{-x^2} }\label{eq:U N dev in terms of H_n}
\end{align}
where we have defined
\be\label{eq:def Ck}
C_k\geq \left|U_2^{(k)}(z)\right|,\quad \forall z\in \rr,\;\;\; k\in \nn^0.
\ee
Using the orthogonality conditions of the Hermite polynomials\cite{Rodriguez}:
\be 
\int_{-\infty}^\infty dy  H_p(y) H_q(y)\me^{- y^2}=\delta_{p,q} \sqrt{\pi}\, 2^q (q!),\quad p,q\in \nn^0
\ee
Eq. \eqref{eq:U N dev in terms of H_n} reduces to 
\begin{align}\label{eq:U N dev in terms of C k}
\int_{-\infty}^\infty dz\left| U^{(N)}(z)\right|
\leq & \frac{2\sigma\pi}{d}\sum_{k=0}^N \binom{N}{k}|n_0+\gamma|^{N-k}\sum_q ^k\binom{k}{q} \left(\frac{d}{\sigma\sqrt{2\pi}}\right)^{q} C_{N-q} \sqrt{(q!)}.
\end{align}
Before we can continue, we will now bound $C_k$ in terms of derivatives of $\tilde \theta(\cdot)$. For this, let us first recall Fa\`a di Bruno's formula written in terms of the Incomplete Bell Polynomials $B_{n,k}(x_1,x_2,\ldots,x_{n-k+1})$ \cite{Faa_di_Bruno}: for $u$,$v$ $n$-times differentiable functions, we have
\be 
\frac{d^n}{dx^n}u\left(v(x)\right)=\sum_{k=0}^n u^{(k)}\left( v(x)\right) B_{n,k}\left(v^{(1)}(x),v^{(2)}(x),\ldots,v^{(n-k+1)}(x)\right).
\ee 
By choosing $U(x)=\me^x$, $v(x)=-\mi \tilde \theta(x)$, it follows 
\be\label{eq:Un bell poly}
U_2^{(n)}(x)=U_2(x) \sum_{k=1}^n B_{n,k}\left(-\mi \tilde\theta^{(1)}(x),\ldots,-\mi \tilde\theta^{(n-k+1)}(x)\right)=:  U_2(x) B_{n}\left(-\mi \tilde\theta^{(1)}(x),\ldots,-\mi \tilde\theta^{(n)}(x)   \right),
\ee
where $B_n$ are the Complete Bell Polynomials. Using the formula\cite{Faa_di_Bruno}
\be 
B_{n,k}(x_1,x_2,\ldots,x_{n-k+1})=\sum_{\{j_k\}_k} \frac{n!}{j_1!j_2!\ldots j_{n-k+1}!} \left(\frac{x_1}{1!}\right)^{j_1}\left(\frac{x_2}{2!}\right)^{j_2}\ldots \left(\frac{x_{n-k+1}}{(n-k+1)!}\right)^{j_{n-k+1}},
\ee
where $j_1+j_2+\ldots+j_{n-k+1}=k$ and $j_1+2j_2+\ldots+({n-k+1})j_{n-k+1}=n$, we see that if $x_i\leq a b^i$, for $a,b>0$ $i=1,2,3,\ldots ,{n-k+1}$, we have
\be 
\left|B_{n,k}(x_1,x_2,\ldots,x_{n-k+1})\right| \leq b^n B_{n,k}(a,a,\ldots,a).
\ee
Let
\be\label{eq: dev tilde theta constraint}
\left| \tilde \theta^{(n)}(x)\right|\leq a\, b^n,
\ee
for some $a,b>0$ for $n\in \nn^+$, Eq. \eqref{eq:Un bell poly} gives us
\be 
\left|U_2^{(n)}(x)\right| \leq b^n B_n(a,a,\ldots,a). 
\ee
Setting $a=1$ and noting that $B_n(1,1,\ldots,1)=\textup{Be}_n$, where $\textup{Be}_n$ is the $n^\textup{th}$ Bell number\cite{Faa_di_Bruno}, we achieve
\be\label{eq: U2 in Bell num}
\left|U_2^{(n)}(x)\right| \leq b^n\, \textup{Be}_n. 
\ee
Using Eqs. \eqref{eq:def Ck}, \eqref{eq:U N dev in terms of C k}, \eqref{eq: U2 in Bell num}, and introducing variables $\upsilon,b_2\geq 0$ via the definition
\be\label{eq:b upsilon and b_2 def}
b=\upsilon b_2,
\ee
we have
\begin{align}
\int_{-\infty}^\infty dz\left| U^{(N)}(z)\right|
\leq & \frac{2\pi\sigma}{d}\sum_{k=0}^N \binom{N}{k}|n_0+\gamma|^{N-k}\sum_{q=0}^k \binom{k}{q}\upsilon^{k-q} \max_{p=0,\ldots,k}\left\{\left(\frac{d}{\sigma\sqrt{2\pi}}\right)^{p} \sqrt{(p!)}\,b_2^{k-p}\, \textup{Be}_{k-p} \right\}\\
\leq & \frac{2\pi\sigma}{d}\sum_{k=0}^N \binom{N}{k}|n_0+\gamma|^{N-k}\sum_{q=0}^k \binom{k}{q}\upsilon^{k-q} \max_{p=0,\ldots,N}\left\{\left(\frac{d}{\sigma\sqrt{2\pi}}\right)^{p} \sqrt{(p!)}\,b_2^{N-p}\,\textup{Be}_{N-p} \right\},\\
\leq & \frac{2\pi\sigma}{d}\left(|n_0+\gamma|+\upsilon+1\right)^N \max_{p=0,\ldots,N}\left\{\left(\frac{d}{\sigma\sqrt{2\pi}}\right)^{p} \sqrt{(p!)}\,b_2^{N-p}\, \textup{Be}_{N-p} \right\},\label{eq:U N 1st max prob}
\end{align}
where we have used twice the identity
\be 
\sum_{k=0}^n\binom{n}{k} g_1^k g_2^{n-k}=(g_1+g_2)^n\quad \forall g_1,g_2,\in \rr,\, n\in \nn^0.
\ee
We will now proceed to upper bound the maximisation problem in Eq. \eqref{eq:U N 1st max prob}. Using Sterling's bound for factorials and the a bound for the Bell numbers \cite{Berendr}, 
\be 
n!\leq \me\, n^{n+1/2} \me^{-n},\quad \textup{Be}_n< \left(\frac{\kappa\, n}{\ln(n+1)}\right)^n, \text{with } \kappa=0.792,\quad n\in\nn^+
\ee
together with $\textup{Be}_0=1$, we can thus write
\begin{align}\label{eq:max f def}
\max_{p=0,\ldots,N}\left\{\left(\frac{d}{\sigma\sqrt{2\pi}}\right)^{p} \sqrt{(p!)}\,b_2^{N-p}\, \textup{Be}_{N-p} \right\}\\
 \leq \sqrt{\me}\, \exp \left(\max_{q=0,\ldots,N}\left\{f(q)\right\}\right),
\end{align}
where we have defined $f: 0\cup [1,N]\rightarrow \rr$ for $N=3,4,5,\ldots$
\be\label{def:f def}
f(x):=
\begin{cases}
N\ln b_2+N\ln\left(\frac{\kappa N}{\ln(N+1)}\right), & \quad \text{if }  x=0
\\ x\ln\left(\frac{d}{\sqrt{2\pi}\sigma}\right)+(x/2+1/4)\ln(x)-\frac{x}{2}+(N-x)\ln\left(\frac{\kappa (N-x)}{\ln(N-x+1)}\right)+(N-x)\ln(b_2), & \quad \text{if }  x\in[1,N)\\
  N\ln\left(\frac{d}{\sqrt{2\pi}\sigma}\right)+(N/2+1/4)\ln(N)-\frac{N}{2}, & \quad \text{if }  x=N.
\end{cases}
\ee
Note that $f(x)$ is continuous on the interval $x\in [1,N]$. By explicit calculation, we have
\be\label{eq:2nd dev of f } 
f^{(2)}(x)=\frac{1}{N-x}+\frac{1}{2x}\left( 1-\frac{1}{2x}\right)+\frac{G(N-x)}{(N-x)(N-x+1)^2\ln^2(N-x+1)} \text{ for}\quad x\in[1,N),\,\, N=3,4,5,\ldots 
\ee
where 
\be
G(x):=(x+1)^2\ln^2(x+1)+x^2-(x+2)x\ln(x+1),\quad x\geq 0.
\ee
Due to the following two properties satisfied by $G$,
\begin{align}
 G(0)&=0,\\
 \frac{d}{dx} G(x)&=\frac{x^2+2(1+x)^2\ln^2(1+x)}{1+x}>0,\quad\text{for } x>0,
\end{align}
we conclude $G(x)>0$ for $x>0$ and thus from Eq. \eqref{eq:2nd dev of f },
\be
f^{(2)}(x)>0 \text{ for } x \in[1,N), \quad \forall \,N= 3,4,5,\ldots
\ee
hence $f(x)$ is convex on $x\in[1,N)$ and we can write Eq. \eqref{eq:max f def} as
\be\label{eq:max written with f q}
\max_{p=0,\ldots,N}\left\{\left(\frac{d}{\sqrt{2\pi}\sigma}\right)^{p} b_2^{N-p}\, \textup{Be}_{N-p}\, \sqrt{(p!)}\right\}
 \leq \sqrt{\me}\, \exp \left(\max_{p=0,1,N}\left\{f(p)\right\}\right).
\ee
We now want $\max_{p=0,1,N}\left\{f(p)\right\}=f(N)$. This is true if $1\geq f(0)/f(N)$ and $1\geq f(1)/f(N)$. By direct calculation using Eq. \eqref{def:f def}, we can solve these constraints for $b_2$. We find
\begin{align}
b_2&\leq \frac{d}{\sqrt{2\pi \me}\sigma}\frac{1}{\kappa}\frac{\ln(N+1)}{N^{1/2-1/4N}}\label{eq:f0/fN ineq}\\
b_2&\leq \frac{d}{\sqrt{2\pi \me}\sigma}\frac{1}{\kappa}\frac{\ln(N)}{(N+1)N^{-\left(N/2(N-1)+1/4(N-1)\right)}}. \label{eq:f1/fN ineq}
\end{align}
Therefore, we need $b_2$ to satisfy both Eqs. \eqref{eq:f0/fN ineq}, \eqref{eq:f1/fN ineq}, namely
\begin{align}\label{eq:b2 condition}
b_2&\leq \min\{b_L,b_R\},\\
b_L&:=\frac{d}{\sqrt{2\pi \me}\sigma}\frac{1}{\kappa}\frac{\ln(N+1)}{N^{1/2-1/4N}}\\
b_R&:= \frac{d}{\sqrt{2\pi \me}\sigma}\frac{1}{\kappa}\frac{\ln(N)}{(N+1)N^{-\left(N/2(N-1)+1/4(N-1)\right)}}
\end{align}
Thus if Eq. \eqref{eq:b2 condition} is satisfied, from Eqs. \eqref{eq:U N 1st max prob}, \eqref{eq:max f def}, and \eqref{eq:max written with f q}, we achieve 
\begin{align}
\int_{-\infty}^\infty dz\left| U^{(N)}(z)\right|
\leq & \frac{2\pi\sigma}{d} \left(|n_0+\gamma|+\upsilon+1\right)^N \sqrt{\me}\,\me^{f(N)}.
\end{align}
And hence plugging this into Eq. \eqref{eq:ep2 1sr up bound}
\begin{align}
|\bar\epsilon_2|
&\leq  Ad\,G  \frac{2\pi\sigma\sqrt{\me}}{d}\left(\frac{(1-|\beta|)d}{2}\right)^{-N}\left( |n_0+\gamma|+\upsilon+1 \right)^N \me^{f(N)} \\
&\leq  2\pi\sigma A G \sqrt{\me} N^{1/4} \exp{\left(-N\ln\left(\frac{(1-|\beta|)d}{2}\right)+N\ln(|n_0+\gamma|+\upsilon+1)+N\ln\left( \frac{d}{\sqrt{2\pi}\sigma} \right)+N\ln\sqrt{N}-N/2\right)}\\
&\leq  2\pi\sigma A G \sqrt{\me} N^{1/4} \exp{\left(N\ln\left(\frac{|n_0+\gamma|+\upsilon+1}{(1-|\beta|)}\sqrt{\frac{2}{\pi\me}\frac{N}{\sigma^2}}\right)\right)}, \quad N=3,4,5,\ldots,\label{eq:ep in terms of N exp form}
\end{align}
Recall that our objective is to prove that
$|\bar\epsilon_2|$ decays exponentially fast in $d$. For this, we will choose $N$ depending on the value of $d$. Although there is no explicit $d$ dependency in the exponential in Eq. \eqref{eq:ep in terms of N exp form}, recall that $\sigma$ is a function of $d$ and as such we will parametrize $N$ in terms of $\sigma$. For the exponential in Eq. \eqref{eq:ep in terms of N exp form} to be negative, we want 
\be  
0<\frac{|n_0+\gamma|+\upsilon+1}{(1-|\beta|)}\sqrt{\frac{2}{\pi\me}\frac{N}{\sigma^2}}<1
\ee
to hold. Solving for $N$ gives us
\be\label{eq:N up bound}
N< \sigma^2\left(\frac{1-|\beta|}{|n_0+\gamma|+\upsilon+1}  \right)^2 \frac{\pi \me}{2}.
\ee
We thus set
\be
N=N(\sigma)=\left\lfloor \sigma^2\left(\frac{1-|\beta|}{|n_0+\gamma|+\upsilon+1}  \right)^2 \frac{\pi \me}{2\chi} \right\rfloor
\ee
where $\chi$ is a free parameter we will choose such that Eq. \eqref{eq:N up bound} holds while optimizing the bound. Using the bounds $\lfloor x \rfloor \leq   x$ for $x\in \rr$, and noting that Eq. \eqref{eq:ep in terms of N exp form} is monotonically increasing in $N=3,4,5,\ldots,$  we find plugging into Eq. \eqref{eq:ep in terms of N exp form} 
\begin{align}
|\bar\epsilon_2|
< \sigma A G (2\pi)^{5/4}\sqrt{\frac{\sigma\, (1-|\beta|)}{2(|n_0+\gamma|+\upsilon+1)}}\left(\frac{e}{\chi}\right)^{1/4} \,\exp\left(-\frac{\pi\me}{4\chi}\left(\frac{\sigma\, (1-|\beta|)}{|n_0+\gamma|+\upsilon+1} \right)^2 \ln(\chi)\right), \quad N(\sigma)=3,4,5,\ldots.
\end{align}
We now choose $\chi=\me$ to maximise $\ln (\chi)/\chi$ in the exponential and choose the parametrization
\be\label{eq:n0 with beta 0}
n_0=\left(\frac{d-1}{2}\right)(1+\beta_0). \quad \beta_0\in[-1,1]
\ee
Recalling Eq. \eqref{eq:gamma in terms of beta}, we thus achieve the final bound
\be\label{eq: exp decay ep2 bound}
|\bar\epsilon_2|
< \sigma A G (2\pi)^{5/4}\sqrt{\frac{\sigma\, (1-|\beta|)}{(d-1)|\beta_0-\beta|+2(\upsilon+1)}} \,\exp\left(-\frac{\pi}{4}\left(\frac{\sigma\, 2(1-|\beta|)}{(d-1)|\beta_0-\beta|+2(\upsilon+1)} \right)^2\right), \quad N(\sigma)=3,4,5,\ldots,
\ee
with
\be\label{eq:N(d)} 
N=N(\sigma)=\left\lfloor \sigma^2\left(\frac{1-|\beta|}{(d-1)|\beta_0-\beta|/2+\upsilon+1}  \right)^2 \frac{\pi}{2} \right\rfloor,
\ee
where recall that $\beta\in(-1,1)$ is a free parametre which we can choose to optimise the bound. In the case that $N(\sigma)$ given by Eq. \eqref{eq:N(d)} does not satisfy $N(\sigma)=3,4,5,\ldots$, we will bound $|\bar\epsilon_2|$ by setting $N=3$ in Eq. \eqref{eq:ep in terms of N exp form}. Taking into account definitions \eqref{eq:gamma in terms of beta}, \eqref{eq:n0 with beta 0}, this gives
\be\label{eq: exp decay ep2 bound N(d)=0,1,2}
|\bar\epsilon_2|< 2\pi\sqrt{\me}3^{1/4}\sigma A G \exp{\left(3\ln\left(\frac{|\beta_0-\beta|(d-1)/2+\upsilon+1}{1-|\beta|}\frac{1}{\sigma}\right)\right)}, \quad N(\sigma)=0,1,2.
\ee
We will now work out an explicit bound for $b_2$ defined via Eq. \eqref{eq:b2 condition}. We will use $N(\sigma)$ (Eq. \eqref{eq:N(d)}) to achieve a definition of $b_2$ as a function of $\sigma$. We start by lower bounding $\min\{b_L,b_R\}$
\begin{align}
\min\{b_L,b_R\}&=\frac{d}{\sqrt{2\pi}\kappa\sigma}\min\left\{ \frac{\ln(N+1)}{N^{1/2-1/4N}} ,  \frac{\ln(N)}{(N-1)N^{-(N/2(N-1)+1/4(N-1))}}\right\}\\
&\geq \frac{d}{\sqrt{2\pi}\kappa\sigma}
\min\left\{ \inf_{x\geq 3}\frac{1}{x^{-1/4x}} , \inf_{x\geq 3}\frac{1}{x^{-1/4(x-1)}}\right\}
\min\left\{ \frac{\ln(N+1)}{N^{1/2}} ,  \frac{\ln(N)}{(N-1)N^{-N/2(N-1)}}\right\}.\label{eq:intermediate 1 low bound min}
\end{align}
Note that the derivatives of $1/x^{-1/4x}$ and  $1/x^{-1/4(x-1)}$ are both negative for $x\geq 3$ and thus $ \inf_{x\geq 3}\frac{1}{x^{-1/4x}}=\lim_{x\rightarrow \infty}\frac{1}{x^{-1/4x}}=1$ and $\inf_{x\geq 3}\frac{1}{x^{-1/4(x-1)}}=\lim_{x\rightarrow\infty }\frac{1}{x^{-1/4(x-1)}}=1$. Thus from Eq. \eqref{eq:intermediate 1 low bound min}, we find 
\begin{align}
\min\{b_L,b_R\}&\geq \frac{d}{\sqrt{2\pi}\kappa\sigma}
\min\left\{ \frac{\ln(N+1)}{N^{1/2}} ,  \frac{\ln(N)}{(N-1)N^{-(N-1)/2(N-1)}}\right\}\\
&\geq \frac{d}{\sqrt{2\pi}\kappa\sigma}\sqrt{N}\,\frac{\ln(N)}{N}
\min\left\{ \frac{\ln(N+1)}{\ln(N)} ,  \frac{N}{(N-1)}\right\}\\
&\geq \frac{d}{\sqrt{2\pi}\kappa\,\sigma}\frac{\ln(N)}{\sqrt{N}}.\label{eq:lowebound min inter 2}
\end{align} 
We now upper bound $N$. From Eqs. \eqref{eq:N(d)}, \eqref{eq:gamma in terms of beta}, \eqref{eq:n0 with beta 0}, it follows
\be 
N\leq \sigma^2\frac{\pi}{2}\left( \frac{1-|\beta|}{|n_0+\gamma|+\upsilon+1} \right)^2\\
\leq 2\pi \sigma^2 \left( \frac{1-|\beta|}{|\beta_0-\beta|(d-1)+2} \right)^2.\label{eq:up bound N}
\ee
Now noting that
\be 
\frac{d}{dx}\frac{\ln(x)}{\sqrt{x}}=\frac{1}{x\sqrt{x}}\left(1-\frac{\ln(x)}{2}\right)<0 ,
\ee
for $x>\me^2\approx 7.39$, we can use Eq. \eqref{eq:up bound N} to lower bound Eq. \eqref{eq:lowebound min inter 2}. We find
\begin{align}
\min\{b_L,b_R\}\geq \mho,\quad \text{if } N(\sigma)\geq 8 \text{ and } \mho \geq 0
\end{align}
where $N(\sigma)$ is given by Eq. \eqref{eq:N(d)} and we have defined
\be\label{def:mho}
\mho:=\frac{d}{\sigma^2}\frac{1}{2\pi\,\kappa}\frac{|\beta_0-\beta|(d-1)+2}{1-|\beta|}\ln\left(2\pi\sigma^2\left( \frac{1-|\beta|}{|\beta_0-\beta|(d-1)+2} \right)\right).
\ee
The constraint $\mho\geq 0$ is for consistency with the requirement $\upsilon\geq 0$.
Recall Eq. \eqref{eq:b2 condition}, namely that $b_2$ must satisfy $\min\{b_L,b_R\} \geq b_2$, thus taking into account Eq. \eqref{eq:N(d)} and recalling that $\upsilon=b/b_2$ a consistent solution is
\be\label{eq:b_2 for N >=8}
b_2=\mho\quad \text{if } \mathcal{N} \geq 8  \text{ and } \mho \geq 0
\ee
where we have defined
\be\label{def:mathcal N}
\mathcal{N}:=\left\lfloor \sigma^2\left(\frac{1-|\beta|}{(d-1)|\beta_0-\beta|/2+b/\mho+1}  \right)^2 \frac{\pi}{2} \right\rfloor
\ee
When the if condition in Eq. \eqref{eq:b_2 for N >=8} is not satisfied, we can find a bound for $|\bar\epsilon_2|$ by setting  $N=3$ in Eq. \eqref{eq:ep in terms of N exp form}. This gives us Eq. \eqref{eq:lemma sampled ft 2}. For $\upsilon'$, we bound $b_2$ by evaluating $\min\{b_L,b_R\}$ for $N=3$ using the bound Eq. \eqref{eq:lowebound min inter 2}. This gives us
\be\label{eq:b_2 for N<4,5,6}
b_2=\frac{\ln(3)}{\sqrt{6\pi}\kappa}\frac{d}{\sigma},\quad \text{if } \mathcal{N}<8 \text{ and/or } \mho< 0
\ee
We will now workout what the constraint Eq. \eqref{eq: dev tilde theta constraint} with $a=1$, $b=\upsilon b_2$ imposes on potential function $V_0$. From definitions Eq. \eqref{eq:theta mischa def} and \eqref{eq:theta tilde mischa def}, we have
\be 
\tilde \theta (z)=\frac{2\pi}{d}\int_{-\frac{d z}{2\pi}+k_0-\Delta}^{-\frac{d z}{2\pi}+k_0} dy  V_0(2\pi y/d)=\int_{-z+\tilde k_0-\tilde \Delta}^{- z+\tilde k_0} dy  V_0(y)=\left( \mathcal{V}(- z+\tilde k_0)- \mathcal{V}(- z+\tilde k_0-\tilde\Delta)\right),
\ee
where we have defined the re-scaled constants $\tilde k_0:=2\pi k_0/d$, $\tilde \Delta:=2\pi \Delta/d$ and have used the Fundamental Theorem of Calculus, to write the integral in terms of $\mathcal{V}$, where $\mathcal{V}^{(1)}(x)=V_0(x)$.  We can now take the first $n$ derivatives of $\tilde\theta$:
\begin{align} 
\tilde\theta^{(n)}(z)&= \left(-1\right)^n\frac{d^n}{dy^n}\left( \mathcal{V}(y)- \mathcal{V}(y-\tilde\Delta)\right)\\
&= \left(-1\right)^{n-1}\frac{d^{n-1}}{dy^{n-1}}\left( V_0(y-\tilde\Delta)-V_0(y)\right)=\left(-1\right)^{n-1}\frac{d^{n-1}}{dx^{n-1}} V_0(x)\Bigg{]}_y^{y-\tilde\Delta}\\
&=\left(-1\right)^{n-1} V_0^{(n-1)}(x)\Big{]}_y^{y-\tilde\Delta},\quad n\in \nn^+
\end{align}
where $y:=-z+\tilde k_0$. Hence
\begin{align} 
\left|\tilde\theta^{(n)}(z)\right|^{1/n}&=   \left|  V_0^{(n-1)}(x)\Big{]}_{y}^{y-\tilde\Delta} \right|^{1/n}\leq  \left(\max_{x_1,x_2\in[0,2\pi]} \left|  V_0^{(k-1)}(x)\Big{]}_{x_1}^{x_2} \right|\,\right)^{1/k}\\
&\leq  \sup_{k\in \nn^+}\left(2\max_{x\in[0,2\pi]} \left|  V_0^{(k-1)}(x)\right|\,\right)^{1/k} , \quad \forall\; n\in\nn^+ \text{ and } \forall\; z\in \rr.
\end{align}
Thus from Eqs. \eqref{eq: dev tilde theta constraint}, \eqref{eq:b upsilon and b_2 def}  we conclude that $b\geq 0$ is any non negative number satisfying 
\be\label{eq: b_1 explicit}
b\geq\; \sup_{k\in \nn^+}\left(2\max_{x\in[0,2\pi]} \left|  V_0^{(k-1)}(x)\right|\,\right)^{1/k}.
\ee
and
\be\label{eq: upsilon explicit}
\upsilon= \frac{b}{b_2}= \frac{\sigma^2}{d}\frac{2\pi\,\kappa(1-|\beta|)}{|\beta_0-\beta|(d-1)+2}\frac{b}{\ln\left(2\pi\sigma^2\left( \frac{1-|\beta|}{|\beta_0-\beta|(d-1)+2} \right)\right)}, \quad \text{if } \mathcal{N}\geq 8 \text{ and } \mho \geq 0
\ee
where we have used Eq. \eqref{eq:b_2 for N >=8}.
For $\mathcal{N}<8$, we use Eq. \eqref{eq:b_2 for N<4,5,6} to achieve
\be
\upsilon= \frac{b}{b_2}=b\,\frac{\sqrt{6\pi}\kappa}{\ln(3)}\frac{\sigma}{d},\quad \text{if } \mathcal{N}<8 \text{ and/or } \mho< 0.
\ee
We are now ready to state the final bound. From Eqs. \eqref{eq: exp decay ep2 bound}, \eqref{eq:b_2 for N >=8} it follows
\be 
|\bar\epsilon_2|\leq (2\pi)^{5/4}\sigma A G\sqrt{\frac{\me\,(1-|\beta|)\,\sigma}{|\beta-\beta_0|(d-1)+2(\upsilon+1)}}\exp\left(x \right),\quad \text{if } \mathcal{N}\geq 8\text{ and } \mho \geq 0, 
\ee
with 
\be 
x=-\pi\frac{\left(1-|\beta|\right)^2}{\left(\left(|\beta-\beta_0|(d-1)+2\right)\frac{d}{\sigma^2}+\frac{4\pi\kappa(1-|\beta|)b}{\left(|\beta_0-\beta|(d-1)+2\right)\ln\left (\frac{2\pi\sigma^2(1-|\beta|)}{|\beta_0-\beta|(d-1)+2}\right)}\right)^2} \left(\frac{d}{\sigma}\right)^2,
\ee
where $\upsilon$ is given by Eq. \eqref{eq: upsilon explicit}, $\mathcal{N}$ by \eqref{def:mathcal N}, and $b$ by \eqref{eq: b_1 explicit}.
Recall that $\beta\in(-1,1)$ is a free parameter which we may choose to optimize the bound. For $\mathcal{N}< 8$, the bound is achieved from Eqs. \eqref{eq: exp decay ep2 bound N(d)=0,1,2}, \eqref{eq:b_2 for N<4,5,6}
\be
|\bar\epsilon_2|\leq \frac{3^{7/4}8}{\sqrt{2\pi}\me}A\,G \frac{\left( |\beta_0-\beta|(d-1)+2\left(\frac{\kappa\sqrt{6\pi}}{\ln(3)}\frac{\sigma}{d}\,b+1\right) \right)^3}{(1-|\beta|)^3}\frac{1}{\sigma^2},\quad \text{for } \mathcal{N}< 8\text{ and/or } \mho< 0.
\ee
The optimal choice of $\beta$ might depend on $\sigma$ and $d$, however, whenever $\beta_0\neq \pm 1$, i.e. mean energy of the initial clock state is not at one of the extremal points $0,d$, a good choice is $\beta_0=\beta$. Taking this into account and in order to simplify the bound, we will set $\beta_0=\beta$ to achieve
\be 
|\bar\epsilon_2|\leq (2\pi)^{5/4}\sigma^{3/2} A G\sqrt{\frac{\me}{2}\frac{\alpha_0}{(\upsilon+1)}}\exp\left(-\frac{\pi}{4}\left(\frac{\alpha_0}{\frac{d}{\sigma^2}+\frac{\pi\kappa\alpha_0}{\ln\left(\pi\alpha_0\sigma^2\right)}b}\right)^2 \left(\frac{d}{\sigma}\right)^2 \right),\quad \text{if } \mathcal{N}\geq 8 \text{ and } \mho\geq 0,
\ee
and
\be 
|\bar\epsilon_2|\leq \frac{3^{7/4}8}{\sqrt{2\pi}\me}\frac{A\,G}{\alpha_0^3}\left(\frac{\kappa\sqrt{6\pi}}{\ln(3)}b+\frac{d}{\sigma}\right)^3 \left( \frac{\sigma}{d^3} \right),\quad \text{if } \mathcal{N}< 8 \text{ and/or } \mho< 0,
\ee
where we have defined
\be 
\alpha_0:=1-|\beta_0|\in (0,1),
\ee
which can be written in terms of $n_0$ as
\begin{align}
\alpha_0&=\min\{ 1+|\beta_0|,1-|\beta_0|\}=\min\{ 1+\beta_0,1-\beta_0\}\\
&=\left(\frac{2}{d-1}\right) \min\{n_0,(d-1)-n_0\}\\
&=1-\left|1-n_0\,\left(\frac{2}{d-1}\right)\right|,\quad \text{for } n_0\in(0,d-1)
\end{align}
where we have used Eq. \eqref{eq:n0 with beta 0}. Furthermore, $\mho$ and $\mathcal{N}$ can also be simplified when $\beta=\beta_0$. From Eqs. \eqref{def:mho} we find
\be 
\mho=\frac{d}{\sigma^2}\frac{\ln(\pi\alpha_0\sigma^2)}{\pi\alpha_0\kappa},
\ee
while from \eqref{def:mathcal N} it follows
\be 
\mathcal{N}=\left\lfloor \frac{\pi\alpha_0^2}{2\left(\frac{\pi\kappa\alpha_0}{\ln(\pi\alpha_0\sigma^2)}b+\frac{d}{\sigma^2}\right)^2} \left(\frac{d}{\sigma}\right)^2 \right\rfloor.
\ee
\end{proof}

Before we proceed, we now define a generalization of Def. \ref{def:Gaussian clock states} which includes a potential dependent phase.
\begin{definition}\label{def:gauss clock steate with pot defintion}\emph{(\gClock~states with Potential)}. Let $\Lambda_{ V_0,\sigma,n_0}$ be the following space of states in the Hilbert space of the $d$ dimensional clock,
\be\label{eq:lambda set with pot def}
\Lambda_{ V_0,\sigma,n_0}=\bigg\{ \ket{\bar{\Psi}(k_0,\Delta)}\in\mathcal{H}_\textup{c},\quad  k_0,\Delta\in\rr\bigg\},
\ee
where
\begin{align}\label{def:gauss clock steate with pot}
	\ket{\bar{\Psi}(k_0,\Delta)} = \sum_{\mathclap{\substack{k\in \mathcal{S}_d(k_0)}}}\me^{-\mi\Theta(\Delta;k)}\psi(k_0;k)\ket{\theta_k},
\end{align}
where $\Theta$, $\psi$ are defined in Eqs. \eqref{eq: Theta def}, \eqref{analyticposition}, respectively and depend on the parameters $\sigma \in (0,d)$, $n_0 \in (0,d-1)$ and function $ V_0$, which is defined in Def. \ref{def:continous pot}. $\mathcal{S}_d(k_0)$ is given by Eq. \eqref{eq: mathcal S def}.\\

In the special case that $\ket{\bar{\Psi}(k_0,\Delta)}$ is normalised, it will be denoted
\be 
\ket{\bar{\Psi}_\textup{nor}(k_0,\Delta)}=\ket{\bar{\Psi}(k_0,\Delta)},
\ee
and $A$ (see Eq. \eqref{analyticposition}) will satisfy Eq. \eqref{eq:A normalised}, s.t. $\braket{\bar{\Psi}_\textup{nor}(k_0,\Delta)|\bar{\Psi}_\textup{nor}(k_0,\Delta)}=1$.
\end{definition}
\begin{remark}
In general $\Lambda_{\sigma,n_0}\subseteq \Lambda_{ V_0,\sigma,n_0}$ with equality if $\Delta=0$ or $ V_0(x)=0$ for all $x\in\rr$ (compare Eqs. \eqref{eq:lambda set def} and \eqref{eq:lambda set with pot def}). 
\end{remark}

\begin{lemma}\label{lemm:1}\emph{(Infinitesimal evolution under the clock Hamiltonian).}
The action of the unitary operator $e^{-i\frac{T_0}{d} \delta \hat{H}_c}$ on an element of $\Lambda_{ V_0,\sigma,n_0}$ may be approximated by a translation by $\delta\geq 0$ on the continuous extension of the clock state. Precisely speaking,
\begin{align}
	e^{-i\frac{2\pi}{\omega d} \delta \hat{H}_c} \sum_{l\in\mathcal{S}_d(k_0)} \me^{-\mi \Theta(\Delta;l)}\psi(k_0;l) \ket{\theta_l} &= \sum_{l\in\mathcal{S}_d(k_0)} \me^{-\mi \Theta(\Delta;l-\delta)}\psi(k_0;l-\delta) \ket{\theta_l} + \ket{\epsilon},
\end{align}
where the $l_2$ norm of the error $\ket{\epsilon}$ is bounded by
\begin{align}
\ltwo{\ket{\epsilon}} &\leq  \delta \epsilon_T \sqrt{d}+C_1\delta^2\sqrt{d},\\
\epsilon_T &<
\begin{cases}
|\bar\epsilon_2|+2 A\left(\frac{2\pi}{1-\me^{-\pi}}+  \frac{\left(b+\frac{2\pi}{d}\right)}{1-\me^{-\pi d}} + \left( 2\pi+\pi d+ \frac{1}{d}\right)\right)\me^{-\frac{\pi}{4}d} &\mbox{if } \sigma=\sqrt{d}\\
	|\bar\epsilon_2|+2 A\left(\frac{2\pi}{1-\me^{-\frac{\pi d}{\sigma^2}}}+ \frac{ \left(b+\frac{2\pi}{d}\right)}{1-\me^{-\pi \frac{d^2}{\sigma^2}}} + \left( 2\pi\frac{d}{\sigma^2}+\pi\frac{d^2}{\sigma^2}+ \frac{1}{d}\right)\right)\me^{-\frac{\pi}{4}\frac{d^2}{\sigma^2}} &\mbox{otherwise} 
\end{cases}\\
&=\begin{cases}
\left(\bo(b)+\bo\left( \frac{d^{3/2}}{\bar\upsilon+1}\right)^{1/2}+\bo\left(d\right)\right) \exp\left(-\frac{\pi}{4}\frac{\alpha_0^2}{\left(1+\bar\upsilon\right)^2} d \right) &\mbox{if } \sigma=\sqrt{d}\\
\left(\bo(b)+\bo\left( \frac{\sigma^3}{\bar\upsilon \sigma^2/d+1}\right)^{1/2}+\bo\left(\frac{d^2}{\sigma^2}\right)\right) \exp\left(-\frac{\pi}{4}\frac{\alpha_0^2}{\left(\frac{d}{\sigma^2}+\bar\upsilon\right)^2} \left(\frac{d}{\sigma}\right)^2 \right)+\bo\left(\frac{d^2}{\sigma^2}\right)\me^{-\frac{\pi}{4}\frac{d^2}{\sigma^2}} &\mbox{otherwise,}
\end{cases}
\end{align}
with $|\bar\epsilon_2|$ given by Lemma \ref{lem:the crucial fourierbound with potential} and $C_1$ is $\delta$ independent.  $\Theta$, $\psi$ are given by Eqs. \eqref{eq: Theta def}, \eqref{analyticposition} respectively. 
\end{lemma}
\emph{Intuition}. This is simply the statement that for the class of \gClock~states that we have chosen, the effect of the clock Hamiltonian for an infinitesimal time is approximately the shift operator w.r.t. the angle space. The proof will follow along similar lines to that of Lemma \ref{infinitesimaltimetranslation} with the main difference being that we will now have to resort to Lemma \ref{lem:the crucial fourierbound with potential} in order to bound $\epsilon_2$ where as before bounding $\epsilon_2$ was straightforward and accomplished directly in the proof.

\begin{proof}
\begin{align}
	\me^{-\frac{\mi T_0}{ d} \delta \hat{H}_c} \ket{\bar{\Psi}(k_0,\Delta)} &= e^{-\frac{\mi 2\pi}{d} \delta \sum_{m=0}^{d-1} m\ket{E_m}\bra{E_m}} \sum_{k\in\mathcal{S}_d(k_0)} \!\!\! \me^{-\mi \Theta(\Delta;k)}\psi(k_0;k)  \ket{\theta_k}.
\end{align}

Switching the state to the basis of energy states, applying the Hamiltonian, and switching back via Eq. \eqref{finitetimestates_main},
\begin{align}
	\me^{-\frac{\mi T_0}{d} \delta \hat{H}_c} \ket{\bar{\Psi}(k_0,\Delta)} &=   \sum_{k,l\in\mathcal{S}_d(k_0)} \!\!\! \me^{-\mi \Theta(\Delta;k)}\psi(k_0;k)  \left( \frac{1}{d} \sum_{n=0}^{d-1} e^{-i2\pi n(k+\delta-l)/d} \right) \ket{\theta_l}
\end{align}

We label the above state as $\ket{\bar{\Psi}^{exact}_\delta}$. On the other hand we label as $\ket{\bar{\Psi}^{approx}_\delta}$ the following expression,
\begin{align}
	\ket{\bar{\Psi}^{approx}_\delta} &= \sum_{l\in\mathcal{S}_d(k_0)} \me^{-\mi \Theta(\Delta;l-\delta)}\psi(k_0;l-\delta) \ket{\theta_l},
\end{align}
which is simply a translation by $+\delta$ of the continuous extension of the clock state. Both the coefficients $\braket{\theta_l|\bar{\Psi}^{exact}_\delta}$ and $\braket{\theta_l|\bar{\Psi}^{approx}_\delta}$ are twice differentiable with respect to $\delta$. For $\braket{\theta_l|\bar{\Psi}^{exact}_\delta}$ this is clear. In the case of $\braket{\theta_l|\bar{\Psi}^{approx}_\delta}$ we note that it is a function of the derivative of 
\be 
\Theta(\Delta;l-\delta)=
\int_{l-\delta-\Delta}^{l-\delta}  V_d(x')dx',
\ee
with respect to $\delta$, and thus due to the fundamental theorem of calculus and the fact the $ V_d$ is a smooth function (and periodic), it follows that $\braket{\theta_l|\bar{\Psi}^{approx}_\delta}$ is differentiable with respect to $\delta$.

By Taylor's remainder theorem, the difference can be expressed as
\begin{equation}\label{Taylorshift}
	\braket{\theta_l|\bar{\Psi}^{approx}_\delta} - \braket{\theta_l|\bar{\Psi}^{exact}_\delta} = \braket{\theta_l|\bar{\Psi}^{approx}_0} - \braket{\theta_l|\bar{\Psi}^{exact}_0} + \delta \frac{d \left( \braket{\theta_l|\bar{\Psi}^{approx}_\delta} - \braket{\theta_l|\bar{\Psi}^{exact}_\delta} \right)}{d\delta} \bigg|_{\delta=0} + R(\delta), 
\end{equation}
where
\begin{equation}\label{Taylorshift C bound}
	\left| R(\delta) \right| \leq \frac{\delta^2}{2} \left(  \text{max}_{|t|\leq|\delta|} \left|  \frac{d^2 \left( \braket{\theta_l|\bar{\Psi}^{approx}_t} - \braket{\theta_l|\bar{\Psi}^{exact}_t} \right)}{dt^2} \right| \right) \leq C_1 \delta^2,
\end{equation}
and where $C_1$ is independent of $\delta$  because the second derivatives of both $\braket{\theta_l|\bar{\Psi}^{exact}_\delta}$ and $\braket{\theta_l|\bar{\Psi}^{approx}_\delta}$ w.r.t. $\delta$ are bounded for $\delta\in \rr$. The zeroth-order term vanishes, i.e. $\braket{\theta_l|\bar{\Psi}^{approx}_0} - \braket{\theta_l|\bar{\Psi}^{exact}_0}$ since $\ket{\bar{\Psi}^{exact}_0} = \ket{\bar{\Psi}^{approx}_0}$. For the first order term,
\begin{align}
	\frac{d\braket{\theta_l|\bar{\Psi}^{approx}_\delta}}{d\delta}\bigg|_{\delta=0} &= \frac{d\;}{d\delta} \me^{-\mi \Theta(\Delta;l-\delta)}\psi(k_0;l-\delta)\bigg|_{\delta=0} \\
	\frac{d \braket{\theta_l|\bar{\Psi}^{exact}_\delta}}{d\delta} \bigg|_{\delta=0} &= \left[ \frac{d}{d\delta} \sum_{k\in\mathcal{S}_d(k_0)}\me^{-\mi \Theta(\Delta;k)}\psi(k_0;k) \left( \frac{1}{d} \sum_{n=0}^{d-1} e^{-i2\pi n(k+\delta-l)/d} \right) \right]_{\delta=0} \\
	&= \left( \frac{-i 2\pi}{d} \right) \sum_{k\in\mathcal{S}_d(k_0)}\me^{-\mi \Theta(\Delta;k)}\psi(k_0;k) \left( \frac{1}{d} \sum_{n=0}^{d-1} n e^{-i2\pi n(k-l)/d} \right)
\end{align}

One can replace the finite sum over $k$ as an infinite sum, and bound the difference using Lemma \ref{G0},
\begin{align}
	\frac{d \braket{\theta_l|\bar{\Psi}^{exact}_\delta}}{d\delta} \bigg|_{\delta=0} &= \left( \frac{-i 2\pi}{d} \right) \sum_{k=-\infty}^\infty \me^{-\mi \Theta(\Delta;k)}\psi(k_0;k) \left( \frac{1}{d} \sum_{n=0}^{d-1} n e^{-i2\pi n(k-l)/d} \right) + \epsilon_1,\label{eq:div exact before poisson}
\end{align}
where,
\begin{align}
\left| \epsilon_1 \right| &\leq \frac{2\pi}{d} \left( \sum_{k\in \zz - \mathcal{S}_d(k_0)} \left|\psi(k_0;k)\right| \right) \frac{1}{d} \left( \sum_{n=0}^{d-1} n \right) \\
	&< 4\pi A\left(\frac{\me^{-\frac{\pi d^2}{4\sigma^2}}}{1-\me^{-\frac{\pi d}{\sigma^2}}}\right) & \text{by Lemma \ref{G0}}.
\end{align}

Applying the Poisson summation formula (Corollary \ref{poissonsummation}) to the sum in Eq. \eqref{eq:div exact before poisson} we achieve
\be 
\sum_{k=-\infty}^\infty \me^{-\mi \Theta(\Delta;k)}\psi(k_0;k) e^{-i2\pi nk/d}=\sum_{s=-\infty}^\infty \tilde{\psi}(k_0,\Delta;n+sd),
\ee
where $\tilde{\psi}$ is given by Def. \ref{def:analyticmomentum control version}. Thus we have
\begin{align}
	\frac{d \braket{\theta_l|\bar{\Psi}^{exact}_\delta}}{d\delta} \bigg|_{\delta=0} &= \left( \frac{-i 2\pi}{d} \right) \frac{1}{\sqrt{d}} \sum_{n=0}^{d-1} \sum_{s=-\infty}^{\infty} \tilde{\psi}(k_0,\Delta;n+sd) \; n e^{i2\pi nl/d} + \epsilon_1.
\end{align}

Since $\sum_{n=0}^{d-1} \sum_{s=-\infty}^{\infty} f(n+sd) = \sum_{n=-\infty}^\infty f(n)$, one can manipulate the expression accordingly,
\begin{align}
	\frac{d \braket{\theta_l|\bar{\Psi}^{exact}_\delta}}{d\delta} \bigg|_{\delta=0} &= \left( \frac{-i 2\pi}{d} \right) \frac{1}{\sqrt{d}} \left( \sum_{n=-\infty}^\infty \tilde{\psi}(k_0,\Delta;n) \; n e^{i2\pi nl/d} - \sum_{n=0}^{d-1} \sum_{s=-\infty}^{\infty} \tilde{\psi}(k_0,\Delta;n+sd) \;sd  e^{i2\pi nl/d} \right) + \epsilon_1
\end{align}

The second summation is a small contribution and has been bound in Lemma \ref{lem:the crucial fourierbound with potential},
\begin{align}
	\frac{d \braket{\theta_l|\bar{\Psi}^{exact}_\delta}}{d\delta} \bigg|_{\delta=0} &= \left( \frac{-i 2\pi}{d} \right) \frac{1}{\sqrt{d}} \sum_{n=-\infty}^\infty \tilde{\psi}(k_0,\Delta;n) \; n e^{i2\pi nl/d} + \epsilon_2 + \epsilon_1,
\end{align}
where
\begin{align}
\epsilon_2 &=\frac{\mi 2\pi}{\sqrt{d}}\sum_{n=0}^{d-1}\sum_{k=-\infty}^\infty k\, \me^{\mi 2\pi n l/d}\, \tilde \psi(k_0,\Delta;n+k d).
\end{align}
On the remaining sum, apply Eq. \eqref{eq: inv Fourier = dev} to write $\tilde \psi(k_0,\Delta;n) n$ in terms of the derivative of the Fourier transform $\tilde \psi(k_0,\Delta;n)$, followed by applying the Poisson summation formula,
\begin{align}
	\frac{d \braket{\theta_l|\bar{\Psi}^{exact}_\delta}}{d\delta} \bigg|_{\delta=0} &= \sum_{m=-\infty}^\infty \frac{d\;}{d\delta}\left(\me^{-\mi\Theta(\Delta;l-\delta+md)}\psi(k_0;l-\delta+md)\right) \bigg|_{\delta=0} + \epsilon_1 + \epsilon_2.
\end{align}

Replacing the sum by the $m=0$ term, and bounding the difference,
\begin{align}
	\frac{d \braket{\theta_l|\bar{\Psi}^{exact}_\delta}}{d\delta} &= \frac{d\;}{d\delta} \left(\me^{-\mi\Theta(\Delta;l-\delta)}\psi(k_0;l-\delta)\right)\bigg|_{\delta=0} + \epsilon_1 + \epsilon_2 + \epsilon_3, \label{deltacoeff}
\end{align}
where
\begin{align} \left| \epsilon_3 \right| &\leq  \sum_{m\in\zz-\{0\}} \abs{\frac{d\;}{d\delta} \left(\me^{-\mi\Theta(\Delta;l-\delta+md)}\psi(k_0;l-\delta+md)\right)}_{\delta=0} \\
	&\leq  \sum_{m\in\zz-\{0\}} \left( \abs{  V_d(l+md)- V_d(l+md-\Delta)   }\abs{\psi(k_0;l+md)}+\abs{\frac{d\,}{d\delta} \psi(k_0;l-\delta+md) }_{\delta=0} \right)\\
	&\leq  \sum_{m\in\zz-\{0\}} \left( b\abs{ \psi(k_0;l+md)}+\abs{ \frac{d\,}{d\delta}\psi(k_0;l-\delta+md)}_{\delta=0} \right)\\
&<
\begin{cases}
2 A\left(  \frac{\left(b+\frac{2\pi}{d}\right)}{1-\me^{-\pi d}} + \left( 2\pi+\pi d+ \frac{1}{d}\right)\right)\me^{-\frac{\pi}{4}d} &\mbox{if } \sigma=\sqrt{d}\\
	2 A\left( \frac{ \left(b+\frac{2\pi}{d}\right)}{1-\me^{-\pi \frac{d^2}{\sigma^2}}} + \left( 2\pi\frac{d}{\sigma^2}+\pi\frac{d^2}{\sigma^2}+ \frac{1}{d}\right)\right)\me^{-\frac{\pi}{4}\frac{d^2}{\sigma^2}} &\mbox{otherwise,} 
\end{cases}
\end{align}
for all $l\in\mathcal{S}_d(k_0)$ and where $b$ is defined in \eqref{eq:b def eq}. To achieve the last two lines, we have used Lemmas \ref{G0}, \ref{G1}.
Thus the first order term in Eq. (\ref{Taylorshift}) is composed of the terms $\epsilon_1,\epsilon_2,\epsilon_3$, and are bounded by
\begin{align}\label{eq:total epsilon}
	\frac{d \braket{\theta_l|\bar{\Psi}^{exact}_\delta}}{d\delta} - \frac{d \braket{\theta_l|\bar{\Psi}^{approx}_\delta}}{d\delta} &= \epsilon_1 + \epsilon_2 + \epsilon_3 =: \epsilon_T.
\end{align}
Note that $\epsilon_T$ is independent of the index $l$.

If we now define the error in the state as $\ket{\epsilon} =  \ket{\bar{\Psi}^{approx}_\delta} - \ket{\bar{\Psi}^{exact}_\delta}$, then from the properties of the norm, and Eqs. \eqref{Taylorshift}, \eqref{Taylorshift C bound}, \eqref{eq:total epsilon},
\begin{align}
	\ltwo{\ket{\epsilon}}^2 &\leq  \sum_{l\in\mathcal{S}_d(k_0)} \left( \bra{\bar{\Psi}_\delta^\text{approx}}- \bra{\bar{\Psi}_\delta^\text{exact}}\right)\ketbra{\theta_l}{\theta_l}\left( \ket{\bar{\Psi}_\delta^\text{approx}}- \ket{\bar{\Psi}_\delta^\text{exact}}\right)\\
	&\leq \sum_{l\in\mathcal{S}_d(k_0)}  \left| \braket{\theta_l|\bar{\Psi}_\delta^\text{approx}}- \braket{\theta_l|\bar{\Psi}_\delta^\text{exact}}\right|^2\\
	&\leq \left| \delta \epsilon_T+C_1\delta^2\right|^2 d.
\end{align}
\end{proof}

\begin{lemma}\label{lem:2 infinitesimal pot} \emph{(Infinitesimal evolution under the interaction potential).} The action of the operator $e^{-i\frac{T_0}{d} \delta \hat{V}_d}$, $\delta\geq 0$, on any state $\ket{\Phi} \in \mathcal{H}_c$ may be approximated by the following transformation,
\begin{align}\label{eq:lemma eq for pot infini}
	\me^{-\mi\frac{T_0}{ d} \delta \hat{V}_d} \ket{\Phi} &= \sum_{l\in\mathcal{S}_d(k_0)} \braket{\theta_l|\Phi} \me^{-\mi \Theta(\delta;l)} \ket{\theta_l} + \ket{\epsilon},
\end{align}
where the $l_2$ norm of the error $\ket{\epsilon}$ is bounded by
\begin{align}
\ltwo{\ket{\epsilon}} \leq C_2 \delta^2,
\end{align}
with $C_2$ being $\delta$ independent. $\hat V_d$ is defined in Def. \ref{def:interaction pot def}.
\end{lemma}

\begin{proof} Since the operator $\hat{V}_d$ is diagonal in the $\{\ket{\theta_k}\}$ basis,
\begin{align}\label{eq:intermidiate theta basis pot infini}
	\me^{-\frac{\mi T_0}{d} \delta \hat{V}_d} \ket{\Phi} = \me^{-\mi \delta \sum_k {V}_d(k) \ket{\theta_k}\bra{\theta_k}} \sum_{l\in\mathcal{S}_d(k_0)} \ket{\theta_l}\braket{\theta_l|\Phi} = \sum_{l\in\mathcal{S}_d(k_0)} \ket{\theta_l}\braket{\theta_l|\Phi} e^{-\mi\delta {V}_d(l)}.
\end{align}

We wish to replace $e^{- \mi \delta {V}_d(l)}$ by $e^{-\mi  \int_{l-\delta}^l V_d(x^\prime) dx^\prime}$. The difference may be bound using Taylor's theorem (to second order),
\begin{align}\label{eq: taylors pot infini}
	\me^{- \mi \delta {V}_d(l)} - \me^{-\mi  \int_{l-\delta}^l V_d(x^\prime) dx^\prime} = A + B\delta + R_l(\delta),
\end{align}
where
\begin{align}\label{Taylorshift C bound pot lemma}
	\left| R_l(\delta) \right|& \leq \frac{\delta^2}{2} \left(  \text{max}_{|t|\leq|\delta|} \left|  \frac{d^2 \left(\me^{- \mi t {V}_d(l)} - \me^{-\mi  \int_{l-t}^l V_d(x^\prime) dx^\prime} \right)}{dt^2} \right| \right)\\
	&\leq \frac{\delta^2}{2} \left(  \text{max}_{t\in\rr} \left|  \frac{d^2 \left(\me^{- \mi t {V}_d(l)} - \me^{-\mi  \int_{l-t}^l V_d(x^\prime) dx^\prime} \right)}{dt^2} \right| \right)\\
	&\leq C_2 \delta^2, \quad \forall\, l\in\mathcal{S}_d(k_0)
\end{align}
and $C_2$ is $\delta$ and $l$ independent. Such a $C_2$ exists, since $ V_d$ is periodic and smooth. Furthermore, by direct calculations $A=B=0$. Thus using the properties of the $l_2$ norm and Eqs. \eqref{eq:lemma eq for pot infini}, \eqref{eq:intermidiate theta basis pot infini}, \eqref{eq: taylors pot infini}, we find
\begin{align}
	\ltwo{\ket{\epsilon}}^2 &\leq \left(   \sum_{l\in\mathcal{S}_d(k_0)} |\! \braket{\theta_l|{\Phi}}\!|^* R_l(\delta)^*    \bra{\theta_l}\right)\left(   \sum_{k\in\mathcal{S}_d(k_0)} |\!\braket{\theta_k|\Phi}\!| R_k(\delta)   \ket{\theta_k}\right)\\
	&\leq \sum_{l\in\mathcal{S}_d(k_0)}|\!\braket{\theta_l|\Phi}\!| |R_l(\delta)|^2\\
	&\leq C_2^2 \delta^4.
\end{align}
\end{proof}

\begin{lemma}\label{Moving the clock through finite time, within unit angle} \emph{(Moving the clock through finite time, within unit angle).} Let $k_0 \in \rr$, and let $a \geq 0$ be s.t. $\mathcal{S}_d(k_0)=\mathcal{S}_d(k_0+a)$, where $\mathcal{S}_d$ is defined by Eq. \eqref{eq: mathcal S def}. Then the effect of the joint Hamiltonian $\hat{H}_c + \hat{V}_d$ for the time $\frac{T_0}{d}a$ on $\ket{\bar{\Psi}(k_0;\Delta)} \in \Lambda_{ V_0,\sigma,n_0}$ is approximated by
\begin{align}
	\me^{-\mi \frac{T_0}{d}a (\hat{H}_c + \hat{V}_d)} \ket{\bar{\Psi}(k_0,\Delta)} &= \ket{\bar{\Psi}(k_0+a,\Delta+a)} + \ket{\epsilon},
\end{align}
where the $l_2$ norm of the error $\ket{\epsilon}$ is bounded by
\begin{align}
\ltwo{\ket{\epsilon}} \leq a\, \epsilon_T \sqrt{d},
\end{align}
and $\epsilon_T$ is given by Lemma  \ref{lemm:1}.
\end{lemma}

\emph{Intuition.} Here is the core of the proof. By applying the Lie product formula, we show that the effect of the combined clock and interaction Hamiltonians is to simply shift the continuous extension of the clock state, while simultaneously adding a phase function that is the integral of the potential that the clock passes through. The discrete clock is thus seen to mimic the idealised momentum clock $\hat H=\hat p$. Here it is assumed the mean of the state does not pass through an integer (the full result follows later).

\begin{proof}
Consider the sequential application of $\me^{-\mi \frac{T_0}{d} \delta \hat{H}_c}$ followed by $\me^{-\mi \frac{T_0}{d} \delta \hat{V}_d}$ on $\ket{{\bar{\Psi}}(k_0,\Delta)}$, $\delta>0$. From the previous two Lemmas \ref{lemm:1}, \ref{lem:2 infinitesimal pot},
\begin{align}
	\me^{-\mi \frac{T_0}{d} \delta \hat{V}_d} \me^{-\mi \frac{T_0}{d} \delta \hat{H}_c} \ket{\bar{\Psi}(k_0,\Delta)} = \ket{\bar{\Psi}(k_0+\delta,\Delta+\delta)} + \ket{\epsilon_\delta}, 
\end{align}
where
\begin{align}
\ltwo{\ket{\epsilon}} \leq \delta \epsilon_T \sqrt{d}+\delta^2\left(C_1\sqrt{d} + C_2 \right).
\end{align}
where we combined the errors using the lemma on compiling norm non-increasing errors, Lemma \ref{unitaryerroraddition}.

Consider the above transformation repeated $m$ times on the state. Combining the errors using Lemma \ref{unitaryerroraddition} as before,
\begin{align}
	\left( \me^{-\mi \frac{T_0}{d} \delta \hat{V}_d} \me^{-\mi \frac{T_0}{d} \delta \hat{H}_c} \right)^m \ket{\bar{\Psi}(k_0,\Delta)} &=\ket{ \bar{\Psi} \left(k_0+ m\delta, \Delta+m\delta \right)} + \ket{\epsilon^{(m)}},\end{align}
where
\begin{align}
\ltwo{\ket{\epsilon^{(m)}}} \leq m\delta \epsilon_T \sqrt{d}+m\delta^2\left(C_1\sqrt{d} + C_2 \right).
\end{align}
This holds as long as the center of the state $k_0+m\delta$ has not crossed an integer value, i.e. if $d$ is even, $\lfloor k_0\rfloor= \lfloor k_0+m\delta\rfloor$, or if $d$ is odd $\lfloor k_0+1/2\rfloor= \lfloor k_0+1/2+m\delta\rfloor$; else $\mathcal{S}_d(k_0)\neq\mathcal{S}_d(k_0+m\delta)$ and the previous Lemmas \ref{lemm:1}, \ref{lem:2 infinitesimal pot} in this section will not hold.

To arrive at the lemma, set $\delta = a/m$, so that
\begin{align}\label{eq:to power m with pot}
	\left( e^{-i \frac{T_0}{d} \frac{a}{m} \hat{V}_d} e^{-i \frac{T_0}{d} \frac{a}{m} \hat{H}_c} \right)^m \ket{\bar{\Psi}(k_0,\Delta)} &= \ket{\bar{\Psi} \left(k_0+ a, \Delta+a \right)} + \ket{\epsilon^{(m)}},
\end{align}
where
\begin{align}
\ltwo{\ket{\epsilon^{(m)}}} \leq a \epsilon_T \sqrt{d}+m\left(\frac{a}{m}\right)^2\left(C_1\sqrt{d} + C_2 \right).
\end{align}

The Lie-Product formlula \cite{Lie}, states that for all $n \times n$ complex matrices $A$ and $B$,
\begin{equation}\label{eq:lie}
	\me^{A+B} = \lim_{N \rightarrow \infty} \left( \me^{\frac{A}{N}} \me^{\frac{B}{N}} \right)^N.
\end{equation}
We now use Eq. \eqref{eq:lie} to finalise the lemma. Consider the limit $m\rightarrow\infty$. On the l.h.s. of Eq. \eqref{eq:to power m with pot}, by the Lie product formula, we have that
\begin{equation}
	\lim_{m\rightarrow\infty} \left( e^{-i \frac{T_0}{d} \frac{a}{m} \hat{V}_d} e^{-i \frac{T_0}{d} \frac{a}{m} \hat{H}_c} \right)^m \ket{\bar{\Psi}(k_0,\Delta)}= e^{-i \frac{T_0}{d}a (\hat H_c+\hat V_d)}\ket{\bar{\Psi}(k_0,\Delta)}.
\end{equation}
On the r.h.s. of Eq. \eqref{eq:to power m with pot}, we have
\be 
\lim_{m\rightarrow\infty} \left(  \ket{\bar{\Psi} \left(k_0+ a, \Delta+a \right)} + \ket{\epsilon^{(m)}} \right)= \ket{\bar{\Psi} \left(k_0+ a, \Delta+a \right)} + \ket{\epsilon^{(\infty)}},
\ee
where $\ket{\epsilon^{(\infty)}}:=\lim_{m\rightarrow\infty} \ket{\epsilon^{(m)}}$, and
\be 
\ltwo{\ket{\epsilon^{(\infty)}}} \leq a \epsilon_T \sqrt{d}.
\ee
\end{proof}

At this point, we have already proven the continuity of the clock state for time translations that are finite, but small (i.e. small enough that the range $\mathcal{S}_d(k_0)$ remains the same). In order to generalize the statement to arbitrary translations we need to be able to shift the range itself, which is the goal of the following Lemma.

\begin{lemma}[Shifting the range of the clock state with potential]\label{lem:Moving the clock through one angle step}
If $d$ is even, and the mean of the clock state $k_0$ is an integer, or alternatively, if $d$ is $odd$ and $k_0$ is a half integer, then
\begin{align}
	  \sum_{k \in \mathcal{S}_d(k_0)} \me^{-\mi\Theta(\Delta;k)}\psi(k_0;k) \ket{\theta_k} &=\sum_{k \in \mathcal{S}_d(k_0-1)} \me^{-\mi\Theta(\Delta;k)}\psi(k_0;k) \ket{\theta_k} +\ket{\bar\epsilon_{step}}, \\
	\text{where} \quad \|\ket{\bar\epsilon_{step}}\|_2:=\bar\epsilon_{step} &<
	\begin{cases}
	  2 A e^{-\frac{ \pi d}{4}} &\mbox{if } \sigma=\sqrt{d}\\
	    2 A e^{-\frac{ \pi d^2}{4\sigma^2}} &\mbox{otherwise.}
	  \end{cases}
\end{align}
\end{lemma}

\begin{proof}
We prove the statement for even $d$, the proof for odd $d$ is analogous.

By definition \eqref{gaussianclock}, $\mathcal{S}_d(k_0)$ is a set of $d$ consecutive integers. Thus the only difference between $\mathcal{S}_d(k_0)$ and $\mathcal{S}_d(k_0-1)$ is the leftmost integer of $\mathcal{S}_d(k_0-1)$ and the rightmost integer of $\mathcal{S}_d(k_0)$, which differ by precisely $d$. By direct calculation, these correspond to the integers $k_0-d/2$ and $k_0+d/2$. These are the only two terms that do not cancel out in the statement of the Lemma,
\begin{align}
\|\ket{\bar\epsilon_{step}}\|_2 &\leq	\ltwo{ \sum_{k \in \mathcal{S}_d(k_0-1)} \me^{-\mi\Theta(\Delta;k)}\psi(k_0;k_0) \ket{\theta_k} - \sum_{k \in \mathcal{S}_d(k_0)} \me^{-\mi\Theta(\Delta;k)}\psi(k_0;k) \ket{\theta_k} } \\
&= \ltwo{ \me^{-\mi\Theta(\Delta;k_0-d/2)}\psi(k_0;k_0-d/2) \ket{\theta_{k_0-d/2}} - \me^{-\mi\Theta(\Delta;k_0+d/2)}\psi(k_0;k_0+d/2) \ket{\theta_{k_0+d/2}}}
\end{align}

But $\ket{\theta_{k_0-d/2}} = \ket{\theta_{k_0-d/2+d}}= \ket{\theta_{k_0+d/2}}$, giving
\begin{align}
		\|\ket{\bar\epsilon_{step}}\|_2 &\leq \abs{ \me^{-\mi\Theta(\Delta;k_0-d/2)}\psi(k_0;k_0-d/2)  - \me^{-\mi\Theta(\Delta;k_0+d/2)}\psi(k_0;k_0+d/2) }\\
		&\leq \abs{ \psi(k_0;k_0-d/2)}+\abs{\psi(k_0;k_0+d/2) }\\
		&=2 A \me^{-\frac{\pi}{4}\frac{d^2}{\sigma^2}}.
\end{align}

By direct substitution of $\psi(k_0;k_0-d/2)$ and $\psi(k_0;k_0+d/2)$ from \eqref{analyticposition}, we arrive at the Lemma statement. The lemma is analogous to that of Lemma, \ref{rangeshift}, but now with a  trivial extension due to the inclusion of the potential in the dynamics.
\end{proof}

\begin{theorem}[Moving the clock through finite time with a potential]\label{movig through finite time}
 Let $k_0,\Delta\in \rr$, and $t\in\rr$ if $V_0: \rr\rightarrow \rr$ while $t\geq 0$ otherwise. Then the effect of the Hamiltonian $\hat{H}_c + \hat{V}_d$ for time $t$ on $\ket{\bar{\Psi}_\textup{nor}(k_0,\Delta)} \in\Lambda_{ V_0,\sigma,n_0}$ is approximated by
\begin{align}\label{eq: main eq in control theorem}
	\me^{-\mi t (\hat{H}_c + \hat{V}_d)} \ket{\bar{\Psi}_\textup{nor}(k_0,\Delta)} &= \ket{\bar{\Psi}_\textup{nor}(k_0+ \frac{d}{T_0} t, \Delta+ \frac{d}{T_0} t)} + \ket{\epsilon},\quad\quad\ltwo{\ket{\epsilon}} \leq \varepsilon_v(t,d), \end{align}
where in the limits $d\rightarrow \infty$, $(0,d)\ni \sigma\rightarrow \infty$,
\be
\varepsilon_v(t,d)=
\begin{cases}
|t| \frac{d}{T_0}\left(\bo\left( \frac{d^{3/2}}{\bar\upsilon+1}\right)^{1/2}+\bo\left(d\right)\right) \exp\left(-\frac{\pi}{4}\frac{\alpha_0^2}{\left(1+\bar\upsilon\right)^2} d \right)+\bo\left( \me^{-\frac{\pi}{2}d} \right) &\mbox{if } \sigma=\sqrt{d}\\[10pt]
|t| \frac{d}{T_0}\!\left(\bo\left( \frac{\sigma^3}{\bar\upsilon \sigma^2/d+1}\right)^{1/2}\!\!+\bo\left(\frac{d^2}{\sigma^2}\right)\right) \exp\left(-\frac{\pi}{4}\frac{\alpha_0^2}{\left(\frac{d}{\sigma^2}+\bar\upsilon\right)^2} \left(\frac{d}{\sigma}\right)^2 \right)+\bo\left(|t|\frac{d^2}{\sigma^2}+1\right)\me^{-\frac{\pi}{4}\frac{d^2}{\sigma^2}}+\bo\left( \me^{-\frac{\pi}{2}\sigma^2} \right) &\mbox{otherwise.}
\end{cases}
\ee
and $V_0$ is defined in Def. \ref{def:continous pot},  $\ket{\bar{\Psi}_\textup{nor}(\cdot,\cdot)}$ in Def. \eqref{def:gauss clock steate with pot defintion}, $\alpha_0$ in Def. \ref{def:stand alone def mathcal N} while $\bar\upsilon$ is defined in Def. \ref{def:decay rate params}.\\

 More precisely, we have
\be\label{eq:ltwo epsilon nor control theorem}
\varepsilon_v(t,d)=|t| \frac{d}{T_0} \epsilon_T+\left(|t| \frac{d}{T_0}+1\right)\bar\epsilon_{step}+\epsilon_\textup{nor}(t),
\ee
where
\begin{align}
\epsilon_T &<
\begin{cases}
|\bar\epsilon_2|+2 A\left(\frac{2\pi}{1-\me^{-\pi}}+  \frac{\left(b+\frac{2\pi}{d}\right)}{1-\me^{-\pi d}} + \left( 2\pi+\pi d+ \frac{1}{d}\right)\right)\me^{-\frac{\pi}{4}d} &\mbox{if } \sigma=\sqrt{d}\\[10pt]
	|\bar\epsilon_2|+2 A\left(\frac{2\pi}{1-\me^{-\frac{\pi d}{\sigma^2}}}+ \frac{ \left(b+\frac{2\pi}{d}\right)}{1-\me^{-\pi \frac{d^2}{\sigma^2}}} + \left( 2\pi\frac{d}{\sigma^2}+\pi\frac{d^2}{\sigma^2}+ \frac{1}{d}\right)\right)\me^{-\frac{\pi}{4}\frac{d^2}{\sigma^2}} &\mbox{otherwise,}
\end{cases}
\end{align}
with $A=\bo(\sigma^{-1/2})$ and is upper bounded by Eq. \eqref{eq:up low bounds for A normalize} and, if $\sigma=\sqrt{d}$
\be 
|\bar\epsilon_2|<
\begin{cases}
 (2\pi)^{5/4}d^{3/4} A \left(1+\frac{\pi^{2}}{8}\right)\sqrt{\frac{\me}{2}\frac{\alpha_0}{(\bar\upsilon+1)}}\exp\left(-\frac{\pi}{4}\frac{\alpha_0^2}{\left(1+\bar\upsilon\right)^2} d \right)\quad &\text{if } \mathcal{N}\geq 8\text{ and } \bar\upsilon\geq 0\\
\frac{3^{7/4}}{\sqrt{2\pi}\me}\frac{A(8+\pi^{2})}{\alpha_0^3}\left(\frac{\kappa\sqrt{6\pi}}{\ln(3)}b+\sqrt{d}\right)^3 d^{-5/2}\quad &\text{otherwise,} 
\end{cases}
\ee
while in general
\be 
|\bar\epsilon_2|<
\begin{cases}
 (2\pi)^{5/4}\sigma^{3/2} A \left(1+\frac{\pi^{2}}{8}\right)\sqrt{\frac{\me}{2}\frac{\alpha_0}{(\bar\upsilon\sigma^2/d+1)}}\exp\left(-\frac{\pi}{4}\frac{\alpha_0^2}{\left(\frac{d}{\sigma^2}+\bar\upsilon\right)^2} \left(\frac{d}{\sigma}\right)^2 \right)\quad &\text{if } \mathcal{N}\geq 8\text{ and } \bar\upsilon\geq 0\\
\frac{3^{7/4}}{\sqrt{2\pi}\me}\frac{A(8+\pi^{2})}{\alpha_0^3}\left(\frac{\kappa\sqrt{6\pi}}{\ln(3)}b+\frac{d}{\sigma}\right)^3 \left( \frac{\sigma}{d^3} \right)\quad &\text{otherwise} 
\end{cases}
\ee
where $b$, $\mathcal{N}$, and $\alpha_0$ are defined in Def. \ref{def:stand alone def mathcal N} while $\bar\upsilon$ is defined in Def. \ref{def:decay rate params}.



\begin{align}
\bar\epsilon_{step} &<
\begin{cases}
2 A e^{-\frac{ \pi d}{4}} &\mbox{if } \sigma=\sqrt{d}\\
2 A e^{-\frac{ \pi d^2}{4\sigma^2}} &\mbox{otherwise.}
\end{cases}\\
\epsilon_\textup{nor}(t)&=\left|\sqrt{\frac{\sum_{k\in \mathcal{S}_d(k_0+td/T_0)} \me^{-\frac{2\pi}{\sigma^2}(k-k_0-\frac{d}{T_0}t)^2}}{\sum_{l\in \mathcal{S}_d(k_0)} \me^{-\frac{2\pi}{\sigma^2}(l-k_0)^2}}}-1 \right| \leq 
\begin{cases}\displaystyle
\frac{40}{3}\frac{ \me^{-\frac{\pi}{2}d}}{1 - \me^{-\pi}}\quad &\forall\, t\in\rr \;\;\mbox{ if } \sigma=\sqrt{d}\\[10pt]\displaystyle
\frac{40\sqrt{2}}{3\sigma}\left(\frac{ \me^{-\frac{\pi d^2}{2\sigma^2}}}{1 - \me^{-\frac{2\pi d}{\sigma^2}}}+\frac{\sigma}{\sqrt{2}}\frac{ \me^{-\frac{\pi \sigma^2}{2}}}{1 - \me^{-\pi \sigma^2}}\right) \quad &\forall\, t\in\rr \;\;\mbox{ otherwise},
\end{cases}\label{eq:q:ep nor above}
\end{align}
where on the r.h.s. of the inequality Eq. \eqref{eq:q:ep nor above}, we have assumed $\sigma\geq1$, $d=2,3,4,\ldots$; (tighter bounds can be found in Section \ref{Re-normalizing the clock state}).
\end{theorem}
\emph{Intuition.} In the special case $V_0: \rr \rightarrow \rr$, $\hat V_d$ is self adjoint and the discrete clock mimics the idealised clock, with an error that grows linearly with time, and scales better than any inverse polynomial w.r.t. the dimension of the clock. The physical significance of this theorem for the more general case in which the image of $V_0$ is the half-complex plane will be dealt with in future work \cite{RMRenatoetal}. The optimal decay of the error terms $\varepsilon_v$ is when the state is symmetric, i.e. when $\sigma=\sqrt{d}$. This gives exponentially small error in $d$, the clock dimension. The rate of decay i.e. the coefficient in the exponential, is determined by two competing factors, $\bar\upsilon$ and $\alpha_0$. The latter being a measure of how the initial clock's state's mean energy is to either the maximum or minimum value, while the former is a measure of how steep the potential function $\hat V_0$ is. The largest decay parameter, $\pi/4$ is only achieved asymptotically when both the mean energy of the initial clock state tends to the middle of the spectrum in the large $d$ limit and $\bar\upsilon$ tends to zero in the said limit. The latter limit, is true for arbitrarily steep potentials. 
\begin{remark}
Observe by Eq. \eqref{eq:ltwo epsilon nor control theorem} that $\lim_{t\rightarrow 0}\varepsilon_v(0,d)\neq 0$, yet from Eq. \eqref{eq: main eq in control theorem} clearly $\|\ket{\epsilon}\|_2=0$ in this limit. One can trivially modify the proof to find upper bounds for $\varepsilon_v(t,d)$ which are of order $t$ in the $t\rightarrow 0$ limit if required.
\end{remark} 
\begin{proof} For, $t\geq 0$, directly apply the previous two Lemmas (\ref{lem:Moving the clock through one angle step}, \ref{Moving the clock through finite time, within unit angle}) in alternation, first to move $k_0$ from one integer to the next, then to switch from ${\lfloor k_0\rfloor}$ to ${\lfloor k_0+1\rfloor}$ if $d$ is even, and ${\lfloor k_0+1/2\rfloor}$ to ${\lfloor k_0+1/2+1\rfloor}$ if $d$ is odd, and finally arriving at 
\begin{align}
\me^{-it (\hat{H}_c + \hat{V}_d)} \ket{\bar{\Psi}(k_0,\Delta)} &= \ket{\bar{\Psi}(k_0+ \frac{d}{T_0} t, \Delta+ \frac{d}{T_0} t)} + \ket{\epsilon},\quad\quad\ltwo{\ket{\epsilon}} \leq t \frac{d}{T_0} \epsilon_T+\left(t \frac{d}{T_0}+1\right)\bar\epsilon_{step}. \end{align}
To conclude Eq. \eqref{eq: main eq in control theorem} for $t>0$, one now has to normalize the states using bounds from Section \ref{Re-normalizing the clock state} and then use Lemma \ref{unitaryerroraddition} to upper bound the total error. For $t<0$ and $V_0: \rr\rightarrow \rr$, simply evaluate Eq. \eqref{eq: main eq in control theorem} for a time $|t|$ followed by multiplying both sides of the equation by the unitary operator  $\exp(\mi t(\hat H_c+\hat V_d))$, mapping $k_0\rightarrow k_0-t\;d/T_0$, $\Delta\rightarrow \Delta-t\;d/T_0$ and noting the unitary invariance of the $l_2$ norm of $\ket{\epsilon}$.
\end{proof}
\begin{corollary}
\label{movig through finite time coro}

\emph{(Evolution in the time basis with potential).} Let $k_0 \in \rr$ and $t\in\rr$ if $V_0: \rr\rightarrow \rr$ while $t\geq 0$ otherwise. Define a \textup{discreate wave function} in analogy with continuous wave functions, 
\be 
\bar{\psi}(x,t):=\bra{\theta_x}\me^{-\mi t (\hat{H}_c + \hat{V}_d)} \ket{\bar{\Psi}_\textup{nor}(k_0)},
\ee
for $x\in\mathcal{S}_d(k_0+t \,d/T_0)$
, $t\in\rr$. 
Then the time evolution of $\bar{\psi}$ is approximated by
\begin{align}\label{eq:bar psi evolution}
\bar{\psi}(x,t) &= \me^{-\mi \int_{x-t\,d/T_0}^x dy  V_d(y)}\,\bar{\psi}(x-\frac{d}{T_0} t,0) + \bar\varepsilon_v(t,d),\quad \quad |\bar\varepsilon_v(t,d)|\leq \varepsilon_v(t,d),
\end{align}
where $\bar\varepsilon_v\in\cc$ and $\varepsilon_v(t,d)$ is given by Eq. \eqref{eq:ltwo epsilon nor control theorem}.
\end{corollary}
\emph{(Intuition.)} This is the analogous statement to that of Eq. \eqref{eq:psi idealised control}, i.e. that up to an error $\varepsilon_v(t,d)$, the time evolution overlap of the \gClock~state in the time basis, by the clock Hamiltonian plus a potential term, is simply the time translated overlap multiplied by an exponentiated phase which integrates over the potential. 
\begin{proof}
This is a direct consequence of Theorem \ref{movig through finite time}. Note that by definition, 
\be 
\bar{\psi}(x,0)=\psi_\textup{nor}(k_0;x).
\ee
Thus by taking the inner product of Eq. \eqref{eq: main eq in control theorem} with $\ket{\theta_k}$ followed by noting $\psi(k_0;k+y)=\psi(k_0-y;k)$ and
 $|\!\braket{\theta_k|\epsilon}\!|\leq\varepsilon_v(t,d)$, we achieve Eq. \eqref{eq:bar psi evolution}.
\end{proof}
\begin{remark}
Note that when $t$ corresponds to one period, namely $t=T_0=2\pi/\omega$, it follows from the definitions that the state $\ket{\bar\Psi(k_0+\frac{d}{T_0} t,\Delta+\frac{d}{T_0} t)}$  in Theorem \ref{movig through finite time} has returned to its initial state, up to a global phase, namely
\be
	\me^{-\mi T_0 (\hat{H}_c + \hat{V}_d)} \ket{\bar{\Psi}(k_0,\Delta)} = \me^{-\mi\Omega}\ket{ \bar{\Psi} \left( k_0,\Delta \right)} + \ket{\epsilon},
\ee
where the $l$-two norm of $\ket{\epsilon}$ is as before. Similarly, for the discrete wafe function of Eq. \eqref{eq:bar psi evolution}, 
\begin{align}
\bar{\psi}(x,T_0) &= \me^{-\mi \Omega}\,\bar{\psi}(x,0) + \bar\varepsilon_v(T_0,d).
\end{align}
\end{remark}

\subsection{Examples of Potential functions: the cosine potential}\label{sec:Examples of Potential functions}
In this section, we calculate abound for $b$ explicitly for the cosine potential. Let
\be\label{def: cosine pot}
 V_0(x)=A_c \cos^{2n}\left( \frac{x}{2} \right),
\ee
where $n\in\nn^+$ is a parameter which determines how steep the potential is and $A_c$ is a to be determined normalisation constant. In order to proceed, we use the identity
\be 
\cos^{2n}(x/2)=\frac{1}{2^{2n}} (\me^{ix/2}+\me^{-ix/2})^{2n}= \frac{1}{2^{2n}} \sum_{k=0}^{2n} \binom{2n}{k} \me^{ix k/2}\me^{-ix(2n-k)/2}=\frac{1}{2^{2n}} \sum_{k=0}^{2n} \binom{2n}{k} \me^{ix (k-n)}.
\ee
Therefore, from Eq. \eqref{eq:theta normalisation}, it follows
\be 
\Omega=A_c \int_0^{2\pi}dy \cos^{2n}\left( \frac{y}{2} \right)=\frac{A_c}{2^{2n}} \sum_{k=0}^{2n} \binom{2n}{k}\int_0^{2\pi} \me^{ix (k-n)}=\frac{A_c}{2^{2n}} \sum_{k=0}^{2n} \binom{2n}{k}2\pi \delta_{k,n}= \frac{A_c}{2^{2n}} \binom{2n}{n} 2\pi,
\ee
giving 
\be\label{def: A c eq def for cos pot} 
A_c=\Omega \frac{2^{2 n}}{2\pi \binom{2n}{n}}=\Omega\frac{2 ^{2 n}}{2\pi}\frac{(n!)^2}{(2n)!}.
\ee
Thus
\begin{align} 
\left|  V_0^{(q)}(x)\right|&=\left| \frac{\Omega}{2\pi}\frac{(n!)^2}{(2n)!}\sum_{k=0}^{2n}\binom{2n}{k}\left(\mi(k-n)\right)^q\me^{i x(k-n)}  \right|\leq \frac{|\Omega|}{2\pi}\frac{(n!)^2}{(2n)!}\sum_{k=0}^{2n}\binom{2n}{k}|k-n|^q\leq  \frac{|\Omega|}{2\pi}\frac{(n!)^2}{(2n)!}\sum_{k=0}^{2n}\binom{2n}{k}n^q\label{eq:V devs up bound, inter line} \\
&=\frac{|\Omega|}{2\pi}\frac{(n!)^2}{(2n)!} 4^n n^q\leq \frac{\me^2}{(2\pi)^{3/2}}\frac{|\Omega|}{\sqrt{2}} n^{q+3/2},\quad q\in \nn^0,
\end{align}
Where we have used Stirling's approximation to  bound the factorials. Note that one can improve upon this bound by explicity calculating in closed form the penultimate equality in line \eqref{eq:V devs up bound, inter line}. Thus we can upper bound the r.h.s. of Eq. \eqref{eq: b_1 explicit}, to achieve
\be
\sup_{k\in \nn^+} \left(2\max_{x\in[0,2\pi]} \left|  V_0^{(k-1)}(x)\right|\,\right)^{1/k}\leq  \sup_{k\in \nn^+} \left(\frac{2\,\me^2}{(2\pi)^{3/2}}\frac{|\Omega|}{\sqrt{2}} n^{k+1/2} \right)^{1/k}.
\ee 
Thus using the formula
\be 
\frac{d}{dk} c^{1/k}=\frac{c^{1/k}}{k^2}\ln(1/c)
\ee
and taking into account Eq. \eqref{eq: b_1 explicit}, we can conclude a value for $b$, namely
\be\label{eq:b final for cos example}
b= 
\begin{cases}
 n \left(\frac{2\,\me^2}{(2\pi)^{3/2}}\frac{|\Omega|}{\sqrt{2}} n^{1/2} \right)^{1/k} \bigg{|}_{k=1} =\frac{2\,\me^2}{(2\pi)^{3/2}}\frac{|\Omega|}{\sqrt{2}}\,n\sqrt{n}  &\mbox{if } |\Omega| \sqrt{n} \geq \frac{(2\pi)^{3/2}}{\sqrt{2}\me^2}\approx 1.507\\
n \lim_{k\rightarrow\infty} \left(\frac{2\,\me^2}{(2\pi)^{3/2}}\frac{|\Omega|}{\sqrt{2}} n^{1/2} \right)^{1/k}=n &\mbox{if } |\Omega| \sqrt{n} \leq \frac{(2\pi)^{3/2}}{\sqrt{2}\me^2}.
\end{cases}
\ee 

\section{Clocks as Quantum control}\label{sec:Clocks as Quantum control}
This section will be concerned with the implementation of energy preserving unitaries on a finite dimensional quantum system with Hilbert space $\mathcal{H}_s$ via a control system. It will be in this section that the importance of Theorem \ref{movig through finite time} for controlling a quantum system will become apparent. For this purpose it will suffice that $V_0:\rr\rightarrow\rr$, so that $\hat V_d$ is self adjoint. This special case will be assumed throughout this section. The first section will consider the case that the unitary's implementation is via a time dependent Hamiltonian (Section\ref{sec:Implementing Energy preserving unitaries via a time dependent Hamiltonian}) followed by showing that the idealised clock can, via a time independent Hamiltonian, perfectly implement controlled energy preserving unitaries (Section \ref{sec:Automation via the idealised clock}). Finally, these previous two sections will serve as an introduction and provide context for the first important section about clocks as quantum control, namely how well (as a function of the clock dimension) can energy preserving unitaries be implemented via finite dimensional \gClocks~(Section \ref{sec:Implementing Energy preserving unitaries with the finite}). To finalize this discussion, we will then outline how one could perform non energy preserving unitaries using the finite dimensional clock in Section \ref{sec:non energy preserving unitaries}. 
To conclude, in Section \ref{Clock Fidelity}, we will bound how the implementation of the unitary on the system via the finite dimensional clock, disturbs the dynamics of the clock.

\subsection{Implementing Energy preserving unitaries via a time dependent Hamiltonian}\label{sec:Implementing Energy preserving unitaries via a time dependent Hamiltonian}

Consider a system of dimension $d_s$, that begins in a state $\rho_s\in\mathcal{S}(\mathcal{H}_s)$, and upon which one wishes to perform the unitary $U_s$ over a time interval $[t_i,t_f]$. Energy-preserving means that $[U_s,\hat H_s]=0$, where $H_s$ is the time-independent finite dimensional Hamiltonian of the system.

The first step in automation is to convert the unitary into a time-dependent interaction, via
\begin{equation}\label{unitaryasinteraction}
	U_s = e^{-\mi \hat H^{int}_s},
\end{equation}
where $\hat H^{int}_s$ also commutes with $\hat H_s$. The unitary can thus be implemented by the addition of $\hat H^{int}_s$ as a time-dependent interaction
\begin{equation}\label{eq:system Ham t depen}
	\hat H = \hat H_s + \hat H^{int}_s \cdot g(t).
\end{equation}
If $g\in L(\rr: \rr_{\geq 0})$ is a normalized pulse, i.e.
\begin{align}\label{eq:g(t) pulse normalisation}
	\int_{t_i}^{t_f}dt g(t) &= 1,
\end{align}
with support interval $[t_i,t_f]$, then it is easily verified that the unitary $U$ will be implemented between $t_i$ and $t_f$, i.e. given the initial system state $\rho_s\in\mathcal{S}(\mathcal{H}_s)$, the state at the time $t$ is
\begin{equation}\label{eq:rho t scho sol commute t depen}
\rho_s(t) = \me^{-\mi t\hat H_s} \me^{-\mi\hat H^{int}_s\int_{t_i}^t dx\, g(x)} \rho_s\; \me^{+\mi \hat H^{int}_s \int_{t_i}^t dx\, g(x)}  \me^{+\mi t\hat H_s}\\
= \me^{-\mi t\hat H_s-\mi\hat H^{int}_s\int_{t_i}^t dx\, g(x)} \rho_s\; \me^{+\mi t \hat H_s+\mi \hat H^{int}_s \int_{t_i}^t dx\, g(x)}.
\end{equation}
To gain a deeper understanding, we first note that since $\hat H_s$ and $\hat H_s^{int}$ are Hermitian and commute, there exists a mutually orthonormal basis, denoted by $\{\ket{\phi_j}\}_{j=1}^{d_s}$, such that
\begin{align}
	\hat H_s &= \sum_{j=1}^{d_s} E_j \ketbra{\phi_j}{\phi_j}_s, \\
	\hat H^{int}_s &= \sum_{j=1}^{d_s} \Omega_j \ketbra{\phi_j}{\phi_j}_s,
\end{align}
where, due to Eq. \eqref{unitaryasinteraction}, without loss of generality, we can confine $\Omega_j\in[-\pi,\pi)$, for $j=1,2,3,\ldots,d_s$. By writing the evolution of the free Hamiltonian in the $\{\ket{\phi_j}\}_{j=1}^{d_s}$ basis
\be\label{eq:rho s in mutual alphog basis}
\rho(t)=\me^{-\mi t \hat H_s}\rho_s \,\me^{\mi t \hat H_s}=\sum_{m,n=1}^{d_s}\rho_{m,n}(t) \ketbra{\phi_m}{\phi_n},
\ee
Eq. \eqref{eq:rho t scho sol commute t depen} can be written in the form
\be \label{eq:rho as a function of t}
\rho_s(t)=\sum_{m,n=1}^{d_s}\rho_{m,n}(t)\,\me^{-\mi (\Omega_m-\Omega_n) \int_{t_i}^t dt\, g(t)} \ketbra{\phi_m}{\phi_n}.
\ee

\subsection{Implementing Energy preserving unitaries with the idealised clock}\label{sec:Automation via the idealised clock}
To remove the explicit time-dependence of Eq. \eqref{eq:system Ham t depen}, we insert a clock, and replace the background time parameter $t$ by the time degree of the clock. For the idealised clock, this corresponds to a Hamiltonian on $\mathcal{H}=\mathcal{H}_s\otimes\mathcal{H}_c$, given by
\begin{equation}\label{eq:idealised s c ham}
	\hat H_{total}^{id} = \hat H_s \otimes \id_c + \id_s \otimes \hat{p}_c + \hat H^{int}_s \otimes g(\hat{x}_c),
\end{equation}
where $\hat x_c$ and $\hat p_c$ are the canonically conjugate position and momentum operators of a free particle in one dimension detailed in Section \ref{idealizedclock}.
To verify that this Hamiltonian can indeed implement the unitary, we let the initial state of the system and clock be in a product form,
\be 
\rho_{sc}^{id}=\rho_s\otimes\ketbra{\Psi}{\Psi}_c.
\ee
Writing the state $\rho_{sc}^{id}$ at a later time in the $\{\ket{\phi_j}\}_{j=1}^{d_s}$ basis, we find
\be\label{eq:rho ideal s c}
\rho_{sc}^{id}(t)=\me^{-\mi t \hat H_{total}^{id}}\rho_{sc}^{id}\, \me^{\mi t \hat H_{total}^{id}}= \sum_{m,n=1}^{d_s} \rho_{m,n}(t)\otimes\ketbra{\Phi_m(t)}{\Phi_n(t)}_c,
\ee
with
\be 
\ket{\Phi_n(t)}_c:=\me^{-\mi t( \hat{p}_c + \Omega_n g(\hat{x}_c))}\ket{\Phi}_c.
\ee
\begin{lemma}\label{lem:doing unitary id}
Let $\psi(x,0)\in D_0$ be the normalised wave-function with support on the interval $x\in[x_{\psi l},x_{\psi r}]$ associated with the state $\ket{\Phi}_c$. Let $g(x)$ have support on $x\in[x_{g l},x_{g r}]$ and denote $\rho_s^{id}(t)=\tr_c[\rho_{sc}^{id}(t)]$ as the partial trace over  $\mathcal{H}_c$, with $\rho_{sc}^{id}(t)$ given by Eq. \eqref{eq:rho ideal s c}. Furthermore, assume that initally the wave packet is on the left of the potential function, namely $x_{\psi r}\leq x_{g l}$.
Then the unitary is implemented perfectly in the time interval $x_{g l}-x_{\psi r} \leq t\leq x_{g r}-x_{\psi l}$, more precisely
\be 
\rho_s^{id}(t)=
\begin{cases}
\sum_{m,n=1}^{d_s}\rho_{m,n}(t) \ketbra{\phi_m}{\phi_n} &\mbox{if }\; t\leq x_{g l}-x_{\psi r}\\ 
\sum_{m,n=1}^{d_s}\rho_{m,n}(t) \ketbra{\phi_m}{\phi_n} \,\me^{-\mi (\Omega_m-\Omega_n)} &\mbox{if } \, t\geq x_{g r}-x_{\psi l}
\end{cases}
\ee 
\end{lemma}
\emph{Intuition.} If one sets $t_i=x_{g l}-x_{\psi r}$, $t_f=x_{g r}-x_{\psi l}$, then the clock can control perfectly the system and mimics the time dependent Hamiltonian of  Eq. \eqref{eq:system Ham t depen}. The Lemma is a consequence of the property noted in Eq. \eqref{eq:psi idealised control}, and hence the name given to dynamics of this form.
\begin{proof}
Taking the partial trace over the idealized clock in Eq. \eqref{eq:rho ideal s c}, we achieve
\be\label{eq:rho s id 1} 
\rho_{s}(t)= \sum_{m,n=1}^{d_s} \rho_{m,n}(t)\braket{\Phi_n(t)|\Psi_m(t)}_c,
\ee
where using Eq. \eqref{eq:psi idealised control}, we find
\begin{align}
\braket{\Phi_n(t)|\Psi_m(t)}_c &=\int_{-\infty}^\infty dx\braket{\Phi_n(t)|x}_{\!\!c}\!\! \braket{x|\Psi_m(t)}
=\int_{x_{\psi l}}^{x_{\psi r}} dx |\psi(x,0)|^2 \,\me^{-\mi (\Omega_m-\Omega_n)\int_x^{x+t}g(x')dx'}.\label{eq:inner product n m}
\end{align}
Note that 
\be
\int_x^{x+t} g(x')dx'= 
\begin{cases}\label{eq:cases conditioons in idealisec control lemma}
0 &\mbox{if } x+t\leq x_{gl}\\
1 &\mbox{if } x\leq x_{vl} \text{ and } x+t\geq x_{vr}
\end{cases}
\ee
Thus in order for $\exp({-\mi (\Omega_m-\Omega_n)\int_x^{x+t}g(x')dx'})$ to be a constant factor in the integral of Eq. \eqref{eq:inner product n m}, we need   the conditions on $x$ in Eq. \eqref{eq:cases conditioons in idealisec control lemma} to hold for all $x\in[x_{\psi l},x_{\psi r}]$. Noting the normalization of the wave function, we conclude the Lemma.
\end{proof}

One might also wonder whether or not the state of the clock is degraded in any way due to the action of performing the unitary.  The wave-function only changes by a global phase, and thus the answer to this question is no. This means that one can implement any number of energy preserving unitaries $\{U_n\}_n$ on the system in any sequence of time intervals. This of course, is under the physically questionable assumption of the validity of the idealized clock.  

\subsection{Implementing Energy preserving unitaries with the finite dimensional \gClock}\label{sec:Implementing Energy preserving unitaries with the finite}
In this section we will see how well the finite dimensional \gClock~can mimic the behavior of the idealized clock studied  in Section \ref{sec:Automation via the idealised clock}. We will find that Theorem \ref{movig through finite time} implies that the clock can perform on the system any sequence of unitaries with an additive error which grows linearly in time and decays faster than any polynomial in clock dimension (the exact error will depend on a number of things, for example which initial \gClock~state is used).\\

In analogy with the idealized clock-system Hamiltonian, Eq. \eqref{eq:idealised s c ham}, we define
\begin{equation}\label{eq:finite s c ham}
	\hat H_{sc} = \hat H_s \otimes \id_c + \id_s \otimes \hat{H}_c + \hat H^{int}_s \otimes \hat{V}_d,
\end{equation}
with $\Omega=1$, where $\hat H_c$ is defined in Eq. \eqref{finiteHamiltonian} and $\Omega$, $\hat{V}_d$ are defined in Eq. \eqref{eq: Theta def}. Furthermore, we denote the system and clock evolution under $\hat H_{sc}$ as
\be\label{eq:rho finite s c}
\rho_{sc}'(t)=\me^{-\mi t \hat H_{sc}}\rho_s\otimes\ketbra{\Psi_\textup{nor}(k_0)}{\Psi_\textup{nor}(k_0)}_c\, \me^{\mi t \hat H_{sc}}.
\ee
Recall that the clock has period $T_0$ and aims to implement the potential $V_0$ while the time dependent Hamiltonian aims to implement the potential $g$ over an interval $t\in[t_i,t_f]$. In order to make a fair comparison we will set
\be\label{eq:g in terms of V_0}
g(x)=\frac{2\pi}{T_0} V_0(x\; 2\pi/T_0),\quad [t_i,t_f]=[0,T_0],
\ee
 which taking into account Eq. \eqref{eq:theta normalisation}, implies that Eq. \eqref{eq:g(t) pulse normalisation} is satisfied. Note that while Eq. \eqref{eq:g in terms of V_0} is a natural definition of $g$, it differs from the definition used in Section \ref{sec:Consequences of Quasi-Autonomous control}. We will remedy this in Corollary \ref{lem:new g}.
\begin{lemma}[Implicit form]\label{lem:trace dist bound t-Ham Vs d dim clock}
Let $\rho_s'(t):=\tr_c[\rho_{sc}'(t)]$, then for all $\rho_s(0)\in\mathcal{S}(\mathcal{H}_s)$, and $t\in[0,T_0]$, the trace distance between $\rho_s'(t)$ and the idealised case of controlling the system with a time dependent Hamiltonian, $\rho_s(t)$ (see Eq. \eqref{eq:rho t scho sol commute t depen}; $g$ given by Eq. \eqref{eq:g in terms of V_0}), is bounded by
\be\label{eq:trace dist rho rho finite and time}
\|\rho_s(t)-\rho_s'(t)\|_1\leq \sqrt{d_s \tr[\rho_s^2(0)]}\left(2\varepsilon_v(t,d)+\varepsilon_v^2(t,d)+\epsilon_V(t,d) \right),
\ee
where $\varepsilon_v(t,d)$ is evaluated for $\Omega=\pi$ and defined by Eq. \eqref{eq:ltwo epsilon nor control theorem} while\footnote{In the following summation, and for now on, we use the convention $\sum_{k=a}^b f(k)= f(a)+f(a+1)+\ldots+f(b)$, where $b-a\in\nn^0$}
\be\label{eq:def epsilon k}
\epsilon_V(t,d):=\max_{\bar\kappa\in[0,1]} \sum_{k=k_0-d/2+\bar\kappa}^{k_0+d/2+\bar\kappa}\left( \epsilon_k^2+\epsilon_k\right)|\psi_\textup{nor}(k_0;k)|^2,\quad\quad \epsilon_k=2\pi\left| \int_{0}^{t\, 2\pi/T_0} dx\big( V_0(x)-V_0(x+2\pi k/d)\big)\right|.
\ee
\end{lemma}
\emph{Intuition}. The point of this Lemma is to make it manifest that the error invoked by using the finite dimensional \gClock~to control the quantum system on $\mathcal{H}_s$ (specifically, in this case, to perform an energy conserving unitary during a specific moment in time) is small. The two error terms on the r.h.s. of Eq. \eqref{eq:trace dist rho rho finite and time} involving $\varepsilon_v(t,d)$ decay fast with $d$ (see Eq. \eqref{eq:ltwo epsilon nor control theorem}) and only grow linearly in time. Note that $\epsilon_V(j\, T_0,d)=0$ for all $d\in \nn^+$, $j\in\nn^0$. More generally, the error term $\epsilon_V$, will be small for a clock initially centered at $k_0$, at time $t$ when $V_0(2\pi(t/T_0+k_0/d))\approx V_0(2\pi t/T_0)$, i.e. is relatively flat. Roughly speaking, this will be true for times $t$ before and after the unitary has been implemented as the example in Fig. \ref{fig:epsilon V exmaple plot} and the following lemma will make explicit. With a small modification, Lemma \ref{lem:trace dist bound t-Ham Vs d dim clock} still holds if we do not  assume the special choice of $g$ in Eq. \eqref{eq:g in terms of V_0}. For general $g$, we simply replace $\epsilon_k$ in Eq. \eqref{eq:def epsilon k} with Eq. \eqref{eq:epsilon k lemma}. The system dependency in Eq. \eqref{eq:trace dist rho rho finite and time} enters via the factor $\sqrt{d_s \tr[\rho_s^2(0)]}$ which varies between $\sqrt{d_s}$ for pure initial system states, and unity in the case of maximally mixed initial system states. The latter is what one should expect, since clearly an initially maximally mixed state is invariant under the unitary dynamics and as such the error associated with implementing a unitary should not depend on its dimension.
\begin{proof}
Proceeding similarly to the derivation of Eq. \eqref{eq:rho s id 1} in the proof of Lemma \ref{lem:doing unitary id}, we achieve
\be 
\rho_{s}(t)= \sum_{m,n=1}^{d_s} \rho_{m,n}(t)\braket{\bar{\Phi}_n(t)|\bar{\Phi}_m(t)},
\ee
with
\be\label{eq:Psi bar def eq}  
\ket{\bar{\Phi}_l(t)}:=\me^{-\mi t( \hat{H}_c + \Omega_l \hat{V}_d)}\ket{\Psi_\textup{nor}(k_0)}, \quad l=n,m.
\ee
Thus noting that the Frobenius norm, $\|\cdot\|_F$, is an upper bound to the trace distance $\|\cdot\|_1/\sqrt{d_s}$, for a $d_s$ dimensional system, followed by using Eq. \eqref{eq:rho as a function of t},
\begin{align}
\|\rho_s(t)-\rho_s'(t)\|_1&
\leq \sqrt{d_s} \abs{\abs{\rho(t)-\rho'(t)}}_F=\sqrt{d_s}\;\sqrt{\sum_{n,m}|\rho_{nm}(t)|^2\left| e^{-\mi(\Omega_m-\Omega_n)\int_{t_i}^{t_f}g(x)dx}-\braket{\bar{\Phi}_n(t)|\bar{\Phi}_m(t)}\right|^2}\\	
	& \leq \sqrt{d_s}\;\sqrt{\sum_{n,m}|\rho_{nm}(t)|^2 \max_{q,r}\left\{\left| e^{-\mi(\Omega_q-\Omega_r)\int_{t_i}^{t_f}g(x)dx)}-\braket{\bar{\Phi}_r(t)|\bar{\Phi}_q(t)} 
	\right|^2\right\}} \\
	& \leq \sqrt{d_s\tr\left(\rho^2_s(0)\right)} \max_{q,r}\left\{\left| e^{-\mi(\Omega_q-\Omega_r)\int_{t_i}^{t_f}g(x)dx}-\braket{\bar{\Phi}_r(t)|\bar{\Phi}_q(t)} 
	\right|\right\}.\label{eq: rho - rho 1}
\end{align}
Applying Theorem \ref{movig through finite time},
\be 
\braket{\bar{\Phi}_n(t)|\bar{\Phi}_m(t)}=\sum_{k\in\mathcal{S}(k_0+dt/T_0)} \me^{-\mi (\Omega_n-\Omega_m)\int_{k-td/T_0}^k dy V_d(y)} \left| \psi_\textup{nor}(k_0;k-td/T_0)\right|^2+\tilde\epsilon_{n,m},
\ee
where
\begin{align}\label{eq:tilde epsilon up bound}
\tilde\epsilon_{n,m}&=\braket{\bar\Psi_\textup{nor}(k_0+td/T_0)|\epsilon}+\braket{\epsilon|\bar\Psi_\textup{nor}(k_0+td/T_0)}+ \braket{\epsilon|\epsilon},\\
|\tilde\epsilon_{n,m}|&\leq 2\varepsilon_v(t,d)+\varepsilon_v^2(t,d),\quad \forall n,m\,\label{eq:up bund n m independ ep tilde}
\end{align}
with $\varepsilon_v(t,d)$ is given by Eq. \eqref{eq:ltwo epsilon nor control theorem}. Note that the r.h.s of Eq. \eqref{eq:tilde epsilon up bound}, does depend on indices  $n$ and $m$. The r.h.s. of Eq. \eqref{eq:up bund n m independ ep tilde} on the other hand, is independent of $n$ and $m$.\\
Thus applying Eq. \eqref{eq:up bund n m independ ep tilde} to Eq. \eqref{eq: rho - rho 1}, we achieve
\begin{align}
\frac{\|\rho_s(t)-\rho_s'(t)\|_1}{ \sqrt{d_s\tr[\rho^2_s(0)]}}
 \leq & \max_{q,r}\left\{\left| 1-\sum_{k\in\mathcal{S}(k_0+dt/T_0)} \me^{\mi (\Omega_q-\Omega_r)\left(\int_{t_i}^{t_f}g(x)dx-\int_{k-td/T_0}^k dy V_d(y)\right)} \left| \psi_\textup{nor}(k_0;k-td/T_0)\right|^2
	\right|\right\}\\
	&+2\varepsilon_v(t,d)+\varepsilon_v^2(t,d)\\
	 \leq & \max_{q\geq r}\left\{\left| 1-\sum_{k\in\mathcal{S}(k_0+dt/T_0)} \me^{\pm\mi |\Omega_q-\Omega_r|\left(\int_{t_i}^{t_f}g(x)dx-\int_{k-td/T_0}^k dy V_d(y)\right)} \left| \psi_\textup{nor}(k_0;k-td/T_0)\right|^2
	\right|\right\}\\
	&+2\ \varepsilon_v(t,d)+\varepsilon_v^2(t,d)\\
	=&\max_{q,r}\left\{ \rule{0cm}{1cm}\right.\sqrt{\left(1-\sum_{k\in\mathcal{S}(k_0+dt/T_0)} \cos(\pm g_k^{q,r})\left| \psi_\textup{nor}(k_0;k-td/T_0)\right|^2 \right)^2+}\\
	&\,\,\quad\quad\,\,\,\,\,\,\,\,\overline{\left(\sum_{k\in\mathcal{S}(k_0+dt/T_0)}\sin(\pm g_k^{q,r})\left| \psi_\textup{nor}(k_0;k-td/T_0)\right|^2\right)^2}\left. \rule{0cm}{1cm}\right\}+2\varepsilon_v(t,d)+\varepsilon_v^2(t,d), 
\end{align}
where
\be\label{eq:gk def in lemm}
g_k^{q,r}:= |\Omega_q-\Omega_r|\left(\int_{t_i}^{t_f}g(x)dx- \int_{k-td/T_0}^k V_d(x)dx\right)\leq 2\pi\left|\int_{t_i}^{t_f}g(x)dx- \int_{k-td/T_0}^k V_d(x)dx\right|=:g_k \quad\forall\, q,r.
\ee
Thus using the identities $(|a|+|b|)^2\geq |a|^2+|b|^2$, $\sin(a)\leq |a|$, $-\cos(a)\leq a^2-1$, for $a,b\in\rr$, to achieve
\begin{align}
\frac{\|\rho_s(t)-\rho_s'(t)\|_1}{ \sqrt{d_s\tr[\rho^2_s(0)]}}
 \leq &\,2\varepsilon_v(t,d)+\varepsilon_v^2(t,d)+\sum_{k\in\mathcal{S}(k_0+dt/T_0)} (g_k^2+g_k)|\psi(k_0;k-td/T_0)|^2.
\end{align}
Performing the change of variable $y=k-td/T_0$, we find
\begin{align}
\frac{\|\rho_s(t)-\rho_s'(t)\|_1}{ \sqrt{d_s\tr[\rho^2_s(0)]}}
 \leq &\,2\varepsilon_v(t,d)+\varepsilon_v^2(t,d)+\sum_{y=\min\{\mathcal{S}(k_0+dt/T_0)\}-td/T_0}^{\max\{\mathcal{S}(k_0+dt/T_0)\}-td/T_0} (\epsilon_y^2+\epsilon_y)|\psi(k_0;y)|^2\\
 \leq &\, 2\varepsilon_v(t,d)+\varepsilon_v^2(t,d)+\max_{\bar\kappa\in[0,1]}\sum_{y=k_0-d/2+\bar\kappa}^{k_0+d/2+\bar\kappa} (\epsilon_y^2+\epsilon_y)|\psi(k_0;y)|^2,
\end{align}
where 
\be\label{eq:epsilon k lemma}
\epsilon_k=2\pi\left| \int_{t_i}^t g(x)dx-\int_{k}^{k+td/T_0}V_d(x)dx\right|.
\ee
and to achieve the last line, we have used the inequalities
\begin{align}
k_0-d/2+\bar\kappa &\leq \min\{\mathcal{S}(k_0+dt/T_0)\}-td/T_0,\\
 k_0+d/2+\bar\kappa &\geq \max\{\mathcal{S}(k_0+dt/T_0)\}-td/T_0,\\
  k_0+d/2+\bar\kappa &-(k_0-d/2+\bar\kappa)=d.
\end{align}
Finally, taking into account Eq. \eqref{eq:g in terms of V_0} followed by a change of variable, Eq. \eqref{eq:epsilon k lemma} can be written as
\be 
\epsilon_k=2\pi\left| \int_{0}^{t 2\pi/T_0}  V_0(x)dx-\int_{k}^{k+td/T_0}V_d(x)dx\right|=2\pi\left| \int_{0}^{t\, 2\pi/T_0} dx\big( V_0(x)-V_0(x+2\pi k/d)\big)\right|.
\ee
\end{proof}

\begin{figure}[!htb]
\minipage{0.32\textwidth}
  \includegraphics[width=\linewidth]{{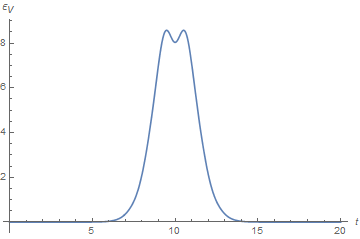}}
\endminipage\hfill
\minipage{0.32\textwidth}
  \includegraphics[width=\linewidth]{{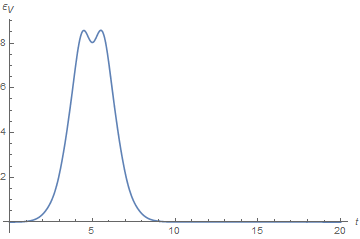}}
\endminipage\hfill
\minipage{0.32\textwidth}%
  \includegraphics[width=\linewidth]{{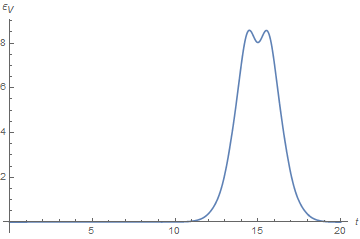}}
\endminipage
 \caption{{\bf Left}: Plot of $\epsilon_V$ as a function of time over one clock period $T_0$ with cosine potential $V_0(x)=A \cos^{2 n}((x-x_0)/2)$ from Section \ref{sec:Examples of Potential functions} with parameters $k_0=0$, $\sigma=\sqrt{d}$, $T_0=20$, $d=20$, $n=60$, $x_0=\pi$. The unitary is to be implemented at a time $t\approx x_0 \,T_0/(2\pi)=T_0/2$ and the complex Gaussian superposition has its peak initially on the L.H.S. of the potential peak. Before the unitary is applied, to a time $t\approx 5$ the error is practically zero, during a period $t\approx 5$ to $t\approx 15$ the unitary is being applied and the error is large, but after $t\approx 15$, the unitary has been applied and the error is practically zero again. As one increases $n$ and $d$, the errors become arbitrarily small and the interval $t\in[5,15]$ in which there is a large error in the above plot becomes arbitrarily small. {\bf Middle and Right}: Same as Left but changing $x_0$ to $\pi/2$ and $3 \pi /2$ respectiverly, to illustrate how the same qualitative behaviour holds when the unitary is implemented at earlier and later times. \label{fig:epsilon V exmaple plot}}
\end{figure}

\begin{corollary}[Explicit form]\label{corr:explicit bound on system trace distance}
Let $V_0(x)\geq 0$ and have a unique global maximum in the interval $x\in[0,2\pi]$ at $x=x_0$ and $g$ be given by Eq. \eqref{eq:g in terms of V_0}. Let $\tilde\epsilon_V$ be such that
\be\label{eq:tilde ep V def}
1-\tilde\epsilon_V=\int_{x_{vl}}^{x_{vr}} dx\, V_0(x+x_0)
\ee
for some $-\pi\leq x_{vl}< 0< x_{vr}\leq \pi$, where, for simplicity, we assume $x_{vl}=- x_{vr}$. Let the initial \gClock~state be centered at zero; $k_0=0$, and introduce the parameter $0< \gamma_\psi\leq 1$. Furthermore, let $x_0$ be such that $x_{vr}+\pi \gamma_\psi  \leq x_0\leq 2\pi-x_{vr}-\pi \gamma_\psi$. Then, for all times $t$ satisfying 
\be\label{eq:t before uni applied}
\begin{split}
0&\leq t \,2 \pi/T_0\leq x_0-x_{vr}-\pi \gamma_\psi
\end{split}
\ee
or
\be\label{eq:t after uni applied}
\begin{split}
x_0+x_{vr}+\pi \gamma_\psi &\leq t \,2 \pi/T_0\leq 2\pi+x_0 - x_{vr}-\pi\gamma_\psi,
\end{split}
\ee
$\epsilon_V$ defined in Lemma \ref{lem:trace dist bound t-Ham Vs d dim clock}, is bounded by
\be\label{eq:epsilon V upper bound explicit version} 
\epsilon_V\leq 4\pi \left( (1+2\pi)A^2\frac{\me^{-2\pi\tilde\kappa\frac{d^2}{\sigma^2}}}{1-\me^{-4\pi\sqrt{\tilde\kappa}\,d/\sigma^2}}+2\tilde\epsilon_V (1+8\pi\tilde\epsilon_V)\right),\quad \tilde\kappa=
\begin{cases}
0 &\mbox{if } d\gamma_\psi/2\leq 1\\
(\gamma_\psi/2-1/d)^2 &\mbox{otherwise}
\end{cases} 
\ee
where $A$ is given by Eq. \eqref{eq:A normalised} and is a decreasing function of $\sigma$ (see  Section \ref{Normalizing the clock state} for bounds).
\end{corollary}
\emph{Intuition}. If the potential is very peaked around $x_0$ and close to zero away from $x_0$, then one can choose $x_{vl}\approx x_{vr}\approx 0$ with $0 \leq \tilde \epsilon_V\ll 1$. Thus by choosing $\gamma_\psi\approx 0$, the intervals permitted for $x_0$ and $t$ by Eqs. \eqref{eq:t before uni applied} and Eqs. \eqref{eq:t after uni applied}, are almost the full range $[0,2\pi]$ and $[0,T_0]$ respectively, and only for times $t\approx x_0 T_0/(2\pi)$ will the error be large.  Eqs. \eqref{eq:t before uni applied} correspond to the times \textit{before} the unitary is applied, while  Eqs. \eqref{eq:t after uni applied} correspond to the times \textit{after} the unitary is applied.
\begin{remark}
The parameters $\gamma_\psi, x_{vr}$ can be chosen to depend on $t$ and/or $d$ in order to improve bounds if desired.
\end{remark}
\begin{proof}
The proof consists in writing the sum in Eq. \eqref{eq:def epsilon k} in terms of three contributions; the left Gaussian tail of $|\psi_\textup{nor}(0;k)|$, the right Gaussian tail of $|\psi_\textup{nor}(0;k)|$, and the contribution from the center of the Gaussian $|\psi_\textup{nor}(0;k)|$. These terms are then all individually bounded.\\
Noting that $0\leq\epsilon_k\leq 2\pi $ and that $|\psi_\textup{nor}(0;k)|$ is normalized, we have from Eq. \eqref{eq:def epsilon k}
\begin{align}\label{eq:ep V up bound coll proof}
\epsilon_V(t,d)\leq \max_{\bar\kappa\in[0,1]}\Bigg(&  \sum_{k=-d/2+\bar\kappa}^{-d\gamma_\psi/2+\bar\kappa}\left( \epsilon_k^2+\epsilon_k\right)|\psi_\textup{nor}(k_0;k)|^2+\sum_{k=d\gamma_\psi /2+\bar\kappa}^{d/2+\bar\kappa}\left( \epsilon_k^2+\epsilon_k\right)|\psi_\textup{nor}(k_0;k)|^2\\
&+\sum_{k=-d\gamma_\psi /2+1+\bar\kappa}^{d\gamma_\psi/2-1+\bar\kappa}\left( \epsilon_k^2+\epsilon_k\right)|\psi_\textup{nor}(k_0;k)|^2 \Bigg)\\
\leq  2\pi(1+2&\pi)\max_{\bar\kappa\in[0,1]}  \bigg(\sum_{k=-d/2+\bar\kappa}^{-d\gamma_\psi/2+\bar\kappa}|\psi_\textup{nor}(k_0;k)|^2+\sum_{k=d\gamma_\psi /2+\bar\kappa}^{d/2+\bar\kappa}|\psi_\textup{nor}(k_0;k)|^2\bigg)\label{eq:first line epsilon bound proof}\\
+ \max_{\bar\kappa\in[0,1]}&\max_{k\in\mathcal{I}_{\gamma_\psi}}\left( \epsilon_k^2+\epsilon_k\right),\label{eq:last line epsilon bound proof} 
\end{align}
where $\mathcal{I}_{\gamma_\psi}:=\{-d\gamma_\psi /2+1+\bar\kappa,\,-d\gamma_\psi /2+2+\bar\kappa,\ldots,d\gamma_\psi/2-1+\bar\kappa  \}$. \\
First we will bound the terms in line \ref{eq:first line epsilon bound proof}. By a change of variable, followed by using Eq. \eqref{eq:psi nor def}, we find
\begin{align}
&  2\pi(1+2\pi)\max_{\bar\kappa\in[0,1]}\bigg(\sum_{k=-d/2+\bar\kappa}^{-d\gamma_\psi/2+\bar\kappa}|\psi_\textup{nor}(k_0;k)|^2+\sum_{k=d\gamma_\psi /2+\bar\kappa}^{d/2+\bar\kappa}|\psi_\textup{nor}(k_0;k)|^2\bigg)\\ 
&= 2\pi(1+2\pi)A^2\max_{\bar\kappa\in[0,1]}\bigg(\sum_{x=0}^{d(1+\gamma_\psi)/2}\me^{-\frac{2\pi}{\sigma^2}(x-d\gamma_\psi/2+\bar\kappa)^2}+\sum_{x=0}^{d(1-\gamma_\psi)/2}\me^{-\frac{2\pi}{\sigma^2}(x+d\gamma_\psi/2+\bar\kappa)^2}\bigg)\\ 
&\leq 2\pi(1+2\pi)A^2\max_{\bar\kappa\in[0,1]}\bigg(\frac{\me^{-\frac{2\pi}{\sigma^2}(d\gamma_\psi/2-\bar\kappa)^2}}{1-\me^{-4\pi|d\gamma_\psi/2-\bar\kappa|/\sigma^2}}+\frac{\me^{-\frac{2\pi}{\sigma^2}(d\gamma_\psi/2+\bar\kappa)^2}}{1-\me^{-4\pi|d\gamma_\psi/2+\bar\kappa|/\sigma^2}}\bigg)\label{eq:used gaus tail bound here}\\
&\leq 4\pi(1+2\pi)A^2\frac{\me^{-2\pi\tilde\kappa\frac{d^2}{\sigma^2}}}{1-\me^{-4\pi\sqrt{\tilde\kappa}\,d/\sigma^2}},
\end{align}
with $\tilde \kappa:=\min_{\bar\kappa\in[0,1]} \{ (\gamma_\psi/2-\bar\kappa/d)^2, (\gamma_\psi/2+\bar\kappa/d)^2\}=\min_{\bar\kappa\in[0,1]}(\gamma_\psi/2-\bar\kappa/d)^2$ and where in line \ref{eq:used gaus tail bound here}, the bound from Lemma \ref{G0}  has been used.\\
Now we will bound the term in line \ref{eq:last line epsilon bound proof}. By the change of variables $y=x-x_0$, $y'=x+2\pi k/d-x_0$, Eq. \eqref{eq:def epsilon k} can be written in the form
\be\label{eq:epsilon k intermediate}
\epsilon_k=2\pi\left| \int_{-x_0}^{t\, 2\pi/T_0-x_0} dx\, V_0(x-x_0)-\int_{\frac{2\pi}{d}k-x_0}^{t\, 2\pi/T_0+\frac{2\pi}{d}k-x_0} dx \,V_0(x-x_0)\right|.
\ee
Taking into account that
\be 
1=\int_{-\pi}^\pi dx V_0(x+x_0),
\ee
together with Eq. \eqref{eq:tilde ep V def} and the periodicity of the potential, it follows that
\be\label{eq:tilde ep equal}
\tilde \epsilon_V=\int_{x_{vr}-2\pi}^{x_{vl}} dx V_0(x+x_0).
\ee
Thus from Eq. \eqref{eq:epsilon k intermediate}, we have that
$\epsilon_k\leq 4\pi \tilde \epsilon_V$ if
\begin{align}\label{eq:intermediate befoe uni proof}
\begin{split}
x_{vr}-2\pi\leq -x_0&\leq t 2\pi/T_0-x_0\leq x_{vl}\\
x_{vr}-2\pi\leq -x_0\leq 2\pi \frac{k}{d}-x_0 &\leq t 2\pi/T_0+2\pi\frac{k}{d}-x_0\leq x_{vl} \quad \forall\, k\in \mathcal{I}_{\gamma_\psi}.
\end{split}
\end{align}
Taking into account $t\geq 0$ and that $\min  \mathcal{I}_{\gamma_\psi}= \bar\kappa-d\gamma_\psi/2+1$,  $\max  \mathcal{I}_{\gamma_\psi}= \bar\kappa+d\gamma_\psi/2-1$,
Eqs. \eqref{eq:intermediate befoe uni proof} are implied by Eqs. \eqref{eq:t before uni applied}.
To derive Eqs. \eqref{eq:t after uni applied}, we first note that
\begin{align}\label{eq:ep k after uni int}
\epsilon_k  =&2\pi\Bigg| \int_{-x_0}^{x_{vl}} dx\, V_0(x+x_0)+\int_{x_{vl}}^{x_{vr}} dx\, V_0(x+x_0)+\int_{x_{vr}}^{t\, 2\pi/T_0-x_0} dx\, V_0(x+x_0)\\
&-\int_{\frac{2\pi}{d}k-x_0}^{x_{vl}} dx\, V_0(x+x_0)-\int_{x_{vl}}^{x_{vr}} dx\, V_0(x+x_0)-\int_{x_{vr}}^{t\, 2\pi/T_0+\frac{2\pi}{d}k-x_0} dx\, V_0(x+x_0)\Bigg|\\
\leq &2\pi\Bigg| \int_{-x_0}^{x_{vl}} dx\, V_0(x+x_0)\Bigg|+2\pi\Bigg|\int_{x_{vr}}^{t\, 2\pi/T_0-x_0} dx\, V_0(x+x_0)\Bigg|\\
&+
2\pi\Bigg|\int_{\frac{2\pi}{d}k-x_0}^{x_{vl}} dx\, V_0(x+x_0)\Bigg|+2\pi\Bigg|\int_{x_{vr}}^{t\, 2\pi/T_0+\frac{2\pi}{d}k-x_0} dx\, V_0(x+x_0)\Bigg|.
\end{align}
In order for $\epsilon_k$, $k\in \mathcal{I}_{\gamma_\psi}$ to be bounded by $8\pi \tilde \epsilon_V$, the following conditions derived from Eq. \eqref{eq:ep k after uni int} are sufficient:
\be\label{eq: constrinats 3 coll}
\begin{split}
x_{vr}-2\pi &\leq -x_0\leq x_{vl}\\
x_{vr}-2\pi &\leq \frac{2\pi}{d}k-x_0\leq x_{vl}\quad \forall \,k\in \mathcal{I}_{\gamma_\psi}
\end{split}
\ee
and
\be\label{eq: constrinats 4 coll}
\begin{split}
x_{vr} &\leq t2\pi/T_0-x_0\leq x_{vl}+2\pi\\
x_{vr} &\leq t2\pi/T_0+\frac{2\pi}{d}k-x_0\leq x_{vl}+2\pi\quad \forall \,k\in \mathcal{I}_{\gamma_\psi}.
\end{split}
\ee
Eqs. \eqref{eq: constrinats 3 coll}, \eqref{eq: constrinats 4 coll} are implied by Eq. \eqref{eq:t after uni applied}. Thus whenever Eqs. \eqref{eq:t before uni applied} or Eqs. \eqref{eq:t after uni applied} are satisfied, the last term in Eq.  \eqref{eq:ep V up bound coll proof} is bounded by
\be 
 \max_{\bar\kappa\in[0,1]}\max_{k\in\mathcal{I}_{\gamma_\psi}}\left( \epsilon_k^2+\epsilon_k\right)\leq 8\pi \tilde \epsilon_V( 8\pi \tilde \epsilon_V+1),
\ee
\end{proof}
Finally, via the following corollary, we will now show how up to a small correction, the bounds derived in Lemma \ref{lem:trace dist bound t-Ham Vs d dim clock} are the same for $g$ as defined in Section \ref{sec:Consequences of Quasi-Autonomous control}. The asymptotic bounds in the large $d$ limit for the fidelity of the unitary implementation of the system reported in Section \ref{sec:Consequences of Quasi-Autonomous control}, will then be derived in the next Section, \ref{sec: control exmaples}.
\begin{corollary}[Section \ref{sec:Consequences of Quasi-Autonomous control} form]\label{lem:new g}
	Let $V_0$ and $\tilde \epsilon_V$ satisfy the same conditions as in Corollary \ref{corr:explicit bound on system trace distance}, and let $g$ be the normalised pulse from Section \ref{sec:Implementing Energy preserving unitaries via a time dependent Hamiltonian} with support interval $[t_i,t_f]=[t_1,t_2]$ 
	, then for all $t\in[0,t_1]\cup[t_2,T_0]$ and $0<t_1<t_2<T_0$,
	\be \label{eq: state dist in appendix in main text form}
	\| \rho_s(t)-\rho'_s(t)\|_1 \leq \sqrt{d_s \tr[\rho_s^2(0)]} \left( 2\varepsilon_v(t,d)+\varepsilon_v^2(t,d)+\varepsilon_s(t,d) \right),
	\ee 
	where 
	\begin{align}
	\varepsilon_s(t,d) :=&\epsilon_V(d)+2\pi T_0 \tilde \epsilon_V(2\pi T_0\tilde{\epsilon}_V+1),\\
	x_0 =& \frac{t_2+t_1}{2}\frac{2\pi}{T_0},
	\end{align}
	and $0<\gamma_\psi \leq 1\,$, $0<x_{vr} \leq \pi$ are any combination satisfying 
	\be 
	x_{vr}+\pi\gamma_\psi =\frac{t_2-t_1}{2}\frac{2\pi}{T_0}.
	\ee
	Furthermore, $\epsilon_V(d)$ is bounded by Eq. \eqref{eq:epsilon V upper bound explicit version}, and $\tilde{\epsilon}_V$ satisfies Eq. \eqref{eq:tilde ep V def}. $\rho_s(t)$ is given by Eq. \eqref{eq:rho t scho sol commute t depen} and $\rho'_s(t)=\tr_c[\rho'_{sc}(t)]$ by Eq. \eqref{eq:rho finite s c}.
\end{corollary}
\begin{proof}
	Let $g_1$ denoted $g$ as defined in this Corollary and $g_2$ denote $g$ as defined via Eq. \eqref{eq:g in terms of V_0}. With these functions, define $\epsilon_g(t)=\int_0^t dx(g_2(x)-g_1(x))$, $\rho_{S,q}(t)=U_q(t)\rho_s(0) U_q^\dag(t),$ $U_q(t)=\me^{-\mi \hat H_s^{int}\int_0^t g_q(t)}$, $q=1,2$. First note
	\be 
	\int_{t_1}^{t_2}dx g_2(x)=\frac{2\pi}{T_0}\int_{t_1}^{t_2}dx V_0(x 2\pi/T_0) =\int_{-x_{vr}-\pi\gamma_\psi}^{x_{vr}+\pi\gamma_\psi}dx V_0(x+x_0)\geq \int_{-x_{vr}}^{x_{vr}}dx V_0(x+x_0)=1-\tilde\epsilon_V.
	\ee 
	Thus 
	\be 
	1=\int_0^{T_0} dx\, g_2(x) \geq \int_0^{t_1} dx\, g_2(x)+1-\tilde{\epsilon}_V +\int_{t_2}^{T_0} dx\, g_2(x)\implies 
	\tilde{\epsilon}_V\geq \begin{cases}
		\int_0^{t_1} dx\, g_2(x)\\
		\vspace{-0.3cm}\\
	\int_{t_2}^{T_0} dx \,g_2(x)
	\end{cases}
	\ee \label{eq:tilde epsilon V}
	Using Eq. \eqref{eq:tilde epsilon V}, $\|\cdot\|_1/\sqrt{d_s}\leq \|\cdot\|_F$, where $\|\cdot\|_F$ is the Frobenius norm for a $d_s$ dimensional system, and the bounds $\sin(a)\leq |a|$, $-\cos(a)\leq a^2-1$, $a\in\rr$, we find
	\begin{align}\label{eq:trigabgle term boud}
	\| \rho_{s,1}(t)-\rho_{s,2}(t)\|_1 &= \| \rho_{s}(0)-U_1^\dag U_2(t)\rho_{s}(0) U_2^\dag U_1(t)\|_1\leq \sqrt{d_s} \| \rho_{s}(0)-U_1^\dag U_2(t)\rho_{s}(0) U_2^\dag U_1(t)\|_F\\
	&= \sqrt{d_s} \sqrt{\sum_{m,n} \big{|}\rho_{m,n}(t)-\me^{-\mi (\Omega_m-\Omega_n)\epsilon_g(t)} \rho_{m,n}(t)\big{|}^2}\leq \sqrt{d_s\tr[\rho_s^2(0)]}\max_{m,n}|1-\me^{-\mi (\Omega_m-\Omega_n)\epsilon_g(t)}|\\
&	\leq 2\pi(2\pi|\epsilon_g(t)|+1)|\epsilon_g(t)|.
	\end{align}
	Using Eq. \eqref{eq:tilde epsilon V}, we thus conclude $|\epsilon_g(t)|\leq \tilde{\epsilon}_V T_0$, for $t\in[0,t_1]\cup[t_2,T_0]$. Finally, note that $\| \rho_s(t)-\rho_{s,1}(t)\|_1 \leq \| \rho_s(t)-\rho_{s,2}(t)\|_1 +\| \rho_{s,1}(t)-\rho_{s,2}(t)\|_1$ via the triangle inequality. $\| \rho_s(t)-\rho_{s,2}(t)\|_1$ was bounded in Eq. \eqref{eq:trace dist rho rho finite and time}, with $\epsilon_V(t,d)$ bounded in Corollary \ref{corr:explicit bound on system trace distance} for $t\in[0,t_1]\cup[t_2,T_0]$. Using these two  bounds and Eq. \eqref{eq:trigabgle term boud}, Eq, \eqref{eq: state dist in appendix in main text form} follows. To finalize the proof, note that $x_0=\frac{t_2+t_1}{2}\frac{2\pi}{T_0}$, $x_{vr}+\pi\gamma_\psi=\frac{t_2-t_1}{2}\frac{2\pi}{T_0}$. As such, the condition $x_{vr}+\pi \gamma_\psi  \leq x_0\leq 2\pi-x_{vr}-\pi \gamma_\psi$ from Lemma \ref{corr:explicit bound on system trace distance} is satisfied for all $0<t_1<t_2<T_0$  and Eq. \eqref{eq:epsilon V upper bound explicit version} is valid.
\end{proof}
\subsubsection{Examples}\label{sec: control exmaples}
We can now bound $\tilde{\epsilon}_V$ for the cosine potential defined in Section \ref{sec:Examples of Potential functions} with the aim of understanding the significance of Corollaries \ref{corr:explicit bound on system trace distance}, \ref{lem:new g}. We find the bound
\be\label{eq: up bound tilde epsilon V}
\tilde \epsilon_V\leq \frac{(\pi-x_{vr})\me^2}{4\pi\sqrt{\pi}}\sqrt{n}\cos^{2n}(x_{vr}/2), \quad\text{if } \cos(x_{vr})\leq 1-\frac{1}{n}.
\ee
A proof of Eq. \eqref{eq: up bound tilde epsilon V} can be found in \supp~\ref{sec:ep V ex pot bound}. By choosing $n=n(d)$ in different ways, one can achieve different decay rates. This choice will also effect the decay rates of $\varepsilon_v(t,d)$ defined in Eq. \ref{eq:ltwo epsilon nor control theorem}, the significance of which will be clarified in Section \ref{Clock Fidelity}. We give three examples of $n=n(d)$:
\begin{itemize}
\item[1)]\emph{Power law decay.} Let 
\be n=\gamma_1\frac{\ln d}{-2\ln \cos(x_{vr}/2)},\quad \gamma_1>0.
\ee
Then
\be 
\tilde \epsilon_V\leq \frac{(\pi-x_{vr})\me^2}{4\pi\sqrt{\pi}} \sqrt{\gamma_1\frac{\ln d}{-2\ln \cos(x_{vr}/2)}}\,d^{-\gamma_1},
\ee
if
\be\label{eq:gamma 1 exmaple condition}
\gamma_1\geq \frac{-2\ln \cos(x_{vr}/2)}{\left(1-\cos(x_{vr})\right)}\frac{1}{\ln d}.
\ee
Note that $\lim_{d\rightarrow \infty} \tilde\epsilon_V d^m=0$ for all $\gamma_1>m>0$. Using the definition of the clock rate parameter $\bar\upsilon$ (Def. \ref{def:decay rate params}), we find
\be 
\bar\upsilon =\frac{\ln d}{\ln (\pi \alpha_0 \sigma^2)} \frac{2 \me^2}{(2\pi)^{3/2} \sqrt{2}}\frac{\pi \kappa \alpha_0}{(-2\ln\cos (x_{vr}/2))^{3/2}}\gamma_1^{3/2} \sqrt{\ln d}. 
\ee
Thus from Eq. \eqref{eq:ltwo epsilon nor control theorem}, for $\sigma=\sqrt{d}$,
\be\label{eq:epsilon decay rate exmaple 1} 
\varepsilon_v(t,d)=\bo\left(t\, poly(d)\,\exp\left(-\frac{\pi}{4}\frac{\alpha_0^2 \chi_1^2 \gamma_1^3}{\left( 1+ \chi_1 \gamma_1^{3/2}/\sqrt{\ln d} \right)^2}  \frac{d}{\ln d}\right)  \right),\quad \text{as }\,\, d\rightarrow\infty
\ee
with 
\be 
\chi_1:= \frac{(2\pi)^{3/2}\sqrt{2}}{2\,\me^2} \left(1+ \frac{\ln(\pi\alpha_0)}{\ln d}\right) \frac{\left(-2\ln(\cos x_{vr}/2)\right)^{3/2}}{\pi \kappa \alpha_0}.
\ee
Furthermore, using Eq. \eqref{eq:epsilon V upper bound explicit version}, for Eqs. \eqref{eq:trace dist rho rho finite and time}, \eqref{eq: state dist in appendix in main text form} one has that
\be 
\frac{\|\rho_s(t)-\rho'_s(t)\|_1}{\sqrt{d_s\tr[\rho_s^2(0)]}}=\bo\left(\varepsilon_v(t,d)\right)+\bo\left(\me^{-2\pi\tilde k (d/\sigma)^2}\right)+\bo\left(\tilde\epsilon_V\right),\quad \text{as }\,\, d\rightarrow\infty.
\ee
Thus since $\varepsilon_v(t,d)$ decays faster than any power of $d$, for $\sigma=\sqrt{d}$ and any $\gamma_\psi\in(0,1]$, we have that 
\be 
\lim_{d\rightarrow \infty} \|\rho_s(t)-\rho'_s(t)\|_1 \,d^m=0,
\ee
for all constants $\gamma_1>m>0$.

\item[2)] \emph{System error faster than power-law decay.}  By parametrizing $\gamma_1$ in example 1) to increase sufficiently slowly with $d$ rather than being constant, we can achieve that the trace distance error between $\rho_s$ and $\rho_s'$ (Eqs. \eqref{eq:trace dist rho rho finite and time} and \eqref{eq: state dist in appendix in main text form}) decays with $d$ faster than any power law in $d^m$ for the optimal clock decay $\sigma=\sqrt{d}$. The penalty of this improvement, will be a worse decay rate for $\varepsilon_v(t,d)$ than that of Eq. \eqref{eq:epsilon decay rate exmaple 1} . Let 
\be\label{eq:gamma 1 equlity eg}
\gamma_1=\frac{\pi}{4}\alpha_0^2 \chi_2^2\frac{d^{1/4}}{\sqrt{\ln d}},\quad \chi_2:=\chi_1\left(\frac{\pi}{4}\alpha_0^2\chi_1^2\right)^{-3/8}.
\ee
such that Eq. \eqref{eq:gamma 1 exmaple condition} is satisfied, then for $\sigma=\sqrt{d}$, 
\begin{align}\label{eq:epsilon decay rate exmaple 1,2} 
\varepsilon_v(t,d)&=\bo\left(t\, poly(d)\,\,\me^{- \frac{\pi}{4}\alpha_0^2 \chi_2^2 \,d^{1/4}\sqrt{\ln d}}\,\right),\quad \text{as }\,\, d\rightarrow\infty,\\
\tilde \epsilon_V&\leq \frac{(\pi-x_{vr})\me^2 \alpha_0\chi_2}{8\pi} \frac{(d\ln d)^{1/4}}{\sqrt{-2\ln \cos(x_{vr}/2)}}\,\me^{- \frac{\pi}{4}\alpha_0^2 \chi_2^2 \,d^{1/4}\sqrt{\ln d}},
\end{align}
and thus up to polynomial factors, both $\varepsilon_v(t,d)$ and $\tilde\epsilon_V$ have exponential decay in $d^{1/4}\sqrt{\ln d}$ and
\be 
\frac{\|\rho_s(t)-\rho'_s(t)\|_1}{\sqrt{d_s\tr[\rho_s^2(0)]}}=\bo\left(\left( t \,poly(d) +(d\ln d)^{1/4}\right)\me^{- \frac{\pi}{4}\alpha_0^2 \chi_2^2 \,d^{1/4}\sqrt{\ln d}}\,\right),
\ee
from which we conclude
\be\label{eq: faster than power decay trace dist} 
\lim_{d\rightarrow \infty} \|\rho_s(t)-\rho_s'(t)\|_1\, d^m=\lim_{d\rightarrow \infty} \varepsilon_v(t,d)\, d^m=0,
\ee
for all constant $m>0$ and for all $t$ such that Eq. \eqref{eq:t before uni applied} and Eq. \eqref{eq:t after uni applied} hold. Furthermore, the constraint of Eq. \eqref{eq:gamma 1 exmaple condition} can be re-written in terms of the new parametrizations to achieve, using Eqs. \eqref{eq:gamma 1 exmaple condition} and \eqref{eq:gamma 1 equlity eg} the constraint
\be 
\frac{4}{\pi\alpha_0^2\chi_2^2}\frac{-2\ln\cos(x_{vr}/2)}{1-\cos (x_{vr})}\leq d^{1/4}\sqrt{\ln d}.
\ee
This constraint is satisfied for all constant $x_{vr}\in(0,2\pi)$ for sufficiently large $d$ (recall that by definition, $\chi_2$ depends on $x_{vr}$). Similarly, $\gamma_\psi$ can be any constant in $(0,1)$.
\item[3)]\emph{Smallest clock error.} In this case, we will bound how quickly the error $\tilde \epsilon_V$ can decay with $d$  while maintaining the smallest disturbance imposed on the clock due to performing the unitary (see section \ref{Clock Fidelity} for details), that is to say, while maintaining a constant rate parameter $\bar\upsilon$, for the optimal clock decay $\sigma=\sqrt{d}$.\\
Let
\be 
n=\frac{\gamma_3}{-2\ln\cos(x_{vr})} (\ln d)^{2/3},\quad \gamma_3>0.
\ee
Then $\tilde\epsilon_V$ decays with increasing $d$ at a rate 
\be 
\tilde\epsilon_V\leq \frac{(\pi-x_{vr})\me^2}{4\pi\sqrt{\pi}} \sqrt{\frac{\gamma_3}{-2\ln\cos(x_{vr})}} (\ln d)^{2/3}  d^{-\gamma_3(\ln d)^{-1/3}}, \quad\text{if } \gamma_3\geq \frac{-2\ln \cos(x_{vr})}{\left(1-\cos(x_{vr})\right)}\frac{1}{(\ln d)^{2/3}},
\ee
and the clock rate parameter (Def. \ref{def:decay rate params}) is
\be 
\bar\upsilon =\frac{\ln d}{\ln (\pi \alpha_0 \sigma^2)} \frac{2 \me^2}{(2\pi)^{3/2} \sqrt{2}}\frac{\pi \kappa \alpha_0}{(-2\ln\cos (x_{vr}))^{3/2}}\gamma_3^{3/2}, 
\ee
which is bounded in the $d\rightarrow \infty$ limit. Thus for this parametrisation, $\varepsilon_v(t,d)$ has exponential decay in $d$ and consequently the clock's disturbance decays exponentially when $\sigma=\sqrt{d}$, (recall Theorem \ref{movig through finite time}).
\end{itemize}

\subsection{How to perform any timed unitary}
\label{sec:non energy preserving unitaries}
So far we have concerned ourselves with performing energy preserving unitaries only. Often, one may want to perform any unitary on the system $\mathcal{H}_s$. The clock provides a source of timing, but not energy.  To perform non-energy preserving unitaries, one will need an energy source. To achieve this, one can introduce a battery system with Hilbert space $\mathcal{H}_b$. Then, the clock can perform energy preserving unitaries on a state $\rho_{sb}$ on $\mathcal{H}_s\otimes\mathcal{H}_b$ according the procedure described in Section \ref{sec:Implementing Energy preserving unitaries with the finite}. The initial state of the battery is chosen such that it has enough energy to implement the unitary and will be such that an energy preserving unitary over $U_{sb}$ on a state $\rho_s\otimes\rho_b$ will induce any desired unitary $U'_{s}$, on $\rho_s$, namely $\tr_b[U_{sb} \rho_s\otimes\rho_b U_{sb}^\dag]= U'_{s} \rho_s U^{\prime \dag}_{s}$, where $[U_{sb},\hat H_s\otimes\id_b+\id_s\otimes \hat H_s]=0$, $[U'_{s},\hat H_s
]\neq 0$. The procedure which allows one to find a state $\rho_b\in\mathcal{S}(\mathcal{H}_b)$, a Hamiltonian $\hat H_b$ on $\mathcal{H}_b$ and an energy preserving unitary $U_{sb}$ on $\mathcal{H}_{sb}$, such that any desired unitary $U'_s$ on $\mathcal{H}_s$ can be performed, is the topic of \cite{Aberg}. Therefore, using the finite dimension clock together with a battery as discussed in this Section, one can achieve any timed non-energy preserving unitary on states on $\mathcal{H}_s$ using the results of \cite{Aberg}. This will be developed with explicit error bounds in an upcoming paper.

\subsection{Clock Fidelity}\label{Clock Fidelity}
If the potential $V_0(x)$ is zero for all $x\in\rr$, then the system on $\mathcal{H}_s$ and the clock on $\mathcal{H}_c$ will evolve in time as non-interacting systems. After a time $T_0$, as we have seen, the clock state will return to its initial state. However, whenever $V_0(x)\neq 0$ for some $x\in\rr$, the system-clock interaction will disturb the dynamics of the clock such that after a time $T_0$, the clock state will not have returned to its initial state. Lemma \ref{Lem:clock fidelity} in this Section bounds quantitatively how large this disturbance is, and how quickly it decays with clock dimension.
\begin{lemma}[Clock disturbance]\label{Lem:clock fidelity}
Let $\rho_c'(t):=\tr_s[\rho_{sc}'(t)]$, ($\rho_{sc}'$ defined in Eq. \eqref{eq:rho finite s c}), be the state of the clock at time $t$. We have the bound,
\be\label{eq:trace distance bound}
\frac{1}{2}\| \rho_c'(0)-\rho_c'(T_0)\|_1\leq \varepsilon_v(T_0,d),
\ee
where $\varepsilon_v(\cdot,\cdot)$ is defined in Eq. \eqref{eq:ltwo epsilon nor control theorem}.
\end{lemma}
\emph{Intuition}
The bound on the trace distance given by Eq. \eqref{eq:trace distance bound}, decays quickly with $d$ while increasing linearly in $T_0$ as is evident from the definition of $\varepsilon_v(\cdot,\cdot)$. It bounds the disturbance which is \textit{only} originating from the implementation of the unitary on the system on $\mathcal{H}_s$, since $\| \rho_c'(0)-\rho_c'(T_0)\|_1=0$ if $V_0 (x)=0$ for all $x\in \rr$.
\begin{proof}
The result follows from Theorem \ref{eq: main eq in control theorem} and simple identities.\\
By taking the partial trace over the system Hilbert space $\mathcal{H}_s$, from Eq. \eqref{eq:rho finite s c} we achieve
\be\label{}
\rho_{sc}^{\prime}(t)= \sum_{m,n=1}^{d_s} \rho_{m,n}(t)\ketbra{\phi_m}{\phi_n}\otimes\ketbra{\bar\Phi_m(t)}{\bar\Phi_n(t)}_c,
\ee
where we have used definitions Eqs. \eqref{eq:rho s in mutual alphog basis}, \eqref{eq:Psi bar def eq}. Taking the partial trace over the system Hilbert space $\mathcal{H}_s$, we achieve
\be
\rho_{c}^{\prime}(t)= \sum_{n=1}^{d_s} \rho_{n,n}(0)\ketbra{\bar\Phi_n(t)}{\bar\Phi_n(t)}_c.
\ee
Thus the \textit{quantum Fidelity} $F$ 
 is
\begin{align}\label{eq: clock inital vs final fidelity in proof 1}
F( \rho_c'(0),\rho_c'(t))&=\tr\left[\sqrt{\sqrt{\rho_{c}^{\prime}(0)}\rho_{c}^{\prime}(t)\sqrt{\rho_{c}^{\prime}(t)}}\right]\\
&=\tr\left[\sqrt{\ketbra{\Psi_\textup{nor}(k_0)}{\Psi_\textup{nor}(k_0)}\left(\sum_{n=1}^{d_s} \rho_{n,n}(0)\hat\Gamma_n(t) \ketbra{\Psi_\textup{nor}(k_0)}{\Psi_\textup{nor}(k_0)}\hat\Gamma_n^\dag(t)\right)\ketbra{\Psi_\textup{nor}(k_0)}{\Psi_\textup{nor}(k_0)}}\right]\\
&=\tr\left[\sqrt{\ketbra{\Psi_\textup{nor}(k_0)}{\Psi_\textup{nor}(k_0)}\left(\sum_{n=1}^{d_s} \rho_{n,n}(0)\left|\bra{\Psi_\textup{nor}(k_0)}\hat\Gamma_n(t)\ket{\Psi_\textup{nor}(k_0)}\right|^2\right)}\right]\\
&=\sqrt{\sum_{n=1}^{d_s} \rho_{n,n}(0)\left|\bra{\Psi_\textup{nor}(k_0)}\hat\Gamma_n(t)\ket{\Psi_\textup{nor}(k_0)}\right|^2},
\end{align}
where we have defined
\be 
\hat\Gamma_n(t):= \me^{-\mi t( \hat{H}_c + \Omega_n \hat{V}_d)}.
\ee
Applying Theorem \ref{eq: main eq in control theorem}, we find
\begin{align}\label{eq:modulus sq of exp Gamma t}
\left|\bra{\Psi_\textup{nor}(k_0)}\hat\Gamma_n(T_0)\ket{\Psi_\textup{nor}(k_0)}\right|^2 &=\left(\me^{-\mi\Omega_n}+\braket{\Psi_\textup{nor}(k_0)|\epsilon_n}\right) \left(\me^{-\mi\Omega_n}+\braket{\Psi_\textup{nor}(k_0)|\epsilon_n}\right)^*\\
&=1+\braket{\Psi_\textup{nor}(k_0)|\epsilon_n}\me^{+\mi \Omega_n}+\braket{\epsilon_n|\Psi_\textup{nor}(k_0)}\me^{-\mi \Omega_n}+|\braket{\epsilon_n|\Psi_\textup{nor}(k_0)}|^2\\
&\geq 1+\braket{\Psi_\textup{nor}(k_0)|\epsilon_n}\me^{+\mi \Omega_n}+\braket{\epsilon_n|\Psi_\textup{nor}(k_0)}\me^{-\mi \Omega_n}.
\end{align}
However, 
\begin{align}
	&1=\left| \bra{\Psi_\textup{nor}(k_0)} \Gamma_n^\dagger(T_0) \Gamma_n(T_0) \ket{\Psi_\textup{nor}(k_0)} \right|^2 =
	\left( \bra{\Psi_\textup{nor}(k_0)} \me^{\mi\Omega_n} + \bra{\epsilon_n} \right) \left( \me^{-\mi\Omega_n} \ket{\Psi_\textup{nor}(k_0)} + \ket{\epsilon_n} \right) \\
	 \implies& \braket{\Psi_\textup{nor}(k_0)|\epsilon_n}\me^{+\mi \Omega_n}+\braket{\epsilon_n|\Psi_\textup{nor}(k_0)}\me^{-\mi \Omega_n} = -\braket{\epsilon_n|\epsilon_n} = -||\ket{\epsilon_n}||^2_2,
\end{align}
Thus from Eq. \eqref{eq:modulus sq of exp Gamma t}, we achieve
\be 
\left|\bra{\Psi_\textup{nor}(k_0)}\hat\Gamma_n(T_0)\ket{\Psi_\textup{nor}(k_0)}\right|^2 \geq 1-||\ket{\epsilon_n}||^2_2.
\ee
Plunging into Eq. \eqref{eq: clock inital vs final fidelity in proof 1},
\be 
F( \rho_c'(0),\rho_c'(T_0))\geq \sqrt{\sum_{n=1}^{d_s} \rho_{n,n}(0)\left(1-||\ket{\epsilon_n}||^2_2\right)}\geq \sqrt{\min_{n\in{1,\ldots,d_s}}\left(1-||\ket{\epsilon_n}||^2_2\right)}=\sqrt{1- \varepsilon_v^2(T_0,d)},
\ee
where $\varepsilon_v(\cdot,\cdot)$ is defined in Eq. \eqref{eq:ltwo epsilon nor control theorem}. Thus we conclude the Lemma using the well known relationship between Fidelity and trace distance, namey $\frac{1}{2} \|\rho-\sigma\|_1\leq \sqrt{1-F^2(\rho,\sigma)}$, for two quantum states $\rho$, $\sigma$.
\end{proof}

\onecolumngrid
\section{Commutator relations for the \gClock~states}\label{appendixcommutator}
For simplicity, in this section we will assume the dimension $d$ to be odd, and shift spectrum of $\hat H_c$ and $\hat t_c$ to be centered at zero, giving us
\begin{align}
	\hat{H}_c &= \sumnsym n \frac{2\pi}{T_0} \ketbra{E_n}{E_n}, & \hat{t}_c &= \sum_{k = -\frac{d-1}{2}}^{+\frac{d-1}{2}} k \frac{T_0}{d} \ketbra{\theta_k}{\theta_k}.
\end{align}
The shifting of the spectrum clearly has no physical effect on the dynamics nor the commutator of these operators since the ground state and initial time can be sifted back to zero by simply adding a term proportional to the identity operator to the r.h.s. of $\hat H_c$ and $\hat t_c$. The mean energy of the state $n_0$ before defined on $n_0\in(0,d-1)$ in this section is now shifted to $n_0\in\left(-\frac{d}{2},\frac{d}{2}\right)$. We will also assume $k_0\in\left(-\frac{d}{2},\frac{d}{2}\right)$ in this section for simplicity and define.
\begin{align}\label{eq:alpha_c def eq}
	\bar\alpha &= \abs{\frac{2n_0}{d}}, & \bar\beta &= \abs{\frac{2 k_0}{d}}.
\end{align}

Thus both $\bar\alpha, \bar\beta \in [0,1)$. They are both measures of how close to the edge of the spectra of $\hat H_c$ and $\hat t_c$ the state is; c.f. similar measure Def. \eqref{Distance of the mean energy from the edge}.

\begin{theorem}[Quasi-Canonical commutation]\label{Quasi-Canonical commutation}
For all states $\ket{\Psi_\textup{nor}(k_0)}\in \Lambda_{\sigma,n_0}$, 
 the time operator $\hat t_c$ and Hamiltonian $\hat{H_c}$ satisfy the commutation relation
\begin{align}
	\left[ \hat{t}_c, \hat{H}_c \right] \ket{\Psi_\textup{nor}(k_0)} &= \mi \ket{\Psi_\textup{nor}(k_0)} + \ket{\epsilon_\textup{comm}},
\end{align}
where 
\begin{align}\label{eq: commutation finite dim theorem}
\|\ket{\epsilon_\textup{comm}}\|_2 \leq &\,  \epsilon_8^{co}+\frac{1}{2}\left( \epsilon_7^{co}+\epsilon_6^{co}+\pi d\, \epsilon_5^{co} \right)+\epsilon_4^{co}+\epsilon_3^{co}+\pi d \left( \epsilon_2^{co}+\epsilon_1^{co} \right)\\
=&\bo\left(d\sigma^{5/2}\right) \me^{-\frac{\pi}{4}\sigma^2(1-\bar\alpha)^2}+ \left( \bo\left(\frac{d^2}{\sigma^{5/2}}\right)+\bo\left(d\sigma^{1/2}\right) \right)\me^{-\frac{\pi}{4}\frac{d^2}{\sigma^2}(1-\bar\beta)^2}\\
&+\left( \bo\left(\frac{d^3}{\sigma^{5/2}}\right)+\bo\left(d^2\sigma^{3/2}\right) \right)\me^{-\frac{\pi}{4}\frac{d^2}{\sigma^2}},
\end{align}
in the limits $d\rightarrow \infty$, $(0,d)\ni\sigma\rightarrow\infty$. Furthermore, for the completely symmetric state (Def. \ref{def:clock stat classes}) centred at the origin, $\bar\alpha=\bar\beta=k_0=0$, we have the simplification $\|\ket{\epsilon_\textup{comm}}\|_2=\bo  \left( d^{9/4}\right) \me^{-\frac{\pi d}{4}} $. Full details for finite $d,\sigma$ can be found in $\epsilon_1^{co}$ given by Eq. \eqref{eq:ep1 comm bound}, $\epsilon_2^{co}$ by Eq. \eqref{eq:ep2 comm bound}, $\epsilon_3^{co}$ by Eq. \eqref{eq:ep3 comm bound}, $\epsilon_4^{co}$ by Eq. \eqref{eq:ep4 comm bound}, $\epsilon_5^{co}$ by Eq. \eqref{eq:ep5 comm bound}, $\epsilon_6^{co}$ by Eq. \eqref{eq:ep6 comm bound}, $\epsilon_7^{co}$ by Eq. \eqref{eq:ep7 comm bound}, and $\epsilon_8^{co}$ by Eq. \eqref{eq:ep8 comm bound}.
\end{theorem}
\emph{Intuition} It is well known that operators $\hat a$, $\hat b$ satisfying $[\hat a, \hat b]=\mi$ can only exist in infinite dimensions. The above results shows that, although this is the case, such operators can exist up to a small error | exponentially small in $d$ when the state $\ket{\Psi_\textup{nor}(k_0)}$ is symmetric (i.e., when $\sigma=\sqrt{d}$). Also see the main text and Sections \ref{idealcommutator}, \ref{Shortcomings of the time-states} for more discussions on the topic, physical intuition and how this result further relates our \gClock~states to the idealized clock case. We also observe that by definition, the l.h.s. of Eq. \eqref{eq: commutation finite dim theorem} is $T_0$ independent, and thus so is the error term $\ket{\epsilon_\textup{comm}}$.
\begin{remark}\label{rem:result3 Serge}
In a different context, discrete approximations to Weyl and canonically conjugate operators have been defined and approximations found for their standard deviation \cite{sergedft}. In \cite{sergedft}, it is conjectured that the approximate canonical commutator on states which are of a similar form to Eq. \eqref{gaussianclock_main}, approaches $\mi\hbar$ as $d\rightarrow \infty$. Theorem \ref{Quasi-Canonical commutation} not only verifies this statement, but shows that the error can be exponentially small.
\end{remark}
\begin{proof}
The proof will consist in going back and forth between representations of the state in the energy basis and the time basis, bounding summations over Gaussian tails while making crucial use of the Poisson summation formula to bound such approximations in the process. The identities used during the proof for the Poisson summation formula can be found in Section \ref{Poisson summation formula}.
For simplicity, we calculate the commutator for $k_0$ and $n_0$ both greater than zero, the general proof is analogous.

Applying the time operator on the state,
\begin{align}
	\hat{t}_c \ket{\Psi(k_0)} &= \sumt{l} \sumS{k}{k_0} \psi(k_0;k) \ket{\theta_k}.
\end{align}

Since $k_0 \in [ 0, d/2 ]$, for $k \in \mathcal{S}_d(k_0)$,
\begin{align}
	\hat{t}_c \ket{\theta_k} = \begin{cases}
					k \frac{T_0}{d} \ket{\theta_k} & \text{if $k < \frac{d}{2} $,} \\
					(k-d) \frac{T_0}{d} \ket{\theta_k} & \text{if $k > \frac{d}{2} $.}
							 			\end{cases}
\end{align}

Accordingly,
\begin{align}
	\hat{t}_c \ket{\Psi(k_0)} &= \sumS{k}{k_0} \frac{T_0}{d} k \psi(k_0;k) \ket{\theta_k} + \ket{\epsilon_1}, \\
	\text{where} \; \ket{\epsilon_1} &= - \sum_{k \in \mathcal{S}_d(k_0), k > \frac{d}{2}} T_0 \psi(k_0;k) \ket{\theta_k},
\end{align}
\be\label{eq:ep1 comm bound}
\epsilon_1^{co}:= \ltwo{\ket{\epsilon_1}}/T_0 
	< 
	\begin{cases}   2  A \frac{\errort}{1 - \errortd}  \mbox{ if } \sigma=\sqrt{d}\\
 2  A \frac{e^{-\frac{\pi d^2}{4 \sigma^2}(1-\bar\beta))^2}}{1 - e^{-\frac{\pi d}{\sigma^2}(1-\bar\beta))}} 	 \mbox{ otherwise}
\end{cases}
\ee

In anticipation of applying the Hamiltonian next, we move to the basis of energy eigenstates,
\begin{align}
	\hat{t}_c \ket{\Psi(k_0)} &= \sumnsym \sum_{k \in \mathcal{S}_d(k_0)} \frac{T_0}{d} k \psi(k_0;k) \frac{e^{-i 2\pi n k/d}}{\sqrt{d}} \ket{E_n} + \ket{\epsilon_1}.
\end{align}

We approximate the sum w.r.t. $k$ by the infinite sum and bound the difference,
\begin{align}
	\hat{t}_c \ket{\Psi(k_0)} &= \sumnsym \sum_{k \in \mathbb{Z}} \frac{T_0}{d} k \psi(k_0;k) \frac{e^{-i 2\pi n k/d}}{\sqrt{d}} \ket{E_n} + \ket{\epsilon_2} + \ket{\epsilon_1}, \\
	\text{where} \; \ket{\epsilon_2} &= \sumnsym \sum_{k \in \mathbb{Z}/\mathcal{S}_d(k_0)} \frac{T_0}{d} k \psi(k_0;k) \frac{e^{-i 2\pi n k/d}}{\sqrt{d}} \ket{E_n},
\end{align}

\be\label{eq:ep2 comm bound}
\epsilon_2^{co}:= \ltwo{\ket{\epsilon_2}}/T_0 
	< 
	\begin{cases}   A \sqrt{d} \left( 1 + \frac{1}{\pi} + \frac{\bar\beta}{1 - e^{-\pi}} \right) e^{-\frac{\pi d}{4}} \mbox{ if } \sigma=\sqrt{d}\\
 A \sqrt{d} \left( 1 + \frac{\sigma^2}{\pi d} + \frac{\bar\beta}{1 - e^{-\frac{\pi d}{\sigma^2}}} \right) e^{-\frac{\pi d^2}{4\sigma^2}}	 \mbox{ otherwise}
\end{cases}
\ee

Applying the Poissonian summation formula on the infinite sum,
\begin{align}
	\hat{t}_c \ket{\Psi(k_0)} &= 
		\frac{iT_0}{2\pi} \sumnsym \sum_{m\in\mathbb{Z}} \frac{d}{dp} \tilde{\psi}(k_0;p) \bigg|_{p=n+md} \ket{E_n} + \ket{\epsilon_2} + \ket{\epsilon_1}.
\end{align}

Applying the Hamiltonian,
\begin{align}
	\hat{H}_c \hat{t}_c \ket{\Psi(k_0)} &= 
		i \sumnsym \sum_{m\in\mathbb{Z}} n \frac{d}{dp} \tilde{\psi}(k_0;p) \bigg|_{p=n+md} \ket{E_n} + \hat{H}_c \left( \ket{\epsilon_2} + \ket{\epsilon_1} \right).
\end{align}

Finally, we shift back to the basis of time-states,
\begin{align}
	\hat{H}_c \hat{t}_c \ket{\Psi(k_0)} &= i \sum_{l \in \mathcal{S}_d(k_0)} \sumnsym \sum_{m\in\mathbb{Z}} n \frac{d}{dp} \tilde{\psi}(k_0;p) \bigg|_{p=n+md} \frac{e^{+i 2\pi nl/d}}{\sqrt{d}} \ket{\theta_l} + \hat{H}_c \left( \ket{\epsilon_2} + \ket{\epsilon_1} \right).
\end{align}

Combining the summations over the indices $n$ and $m$, and using $\sumnsym \sum_{m\in\mathbb{Z}} f(n+md) = \sum_{s\in\mathbb{Z}} f(s)$, and noting that
\begin{align}
	\sum_{m \in \mathbb{Z}} \abs{ m \frac{d}{dp} \tilde{\psi}(k_0;p) \bigg|_{p=n+md} } = 2 \pi A \frac{\sigma}{\sqrt{d}} \sum_{m \in \mathbb{Z}} e^{-\frac{\pi \sigma^2}{d^2}(n-n_0+md)^2} \abs{m} \abs{ \frac{\sigma^2}{d^2} (n-n_0+md) + i \frac{k_0}{d} },
\end{align}
we have
\begin{align}
	\hat{H}_c \hat{t}_c \ket{\Psi(k_0)} &= i \sum_{l \in \mathcal{S}_d(k_0)} \sum_{s\in\mathbb{Z}} s \frac{d}{dp} \tilde{\psi}(k_0;p) \bigg|_{p=s} \frac{e^{+i 2\pi sl/d}}{\sqrt{d}} \ket{\theta_l} + \ket{\epsilon_3} + \hat{H}_c \left( \ket{\epsilon_2} + \ket{\epsilon_1} \right), \\
	\text{where} \; \ket{\epsilon_3} &= -id \sum_{l \in \mathcal{S}_d(k_0)} \sumnsym \sum_{m\in\mathbb{Z}} m \frac{d}{dp} \tilde{\psi}(k_0;p) \bigg|_{p=n+md} \frac{e^{+i 2\pi nl/d}}{\sqrt{d}} \ket{\theta_l},
\end{align}
\begin{align}\label{eq:ep3 comm bound}
\epsilon_3^{co} &:= \ltwo{\ket{\epsilon_3}} 
	\\&< 
	\begin{cases}  2\pi A d^2 \sqrt{d} \left( (1-\bar\alpha) + \frac{1}{\pi d} \left( 2 + \frac{1}{1 - e^{-\pi d (1-\bar\alpha)}} \right) + \frac{\bar\beta}{2} \left( 1 - \bar\alpha + \frac{1}{\pi d} + \frac{1+\bar\alpha}{1 - e^{-\pi d(1-\bar\alpha)}} \right) \right) e^{-\frac{\pi d}{4}(1-\bar\alpha)^2} \mbox{ if } \sigma=\sqrt{d}\\
 2\pi A \sigma d^2 \left( \frac{\sigma^2}{d} (1-\bar\alpha) + \frac{1}{\pi d} \left( 2 + \frac{1}{1 - e^{-\pi\sigma^2(1-\bar\alpha)}} \right) + \frac{\bar\beta}{2} \left( 1 - \bar\alpha + \frac{1}{\pi\sigma^2} + \frac{1+\bar\alpha}{1 - e^{-\pi\sigma^2(1-\bar\alpha)}} \right) \right) e^{-\frac{\pi\sigma^2}{4}(1-\bar\alpha)^2}  	 \mbox{ otherwise}
\end{cases}
\end{align}

Applying the Poissonian summation formula on the infinite sum w.r.t. $s$,
\begin{align}
	\hat{H}_c \hat{t}_c \ket{\Psi(k_0)} &= -i \sum_{l \in \mathcal{S}_d(k_0)} \sum_{m\in\mathbb{Z}}  \frac{d}{dx} \left( x \psi(k_0;x) \right)\bigg|_{x=l+md} \ket{\theta_l} + \ket{\epsilon_3} + \hat{H}_c \left( \ket{\epsilon_2} + \ket{\epsilon_1} \right).
\end{align}

Approximating the sum over $m$ by the $m=0$ term,
\begin{align}
	\hat{H}_c \hat{t}_c \ket{\Psi(k_0)} &= -i \sum_{l \in \mathcal{S}_d(k_0)} \psi(k_0;l) \ket{\theta_l} + i \sum_{l \in \mathcal{S}_d(k_0)} l \frac{d}{dx} \psi(k_0;x)\bigg|_{x=l} \ket{\theta_l} + \ket{\epsilon_4} + \ket{\epsilon_3} + \hat{H}_c \left( \ket{\epsilon_2} + \ket{\epsilon_1} \right), \\
	\text{where} \; \ket{\epsilon_4} &= -i \sum_{l \in \mathcal{S}_d(k_0)} \sum_{m\in\mathbb{Z}/\{0\}} \frac{d}{dx} \left( x\psi(k_0;x) \right)\bigg|_{x=l+md} \ket{\theta_l},
\end{align}
\be\label{eq:ep4 comm bound}
\epsilon_4^{co}:= \ltwo{\ket{\epsilon_4}} 
	< 
	\begin{cases} dA \left( d (\pi+1) (1+\bar\beta) + 2 + \frac{2}{1 - e^{-\pi}} + \bar\alpha \left( d (\pi+1) + \frac{\pi\bar\beta}{1 - e^{-\pi}} \right) \right) e^{-\frac{\pi d}{4}} \mbox{ if } \sigma=\sqrt{d}\\
 dA \left( \left( \frac{\pi d^2}{\sigma^2} + d \right)(1+\bar\beta) + 2 + \frac{2}{1 - e^{-\frac{\pi d}{\sigma^2}}} + \bar\alpha \left( \pi d + \sigma^2 + \frac{\pi\bar\beta}{1 - e^{-\frac{\pi d}{\sigma^2}}} \right) \right) e^{-\frac{\pi d^2}{4\sigma^2}}	 \mbox{ otherwise}
\end{cases}
\ee

Consider now the alternate product $\hat{t}_c \hat{H}_c \ket{\Psi(k_0)}$. To begin with, we convert the state into the energy basis,
\begin{align}
	\ket{\Psi(k_0)} &= \sumnsym \sum_{k \in \mathcal{S}_d(k_0)} \psi(k_0;k) \frac{e^{-i2\pi nk/d}}{\sqrt{d}} \ket{E_n}.
\end{align}

Converting the summation over $k$ into the infinite sum and bounding the difference,
\begin{align}
	\ket{\Psi(k_0)} &= \sumnsym \sum_{k \in \mathbb{Z}} \psi(k_0;k) \frac{e^{-i2\pi nk/d}}{\sqrt{d}} \ket{E_n} + \ket{\epsilon_5}, \\
	\text{where} \; \ket{\epsilon_5} &= \sumnsym \sum_{k \in \mathbb{Z}/\mathcal{S}_d(k_0)} \psi(k_0;k) \frac{e^{-i2\pi nk/d}}{\sqrt{d}} \ket{E_n},
\end{align}
\be\label{eq:ep5 comm bound}
\epsilon_5^{co}:= \ltwo{\ket{\epsilon_5}} 
	< 
	\begin{cases}   2\sqrt{d} A \frac{ e^{-\frac{\pi d}{4}}}{1 - e^{-\pi}} \mbox{ if } \sigma=\sqrt{d}\\
 2\sqrt{d} A \frac{ e^{-\frac{\pi d^2}{4\sigma^2}}}{1 - e^{-\frac{\pi d}{\sigma^2}}} 	 \mbox{ otherwise}
\end{cases}
\ee

Applying the Poissonian summation formula on the infinite sum,
\begin{align}
	\ket{\Psi(k_0)} &= \sumnsym \sum_{m \in \mathbb{Z}} \tilde{\psi}(k_0;n+md) \ket{E_n} + \ket{\epsilon_5}.
\end{align}

We now apply the Hamiltonian,
\begin{align}
	\hat{H}_c \ket{\Psi(k_0)} &= \sumnsym \sum_{m \in \mathbb{Z}} \frac{2\pi}{T_0} n \tilde{\psi}(k_0;n+md) \ket{E_n} + \hat{H}_c \ket{\epsilon_5}.
\end{align}

In anticipation of applying the time operator, we shift back to the basis of time-states,
\begin{align}
	\hat{H}_c \ket{\Psi(k_0)} &= \sum_{l \in \mathcal{S}_d(k_0)} \sumnsym \sum_{m \in \mathbb{Z}} \frac{2\pi}{T_0} n \tilde{\psi}(k_0;n+md) \frac{e^{+i2\pi nl/d}}{\sqrt{d}} \ket{\theta_l} + \hat{H}_c \ket{\epsilon_5}.
\end{align}

Combining the sums over $n$ and $m$,
\begin{align}
	\hat{H}_c \ket{\Psi(k_0)} &= \sum_{l \in \mathcal{S}_d(k_0)} \sum_{s \in \mathbb{Z}} \frac{2\pi}{T_0} s \tilde{\psi}(k_0;s) \frac{e^{+i2\pi sl/d}}{\sqrt{d}} \ket{\theta_l} + \ket{\epsilon_6} + \hat{H}_c \ket{\epsilon_5}, \\
	\text{where} \; \ket{\epsilon_6} &= -d \sum_{l \in \mathcal{S}_d(k_0)} \sumnsym \sum_{m \in \mathbb{Z}} \frac{2\pi}{T_0} m \tilde{\psi}(k_0;n+md) \frac{e^{+i2\pi nl/d}}{\sqrt{d}} \ket{\theta_l},
\end{align}
\be\label{eq:ep6 comm bound}
\epsilon_6^{co}:= \ltwo{\ket{\epsilon_6}} T_0 
	< 
	\begin{cases}  2\pi A d^2 \sqrt{d} \left( \frac{1-\bar\alpha}{2} + \frac{1}{2\pi d} + \left( \frac{1+\bar\alpha}{2} \right) \frac{1}{1 - e^{-\pi d(1-\bar\alpha)}} \right) e^{-\frac{\pi d}{4}(1-\bar\alpha)^2} \mbox{ if } \sigma=\sqrt{d}\\
 2\pi A d^2 \sigma \left( \frac{1-\bar\alpha}{2} + \frac{1}{2\pi\sigma^2} + \left( \frac{1+\bar\alpha}{2} \right) \frac{1}{1 - e^{-\pi\sigma^2(1-\bar\alpha)}} \right) e^{-\frac{\pi\sigma^2}{4}(1-\bar\alpha)^2}	 \mbox{ otherwise}
\end{cases}
\ee

Applying the Poissonian summation formula,
\begin{align}
	\hat{H}_c \ket{\Psi(k_0)} &= -\frac{id}{2\pi}  \frac{2\pi}{T_0} \sum_{l \in \mathcal{S}_d(k_0)} \sum_{m \in \mathbb{Z}} \frac{d}{dx} \psi(k_0;x) \bigg|_{x=l+md} \ket{\theta_l} + \ket{\epsilon_6} + \hat{H}_c \ket{\epsilon_5}.
\end{align}

Approximating the sum over $m$ by the $m=0$ term,
\begin{align}
	\hat{H}_c \ket{\Psi(k_0)} &= -\frac{id}{T_0} \sum_{l \in \mathcal{S}_d(k_0)} \frac{d}{dx} \psi(k_0;x) \bigg|_{x=l} \ket{\theta_l} + \ket{\epsilon_7} + \ket{\epsilon_6} + \hat{H}_c \ket{\epsilon_5}, \\
	\text{where} \; \ket{\epsilon_7} &= -\frac{id}{T_0} \sum_{l \in \mathcal{S}_d(k_0)} \sum_{m \in \mathbb{Z}/\{0\}} \frac{d}{dx} \psi(k_0;x) \bigg|_{x=l+md} \ket{\theta_l},
\end{align}
\be\label{eq:ep7 comm bound}
\epsilon_7^{co}:= \ltwo{\ket{\epsilon_7}}T_0 
	< 
	\begin{cases}   2\pi d A \left( 1 + \frac{1}{\pi} + \frac{\bar\alpha}{1 - e^{-\pi}} \right) e^{-\frac{\pi d}{4}} \mbox{ if } \sigma=\sqrt{d}\\
2\pi d A \left( \frac{d}{\sigma^2} + \frac{1}{\pi} + \frac{\bar\alpha}{1 - e^{-\frac{\pi d}{\sigma^2}}} \right) e^{-\frac{\pi d^2}{4\sigma^2}} 	 \mbox{ otherwise}
\end{cases}
\ee

We now apply the time operator,
\begin{align}
	\hat{t}_c \hat{H}_c \ket{\Psi(k_0)} &= -i \sumt{l} \sumS{k}{k_0} \frac{d}{dx} \psi(k_0;x) \bigg|_{x=k} \ket{\theta_l} + \hat{t}_c \left( \ket{\epsilon_7} + \ket{\epsilon_6} + \hat{H}_c \ket{\epsilon_5} \right).
\end{align}

As before, we split the sum case wise, we find
\begin{align}
	\hat{t}_c \hat{H}_c \ket{\Psi(k_0)} &= -i \sum_{l \in \mathcal{S}_d(k_0)} l \frac{d}{dx} \psi(k_0;x) \bigg|_{x=l} \ket{\theta_l} + \ket{\epsilon_8} + \hat{t}_c \left( \ket{\epsilon_7} + \ket{\epsilon_6} + \hat{H}_c \ket{\epsilon_5} \right), \\
	\text{where} \; \ket{\epsilon_8} &= id \sum_{l \in \mathcal{S}_d(k_0),l > d/2} \frac{d}{dx} \psi(k_0;x) \bigg|_{x=l} \ket{\theta_l},
\end{align}
\be\label{eq:ep8 comm bound}
\epsilon_8^{co}:= \ltwo{\ket{\epsilon_8}} 
	< 
	\begin{cases}  \pi d A \left( 1 -\bar\beta + \frac{1}{\pi} + \frac{\bar\alpha}{1 - e^{-\pi(1-\bar\beta)}} \right) e^{-\frac{\pi d}{4}(1-\bar\beta)^2} \mbox{ if } \sigma=\sqrt{d}\\
 \pi d A \left( \frac{d}{\sigma^2}(1-\bar\beta) + \frac{1}{\pi} + \frac{\bar\alpha}{1 - e^{-\frac{\pi d}{\sigma^2}(1-\bar\beta)}} \right) e^{-\frac{\pi d^2}{4\sigma^2}(1-\bar\beta)^2}	 \mbox{ otherwise}
\end{cases}
\ee

We may thus calculate the commutator on the state,
\begin{align}
	\left[ \hat{t}_c, \hat{H}_c \right] \ket{\Psi(k_0)} &= i \ket{\Psi(k_0)} + \ket{\epsilon_\textup{comm}} , \\
	\text{where} \quad \ket{\epsilon_\textup{comm}} &= \ket{\epsilon_8} + \hat{t}_c \left( \ket{\epsilon_7} + \ket{\epsilon_6} + \hat{H}_c \ket{\epsilon_5} \right) - \left( \ket{\epsilon_4} + \ket{\epsilon_3} + \hat{H}_c \left( \ket{\epsilon_2} + \ket{\epsilon_1} \right) \right).
\end{align}

The operator norms of $\hat{t}_c$ and $\hat{H}_c$ are bounded by $T_0/2$ and $\pi d/T_0$ respectively, and we can thus bound the norm of the total error, 
\be 
\|\ket{\epsilon_\textup{comm}}\|_2 \leq \,  \epsilon_8^{co}+\frac{1}{2}\left( \epsilon_7^{co}+\epsilon_6^{co}+\pi d\, \epsilon_5^{co} \right)+\epsilon_4^{co}+\epsilon_3^{co}+\pi d \left( \epsilon_2^{co}+\epsilon_1^{co} \right).
\ee 
\end{proof}

\section{Conjectures and open questions}\label{sec:conjectures}
In this section we state a conjecture about the tightness of our results. This is based on numerical studies and some intuition. We will also discuss some open questions about the properties of the bounds.\\

\textit{Conjecture 1: tightness in exponential decay of clock quasi-continuity}. Based on numerical studies, we conjecture  that the norm of $\ket{\epsilon}$ in Theorem \ref{gaussiancontinuity}, does not decay super-exponentially for all $t\in[0,T_0]$. In other words, there does \textit{not} exist a $\ket{\Psi_\textup{nor}(k_0)}\in\Lambda_{\sigma,n_0}$ and an $f: \rr \rightarrow \rr^+$, with $\lim_{d\rightarrow \infty} d/f(d)=0$ such that
\be 
\lim_{d\rightarrow \infty} \|\ket{\epsilon}\|_2 \,\me^{f(d)}= 0 \quad \forall \,\,t\in[0,T_0],
\ee
where $T_0>0$ and $\ket{\epsilon}$ is defined in Theorem \ref{gaussiancontinuity}.\\

If the conjecture is indeed correct, it would prove the tightness of the exponential decay bound for symmetric states, the next related interesting questions would be whether this bound is only achievable for symmetric states and whether the factor of $\pi/4$ is optimal for completely symmetric states. Furthermore, since the family of \gClock~states $\ket{\Psi_\textup{nor}(k_0)}\in\Lambda_{\sigma,n_0}$ are minimum uncertainty states\footnote{Up to an exponentially decaying correction in clock dimension.}, we propose that out of all normalized states $\ket{\psi}$ in $\mathcal{H}_c$ satisfying $[\hat t_c,\hat H_c]\ket{\psi}=\mi\ket{\psi}+\ket{\epsilon}$;  their corresponding error $\|\ket{\epsilon}\|_2$ will decay, at most, exponentially fast in $d$. As such, we expect that the exponential decay in clock dimension is a fundamental limitation of finite dimensional clocks.\\

There are also related interesting questions concerning the continuous  quasi-control too. For example, we know that after one period, $t=T_0$, when the potential is zero, $V_0(x)=0,\;\;\forall \,x\in\rr$, that the clock is returned to its initial state $\ket{\bar\Psi_\textup{nor}(k_0,\Delta)}$. We know that whenever $V_0(x)\neq 0$ for some $x\in\rr$ that this is not true. While Theorem \ref{movig through finite time} upper bounds how small this error is, what is not so clear is how quickly the error $\|\ket{\epsilon}\|_2$ in Theorem \ref{movig through finite time} goes to zero in the limit that $\| V_0\|_2\rightarrow 0$. Such scaling is numerically challenging to estimate, and we thus do not propose an answer. e.g. it  would be interesting to know whether it is power law decay or exponential decay in $1/\|V_0\|_2$. 

\acknowledgments
\vspace{-0.5cm}
The authors would like to thank Sandu Popescu, Michael Berry, Micha\l~Horodecki, William D. Matthews, for helpful discussions. MW and JO acknowledge support from the EPSRC, Royal Society and the COST network (Action MP1209). RS acknowledges support from the Swiss National Science Foundation (grant
PP00P2 138917 and QSIT).


\newpage
\begin{appendices} 
	
\section{The idealized momentum clock}\label{idealizedclock}
In this section we will elaborate further on the idealised clock discussed in Section \ref{sec:ideal}.

\bigskip
\begin{remark}\label{rem:Paulidiscussion}

Some authors disagree with Pauli's conclusion that no perfect physical time operator exists in quantum mechanics by pointing out that he only considered generators of Weyl pairs in his analysis \cite{wong}. Alternatively, one could demand the weaker condition that the time operator measures time within a finite time interval only, thus replacing $\forall t \in \rr$ with $\forall\, t \in [0,t_\textup{max}]$ in Eq. \eqref{eq:div hat t}. Under such conditions, $\hat t$ and $\hat H$ can form Heisenberg pairs and a time operator $\hat t$ with spectrum $[0,t_\textup{max}]$ with a bounded from below Hamiltonian $\hat H$ exists, however the restriction of the spectrum of $\hat t$ to a finite interval requires \textit{confinement}, that is, the probability of finding the wave-function at the boundaries $0$, $t_\textup{max}$ must vanish \cite{wong}. While mathematically such a construction exists, it is conceivable that such confinement can only physically be achieved when one has an infinite potential at the boundaries of the interval $[0,t_\textup{max}]$ in order to prevent quantum mechanical tunnelling, which again, would require infinite energy and is therefore unphysical (e.g. particle in a box \cite{reed1975methods}). Thus the statement still holds true that there cannot exist a perfect time operator without requiring either states of infinite energy or infinite potentials.

\end{remark}
\subsection{The ideal time operator and idealised clock}\label{The ideal time operator and clock}
Consider the idealised clock, described by the time operator $\hat t=\hat x_c$, and Hamiltonian $\hat H=\hat p_c$, where $\hat x_c$ and $\hat p_c$ are the canonically conjugate position and momentum operators of a free particle in one dimension with domain $D_0$ of infinitely differentiable functions of compact support on $L^2(\rr)$.\footnote{We have used this domain for simplicity, one can equip $\hat x_c$ and $\hat p_c$ with larger domains if desired.} All operators acting on the clock in this section and the following will be assumed to be on $D_0$. Due to the commutator,
\begin{equation}\label{idealcommutator}
-\mi [\hat{t},\hat{H}] = \id,
\end{equation}
it follows that in the Heisenberg picture, $\hat t$ satisfies
\begin{equation}\label{eq:t op appendix}
\frac{d}{dt} \hat{t}(t) = \id, \quad\forall t\in \rr.
\end{equation}


For the idealised clock the equation of motion \eqref{idealcommutator} may equivalently be expressed as a property of the state of the clock. In the Sch\"{o}dinger picture,
\begin{align}\label{idealcontinuity}
\braket{x|e^{-i\hat{H}t}|	\Psi} &= \braket{x-t|\Psi},
\end{align}
where $\ket{x}$, $x\in\rr$ is a generalised eigenvector of $\hat x_c$.
This is true for any $t\in\rr$, and we label this property \emph{continuity}.

\subsection{The idealised quantum control}
%
%
%
%
%
Quantum control requires more than accurate measurements of time, in addition the control must be able to influence another system without the intervention of an external observer, and at well-defined times. In this Section we review a well-known result regarding the idealised clock which will be used in Section \ref{sec:Automation via the idealised clock} to prove the idealised clock's ability to control. This result is 
that the idealised clock remains continuous under the action of a potential, i.e.
\begin{equation}\label{idealregularity}
\braket{x|e^{-i(\hat{p}_c + V(\hat{x}_c))t}|	\Psi} = e^{-i \int_{x-t}^x V(x^\prime) dx^\prime} \braket{x-t|\Psi},
\end{equation}
$V\in D_0$ 
,\, $x,t\in\rr$.
Labelling the wavefunction $\braket{x|e^{-i(\hat{p}_c + V(\hat{x}_c))t}|\Psi}$ as $\psi(x,t)$, then \eqref{idealregularity} is equivalent to
\begin{equation}\label{eq:psi idealised control}
\psi(x,t) = \psi(x-t,0) e^{-i \int_{x-t}^x V(x^\prime) dx^\prime}.
\end{equation}

We proceed to verify that it is the solution to Schrodinger's equation for the clock,
\begin{equation}\label{schrodingerequation}
i \frac{\partial}{\partial t} \psi(x,t) = -i \frac{\partial}{\partial x} \psi(x,t) + V(x) \psi(x,t).
\end{equation}

Direct calculation gives
\begin{align}
\frac{\partial}{\partial t} \psi(x,t) = &- \left( \frac{\partial}{\partial x} \psi(x-t,0) \right) e^{-i \int_{x-t}^x V(x^\prime) dx^\prime} \nonumber\\
&- i V(x-t) \psi(x,t), \\
\frac{\partial}{\partial x} \psi(x,t) = &+ \left( \frac{\partial}{\partial x} \psi(x-t,0) \right) e^{-i \int_{x-t}^x V(x^\prime) dx^\prime} \nonumber\\
&- i \left( V(x) - V(x-t) \right) \psi(x,t),
\end{align}
from which it is easily verified that $\psi(x,t)$ does indeed satisfy \eqref{schrodingerequation}.

\section{The Salecker-Wigner-Peres finite clock - Introduction and shortcomings}\label{SWPclock}

\subsection{Hamiltonian, time-states, and the time operator}

In this section we review the finite clock introduced by H. Salecker and E.P. Wigner \cite{SaleckerWigner}, and studied by A. Peres \cite{Peres}. We review the features that make this the model of choice for finite-dimensional clocks, as well as some of its apparent difficulties in copying the behaviour of the idealised clock.

The Hamiltonian of the Salecker-Wigner clock is equally spaced among $d$ energy levels,
\begin{equation}\label{finiteHamiltonian_main}
\hat{H}_c = \sum_{n=0}^{d-1} n \omega \ketbra{E_n}{E_n}.
\end{equation}

The frequency $\omega$ determines both the energy spacing as well as the time period of the clock, $T_0 = 2\pi/\omega$. It is easily verifiable that $e^{-i \hat{H}_c T_0} = \id_c$, and thus any state of the clock repeats itself after time $T_0$. (We will use $T_0$ rather than $\omega$ wherever possible).

Given the Salecker-Wigner Hamiltonian, one can construct a basis of \textit{time-states}, $\left\{\ket{\theta_k}\right\}_{k=0}^{d-1}$, mutually unbiased w.r.t. the energy states,
\begin{align}\label{finitetimestates}
\ket{\theta_k} &= \frac{1}{\sqrt{d}} \sum_{n=0}^{d-1} e^{-i2\pi n k/d} \ket{E_n}.
\end{align}

This is precisely the \emph{discrete Fourier transform}, (D.F.T.). It will also be useful to extend the range of $k$ to $\zz$. Extending the range of $k$ in Eq. \eqref{finitetimestates_main} it follows $\ket{\theta_k}=\ket{\theta_{k \textup{ mod. } d}}$ for $k\in\zz$. One may switch back to the energy basis via the inverse D.F.T.,
\begin{equation}\label{finiteenergystates}
\ket{E_n} = \frac{1}{\sqrt{d}} \sum_{k=0}^{d-1} e^{+i 2\pi nk/d} \ket{\theta_k}.
\end{equation}

The basis of time-states is orthonormal, $\braket{\theta_{k^\prime}|\theta_k}=\delta_{kk^\prime}$. The label \emph{time states} is assigned to them because they rotate into each other in regular time intervals of $T_0/d$, i.e. 
\begin{equation}\label{timestaterotation}
e^{-i \hat{H}_c T_0/d} \ket{\theta_k} = \ket{\theta_{k+1}},\quad k\in\zz.
\end{equation}
The rotation is cyclic, meaning that $\ket{\theta_{d-1}}$ is rotated into $\ket{\theta_0}$.

Since any state of the clock may be expressed in the basis of time-states, the rotation property is true for every state,
\begin{equation}\label{finiteregularity}
\braket{\theta_k | e^{-i \hat{H}_c \,m\!\, T_0 /d} |\Psi} = \braket{\theta_{k-m} | \Psi}, \quad k,m\in\zz. 
\end{equation}

This motivates the definition of a \textit{Time operator} analagous to the Energy operator \eqref{finiteHamiltonian_main},
\begin{equation}\label{finitetimeoperator}
\hat{t}_c = \sum_{k=0}^{d-1} k \frac{T_0}{d} \ketbra{\theta_k}{\theta_k}.
\end{equation}

Each time state $\ket{\theta_k}$ is an eigenstate of the time operator, with the eigenvalue equal to the time taken to rotate from $\ket{\theta_0}$ to $\ket{\theta_k}$. Given the connection between canonical position and momentum operators, and the cannonical time and Hamiltonian operators discussed here, one may wander whether such a connection exists between the finite dimensional time operator $\hat t_c$ and Hamiltonian operator $\hat H_c$. There are finite dimensional analogues to the position and momentum operators. These are discussed by Serge Massar and Philippe Spindel in \cite{sergedft}. They derive uncertainty relations for the D.F.T., as well as commutator relations. In fact, when Theorem \ref{Quasi-Canonical commutation} is written in terms of the discreate position and momentum operators of \cite{sergedft}, it proves an open conjecture in \cite{sergedft}.

\bigskip
\subsection{Shortcomings of the time-states}\label{Shortcomings of the time-states}
The first problem with the time-states is that their behaviour is not \emph{continuous} in any sense, i.e. the regular rotation of one angle state into another \eqref{finiteregularity} is only true for particular time intervals, unlike the continuity of the idealised clock \eqref{idealcontinuity}.

For intermediate times, as calculated by A. Peres \cite{Peres}, a time-state spreads out to a superposition of a number of time states,
\begin{equation}
e^{-i \hat{H}_c x \,T_c/d } \ket{\theta_k} = \sum_{l=0}^{d-1} \frac{1}{d} \left( \frac{1 - e^{-i 2\pi (k+x-l)}}{1 - e^{-i 2\pi (k+x-l)/d}} \right) \ket{\theta_l},\quad x\in[0,d],\,\,\, k\in\zz,
\end{equation}
In Fig. \ref{peresbehaviour}, we compare the behaviour of the time-state against the idealised case, for $d=8$.

\begin{figure}[h]
	\includegraphics[scale=0.6]
	{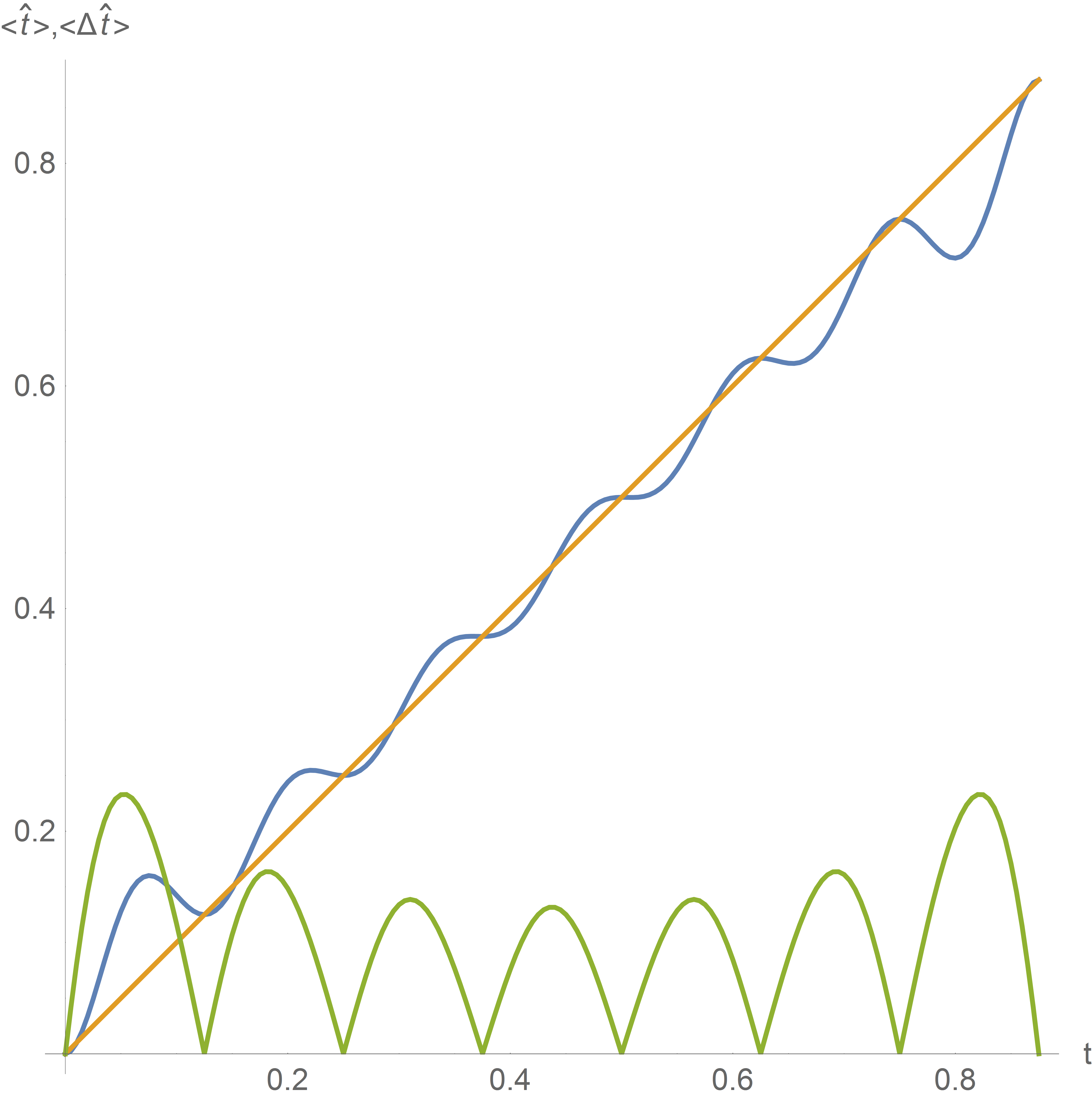}
	\caption{The expectation value (blue) and variance (green) of the time operator \eqref{finitetimeoperator} given the initial state $\ket{\theta_0}$, for $d=8$, $T_0 = 1$. The idealised case $\braket{\hat{t}} = t$ is in orange. \label{peresbehaviour}}
\end{figure}
%
%
%
%

A more fundamental shortcoming is that the time operator \eqref{finitetimeoperator} can never approach the ideal commutation relation \eqref{idealcommutator} with the Hamiltonian. As noted by Peres \cite{Peres},
\begin{equation}\label{eq:time sates commutator inner product}
\braket{\theta_{k} | [\hat{t}_c,\hat{H}_c] | \theta_k} = 0, \quad \forall \,\,k\in\zz, \quad \forall\,\, d\in\nn^+.
\end{equation}

This is consistent with the observation by H. Weyl \cite{Weyl} that the canonical commutation relation cannot be obeyed by finite dimensional operators.

However, while the example of Eq. \eqref{eq:time sates commutator inner product} demonstrates that ideal commutator is impossible to achieve for operators with domain on the full $d$ dimensional Hilbert space of the clock (even in the $d\rightarrow \infty$ limit), it is possible to approximately achieve this after one restricts the domain of the operators to a sub domain. The domain defined by the linear span of the clock states (Def. \ref{def:Gaussian clock states}), which excludes pure angle states $\ket{\theta_k}$, is one such example, i.e.
\begin{equation}\label{eq:commutator on gaussian space aprox}
[\hat{t}_c,\hat{H}_c] \ket{\Psi} \approx i \ket{\Psi},
\end{equation}
for clock states $\ket{\Psi}$, and where the approximation becomes exact in the $d\rightarrow \infty$ limit.
In Section \ref{appendixcommutator} a rigorous version of Eq. \eqref{eq:commutator on gaussian space aprox} was derived.

\section{Mathematical results used in the proofs}\label{mathidentities}
There are many mathematical results used in the proof of this \comm. Here we state the ones which are used repetitively in many proofs. When a  mathematical result is used only once or twice; or confined to a particular lemma/theorem, it has been stated directly in the proof.\\

\subsection{Fourier transform as a function of dimension $d$}

\begin{definition} [Standard fourier transform] \cite{classicalfourieranalysis} Given $f \in \ell(\rr)$\footnote{$\ell$ refers to the Schwartz space, i.e. the function space of functions all of whose derivatives decrease faster than any polynomial.\cite{classicalfourieranalysis} Gaussian functions are easily seen to be within this class.} its Fourier transform is defined as
\begin{equation}
	\bar{f}(\zeta) = \int_{\rr} f(x) e^{-i 2\pi x \zeta} dx.
\end{equation}
\end{definition}

In this article, we do not use this version of the Fourier transform, rather we work extensively with the following Fourier transform with an additional parameter $d$,
\begin{definition}
\begin{equation}\label{def:dft}
	\tilde{\psi}(p) = \frac{1}{\sqrt{d}} \int_{\rr} \psi(x) e^{-i 2\pi p x /d} dx,
\end{equation}
\end{definition}
which is closely related to the \emph{Discrete Fourier transform}.

Throughout the article, the words `Fourier transform' and the shorthand $\mathcal{F}_d$ will refer to \eqref{def:dft}.
\begin{corollary}[Useful Fourier relations] If $\cft{\psi(x)} = \tilde{\psi}(p)$,
\begin{align}
	\psi(x) = \cfti{\tilde{\psi}(p)} &= \frac{1}{\sqrt{d}} \int_{-\infty}^\infty \tilde{\psi}(p) e^{+i2\pi px/d} dp
\end{align}
\begin{align}
	\cft{\psi(x)e^{-i 2\pi xa/d}} &= \tilde{\psi}(p+a), \\
	\cfti{\tilde{\psi}(p) e^{i 2\pi pb/d}} &= \psi(x+b)
\end{align}
\begin{align}
	\cft{x \psi(x)} &= \left( \frac{id}{2\pi} \right) \frac{d}{dp} \tilde{\psi}(p), \\
	\cfti{p \tilde{\psi(p)}} &= \left( -\frac{id}{2\pi} \right) \frac{d}{dx} \tilde{\psi}(x).\label{eq: inv Fourier = dev}
\end{align}
\end{corollary}

\subsection{Poisson summation formula}\label{Poisson summation formula}

We first present the formula for the case of the standard Fourier transform, and then modify it to the version we need.
\begin{lemma}
Suppose that a function $f$ and its (standard) Fourier transform $\bar{f}$ belong to $L^1(\rr^n)$ and satisfy
\begin{equation}\label{poissoncriterion}
	\left| f(x) \right| + \left| \bar{f}(x) \right| \leq C \left( 1 + x \right)^{-n-\delta}
\end{equation}
for some $C,\delta>0$. Then $f$ and $\bar{f}$ are both continuous, and $\forall x \in \rr^n$ we have
\begin{equation}
	\sum\limits_{m\in\mathbb{Z}^n} \bar{f}(m) = \sum\limits_{m\in\mathbb{Z}^n} f(m).
\end{equation}
\end{lemma}

Note that any function in the Schwartz space\cite{classicalfourieranalysis} satisfies \eqref{poissoncriterion}. In this article we use the Poisson summation only upon Gaussian functions or Gaussians multipled by functions of bounded derivatives, both of which are members of the Schwartz space.

\begin{corollary}[Poisson summation for  $\mathcal{F}_d$]\label{poissonsummation} If $\psi \in \ell(\rr)$,
\begin{align}
	\sum_{m \in \mathbb{Z}} \psi(m) &= \sqrt{d} \sum_{m \in \mathbb{Z}} \tilde{\psi}(md), \text{ and } \\
	\sum_{m \in \mathbb{Z}} \tilde{\psi}(m) &= \sqrt{d} \sum_{m \in \mathbb{Z}} \psi(md).
\end{align}
\end{corollary}
\onecolumngrid

\subsection{Addition of errors}
\begin{lemma}\label{unitaryerroraddition} \emph{(norm non-increasing errors add linearly)}
	Consider a sequence of operators $\{\Delta_m\}_{m=1}^N$ on a finite dimensional Hilbert space $\mathcal{H}$ which are normalised or sub-normlaised w.r.t. the induced operator norm from the $l_2$ vector norm, i.e. $\| \Delta_m\|_2\leq 1$, and let $\{\ket{\Phi_m}\}_{m=0}^N$ be a sequence of (not necessarily normalised) pure states in $\mathcal{H}$ with the following property:
	\begin{equation}
	\norm{\,\ket{\Phi_m} - \Delta_m \ket{\Phi_{m-1}} } =  \epsilon_m,
	\end{equation}
	where $\|\cdot\|_2$ is the $l_2$ vector norm. Then $\forall\, n\in 1,2,3,\ldots, N$ we have
	\begin{equation}
	\norm{\, \ket{\Phi_n} - \Delta_n \Delta_{n-1}\ldots \Delta_1 \ket{\Phi_0} } \leq \sum_{m=1}^n \epsilon_m.
	\end{equation}
\end{lemma}

\begin{proof} By induction. The theorem is true by definition for $n=1$, and if the theorem is true for all $n$ up to $k$, then for $n=k+1$,
	\begin{align}
	\norm{ \; \ket{\Phi_{k+1}} - \Delta_{k+1}\Delta_{k}\ldots\Delta_1 \ket{\Phi_0} \; } &= \norm{ \; \ket{\Phi_{k+1}} - \Delta_{k+1}\ket{\Phi_k} + \Delta_{k+1} \left( \ket{\Phi_k} - \Delta_{k}\ldots\Delta_1 \ket{\Phi_0} \right) \; } \\
	&\leq \norm{ \; \ket{\Phi_{k+1}} - \Delta_{k+1}\ket{\Phi_k} \; } + \norm{ \; \Delta_{k+1} \left( \ket{\Phi_k} - \Delta_{k}\ldots\Delta_1 \ket{\Phi_0} \right) \; } \\
	&\leq \norm{ \; \ket{\Phi_{k+1}} - \Delta_{k+1}\ket{\Phi_k} \; } + \norm{\Delta_{k+1}}\, \norm{ \; \left( \ket{\Phi_k} - \Delta_{k}\ldots\Delta_1 \ket{\Phi_0} \right) \; }\label{line errors add linearly} \\
	&\leq \norm{ \; \ket{\Phi_{k+1}} - \Delta_{k+1}\ket{\Phi_k} \; } +  \norm{ \,  \ket{\Phi_k} - \Delta_{k}\ldots\Delta_1 \ket{\Phi_0}  } \\
	&= \epsilon_{k+1} + \sum_{m=1}^k \epsilon_m = \sum_{m=1}^{k+1} \epsilon_m,
	\end{align}
	where we used the Minkowski vector norm inequality and the equivalence between the induced $l_2$ operator norm and the property $\| \Delta_m\|_2\leq 1$ in line \eqref{line errors add linearly}.\\
\end{proof}

\onecolumngrid

\section{A proof for the upper bound to $\tilde \epsilon_V$}\label{sec:ep V ex pot bound}
Here we present a proof of Eq. \eqref{eq: up bound tilde epsilon V}; namely the inequality
\be\label{eq: up bound tilde epsilon V 2}
\tilde \epsilon_V\leq \frac{(\pi-x_{vr})\me^2}{4\pi\sqrt{\pi}}\sqrt{n}\cos^{2n}(x_{vr}/2), \quad\text{if } \cos(x_{vr})\leq 1-\frac{1}{n}.
\ee
\begin{proof}
	From Eqs. \eqref{eq:tilde ep equal} and \eqref{def: cosine pot} we find
	\begin{align}\label{eq: tilde ep V int written in convex form}
	\tilde \epsilon_V&=\int_{x_{vr}-2\pi}^{x_{vl}}dx A_c \cos^{2 n}\left(\frac{x}{2}\right)=A_c \left( \int_{x_{vr}-2\pi+2\pi}^{-\pi+2\pi} dx \cos^{2n}(x/2-\pi)-\int_{\pi}^{-x_{vl}}dx \cos^{2n}(-x/2) \right)\\
	&= 2 A_c\int_{x_{vr}}^\pi dx \cos^{2 n}(x/2)=2 A_c (\pi-x_{vr})\int_0^1 dy \cos^{2 n}\left( (y\, x_{vr}+(1-y) \pi)/2 \right).
	\end{align}
	We now note the convexity of the potential, namely on the interval $x\in(x_{vr},\pi)$,
	\be 
	\frac{d^2\;}{dx^2} \cos^{2n}(x/2)=\frac{n}{2}\cos^{2(n-1)}(x/2) \left(n(1-\cos(x))-1\right)\geq 0 \implies \cos(x)\leq 1-\frac{1}{n},
	\ee
	thus noting that $\max_{x\in[x_{vr},\pi]}\cos(x)=\cos(x_{vr})$, it follows that $\cos^{2n}(x/2)$ is convex on $x\in[x_{vr},\pi]$ if $\cos(x_{vr})\leq 1-1/n$ and $ \cos^{2 n}\left( (y\, x_{vr}+(1-y) \pi)/2 \right)\leq  y \cos^{2 n}(x_{vr}/2)+(1-y) \cos^{2 n}(\pi/2) =y \cos^{2 n}(x_{vr}/2)$, $y\in[0,1]$. Hence using Eq. \eqref{eq: tilde ep V int written in convex form}, we conclude
	\be 
	\tilde \epsilon_V\leq A_c(\pi-x_{vr})\cos^{2n}(x_{vr}/2), \quad\text{if } \cos(x_{vr})\leq 1-\frac{1}{n}.
	\ee
	Using Eq. \eqref{def: A c eq def for cos pot} with $\Omega=1$ and Sterling's formula, we find
	\be 
	A_c= \frac{2^{2 n}}{2\pi \binom{2n}{n}}=\frac{2 ^{2 n}}{2\pi}\frac{(n!)^2}{(2n)!}\leq \frac{\me^2}{4\pi\sqrt{\pi}}\sqrt{n},
	\ee 
	thus giving us Eq. \eqref{eq: up bound tilde epsilon V} (Eq. \eqref{eq: up bound tilde epsilon V 2} above).
\end{proof}

\section{Error Bounds}\label{Error Bounds}

\subsection{The norm of a discretized Gaussian (Normalization of $\ket{\Psi(k_0)}$)}
In this section only, for simplicity, we further restrict the range of $\sigma\in(0,d)$ and $d$ to $\sigma\in[1,d)$ and $d=2,3,4,\ldots$. 
\subsubsection{Normalizing $\ket{\Psi(k_0)}$}\label{Normalizing the clock state}
We calculate the norm of a state in the space $\Lambda_{\sigma,n_0}$:.
\begin{align}
	\ltwo{\ket{\Psi(k_0)}}^2 &= A^2 \sum_{k \in \mathcal{S}_d(k_0)} e^{-\frac{2\pi}{\sigma^2}(k-k_0)^2} \\
	&= A^2 \left( \sum_{k\in \mathbb{Z}} e^{-\frac{2\pi}{\sigma^2}(k-k_0)^2} + \epsilon_1 \right), \\
	\text{where} \abs{\epsilon_1} &= \sum_{k\in \mathbb{Z}/\mathcal{S}_d(k_0)} e^{-\frac{2\pi}{\sigma^2}(k-k_0)^2} < \frac{2 e^{-\frac{\pi d^2}{2\sigma^2}}}{1 - e^{-\frac{2\pi d}{\sigma^2}}}:=\bar\epsilon_1,\label{eq:norm bound ep bar 1}
\end{align}
using results from Sec. \ref{gaussianbounds}. Applying the Poissonian summation formula on the sum,
\begin{align}
	\ltwo{\ket{\Psi(k_0)}}^2 &= A^2 \left( \frac{\sigma}{\sqrt{2}} \sum_{m \in \mathbb{Z}} e^{-\frac{\pi \sigma^2 (md)^2}{2d^2}} e^{-i 2\pi (md) k_0/d} + \epsilon_1 \right) \\
	&= A^2 \left( \frac{\sigma}{\sqrt{2}} + \epsilon_2 + \epsilon_1 \right), \\
	\text{where} \abs{\epsilon_2} &\leq \frac{\sigma}{\sqrt{2}}\sum_{m\in\mathbb{Z}-\{0\}} e^{-\frac{\pi \sigma^2 m^2}{2}} < \frac{\sigma}{\sqrt{2}}\,\frac{2 e^{-\frac{\pi \sigma^2}{2}}}{1 - e^{-\pi \sigma^2}}:=\bar\epsilon_2.\label{eq:norm bound ep bar 2}
\end{align}

Thus for $\ket{\Psi_\textup{nor}(k_0)}$ defined in Def. \ref{def:Gaussian clock states}, we have that $A$ given by Eq. \eqref{eq:A normalised} satisfies
\be\label{eq:up low bounds for A normalize}
\left(\frac{2}{\sigma^2}\right)^{1/2}-\frac{\bar\epsilon_1+\bar\epsilon_2}{\frac{\sigma}{\sqrt{2}}\left(\frac{\sigma}{\sqrt{2}}+\bar\epsilon_1+\bar\epsilon_2 \right)}\leq A^2\leq \left(\frac{2}{\sigma^2}\right)^{1/2}+\frac{\bar\epsilon_1+\bar\epsilon_2}{\frac{\sigma}{\sqrt{2}}\left(\frac{\sigma}{\sqrt{2}}-\bar\epsilon_1-\bar\epsilon_2 \right)},
\ee
where $\bar\epsilon_1$, $\bar\epsilon_2$ are given by Eqs. \eqref{eq:norm bound ep bar 1} and \eqref{eq:norm bound ep bar 2} respectively. Note that since $\sigma\geq 1$, $d/\sigma \geq 2$, it follows that $\sigma/\sqrt{2}-\bar\epsilon_1-\bar\epsilon_2>0$ so that Eq. \eqref{eq:up low bounds for A normalize} has no singular points.

Of the two bounds $\bar\epsilon_1$ and $\bar\epsilon_2$, the first is the greater error for $\sigma>\sqrt{d}$, while the second dominates for $\sigma<\sqrt{d}$.

\subsubsection{Re-normalizing $\ket{\Psi(k_0)}$}\label{Re-normalizing the clock state}
In some instances we will have to re-normalize the state $\ket{\Psi(k_0)}$. This will consist in upper bounding $A/A'$ where both $A$ and $A'$ satisfy Eq. \eqref{eq:A normalised} but for different values of $k_0$. We will thus want to upper bound
\be 
\epsilon_\textit{A}=\left|\frac{A}{A'}-1\right|.
\ee
re-writing Eq. \eqref{eq:up low bounds for A normalize} using the short hand $a-\epsilon_L\leq A^2\leq a+\epsilon_R$, we can write
\be 
\epsilon_\textit{A}\leq \sqrt{\frac{a+\epsilon_R}{a-\epsilon_L}}-1= \sqrt{1+\frac{\epsilon_L+\epsilon_R}{a-\epsilon_L}}-1\leq \frac{\epsilon_L+\epsilon_R}{a-\epsilon_L}= \frac{2\sqrt{2}}{\sigma}\,\frac{\bar\epsilon_1+\bar\epsilon_2}{1-2(\bar\epsilon_1+\bar\epsilon_2)/\sigma}<  \frac{20\sqrt{2}}{3}\,\frac{\bar\epsilon_1+\bar\epsilon_2}{\sigma},
\ee
where in the last line we have converted back to the notation of Eq. \eqref{eq:up low bounds for A normalize} and simplified the expression. Recall that $\bar\epsilon_1$, $\bar\epsilon_2$ are given by Eqs. \eqref{eq:norm bound ep bar 1} and \eqref{eq:norm bound ep bar 2} respectively.

\onecolumngrid

\subsection{Bounds on the tails of discrete Gaussians}\label{gaussianbounds}
In this section we state some well known useful bounds which will be used throughout the proofs in appendix. For $\Delta\in\rr$, we have the following bounds on the summations over Gaussian tails.
\begin{lemma}\label{G0}
\begin{equation}
	\sum_{n=a}^\infty e^{-\frac{(n-X)^2}{\Delta^2}} < \frac{e^{-\frac{(a-X)^2}{\Delta^2}}}{1-e^{-\frac{2(a-X)}{\Delta^2}}},\quad \text{for } a>X\in\rr
\end{equation}
\end{lemma}

\begin{proof}
\begin{align}
	\sum_{n=a}^\infty e^{-\frac{(n-X)^2}{\Delta^2}} &= \sum_{m=0}^\infty e^{-\frac{(a-X+m)^2}{\Delta^2}}= e^{-\frac{(a-X)^2}{\Delta^2}} \sum_{m=0}^\infty e^{-\frac{2m(a-X)}{\Delta^2}} e^{-\frac{m^2}{\Delta^2}} \\
	&< e^{-\frac{(a-X)^2}{\Delta^2}} \sum_{m=0}^\infty e^{-\frac{2m(a-X)}{\Delta^2}} = \frac{e^{-\frac{(a-X)^2}{\Delta^2}}}{1-e^{-\frac{2(a-X)}{\Delta^2}}}
\end{align}
\end{proof}


\begin{lemma}\label{G1}
\begin{equation}
	\sum_{n=a}^\infty (n-X) e^{-\frac{(n-X)^2}{\Delta^2}} < \left( a - X + \frac{\Delta^2}{2} \right) e^{-\frac{(a-X)^2}{\Delta^2}},\quad \text{for } a>X+\Delta\in\rr
\end{equation}

\end{lemma}

\begin{proof}
\begin{align}
	\sum_{n=a}^\infty (n-X) e^{-\frac{(n-X)^2}{\Delta^2}} &= (a-X) e^{-\frac{(a-X)^2}{\Delta^2}} + \sum_{n=a+1}^\infty (n-X) e^{-\frac{(n-X)^2}{\Delta^2}}
\end{align}

Since $(x-X)e^{-\frac{(x-X)^2}{\Delta^2}}$ is monotonically decreasing for $x>X+\Delta$, (which $a$ satisfies),
\begin{equation}
	\sum_{n=a+1}^\infty (n-X) e^{-\frac{(n-X)^2}{\Delta^2}} < \int_a^\infty (x-X) e^{-\frac{(x-X)^2}{\Delta^2}} dx = \frac{\Delta^2}{2} e^{-\frac{(a-X)^2}{\Delta^2}},
\end{equation}
where we have used
\be 
\int_b^\infty x e^{-\frac{x^2}{\Delta^2}} dx = \frac{\Delta^2}{2} e^{-\frac{b^2}{\Delta^2}}.
\ee
\end{proof}

\begin{lemma}
\begin{align}
	\sum_{n=a}^\infty (n-X)^2 e^{-\frac{(n-X)^2}{\Delta^2}} < \left( (a-X)^2 + \frac{\Delta^2}{2} \left( a - X + \frac{1}{1 - e^{-\frac{2(a-X)}{\Delta^2}}} \right) \right) e^{-\frac{(a-X)^2}{\Delta^2}},\quad \text{for } a>X+\sqrt{2} \Delta\in\rr
\end{align}
\end{lemma}

\begin{proof}
\begin{equation}
	\sum_{n=a}^\infty (n-X)^2 e^{-\frac{(n-X)^2}{\Delta^2}} = (a-X)^2 e^{-\frac{(a-X)^2}{\Delta^2}} + \sum_{n=a+1}^\infty (n-X)^2 e^{-\frac{(n-X)^2}{\Delta^2}}
\end{equation}

Since $(x-X)^2e^{-\frac{(x-X)^2}{\Delta^2}}$ is monotonically decreasing for $x>X+\sqrt{2}\Delta$, (which $a$ satisfies),
\begin{align}
	\sum_{n=a+1}^\infty (n-X)^2 e^{-\frac{(n-X)^2}{\Delta^2}} &< \int_{a}^\infty (x-X)^2 e^{-\frac{(x-X)^2}{\Delta^2}} dx  \\
	&= \frac{\Delta^2}{2} \left( (a-X) e^{-\frac{(a-X)^2}{\Delta^2}} + \int_a^\infty e^{-\frac{(x-X)^2}{\Delta^2}} dx \right)
\end{align}

We use the monotonicity again, together with lemma \ref{G0},
\begin{align}
	\int_a^\infty e^{-\frac{(x-X)^2}{\Delta^2}} dx &<\sum_{n=a}^\infty e^{-\frac{(n-X)^2}{\Delta^2}} <  \frac{e^{-\frac{(a-X)^2}{\Delta^2}}}{1-e^{-\frac{2(a-X)}{\Delta^2}}}
\end{align}
\end{proof}

\subsection{Commutator section}
Here we provide some supplementary details to the proofs of some of the bounds used in the proof of Theorem \ref{Quasi-Canonical commutation}.\\

\centerline{\textit{1. Bounding} $\quad\epsilon_1^{co}$}
\begin{align}
	\ltwo{\ket{\epsilon_1^{co}}} &< T_0 \sum_{k \in \mathcal{S}_d(k_0), k > \frac{d}{2}} \abs{\psi(k_0;k)} \\
	&< T_0 \sum_{k=\frac{d+1}{2}}^\infty \abs{\psi(k_0;k)} & & \text{increasing the range of the sums} \\
	&< T_0 \sum_{k=\frac{d}{2}}^\infty \abs{\psi(k_0;k)} & & \text{monotonicity (decreasing) of $\psi(k_0;k)$} \\
	&< T_0 \sum_{k=\frac{d}{2}}^\infty e^{-\frac{\pi d^2}{\sigma^2}(k-k_0)^2} \\
	&< T_0 A \frac{e^{-\frac{\pi}{\sigma^2}\left( \frac{d}{2} - k_0 \right)^2}}{1 - e^{-\frac{2\pi}{\sigma^2}\left( \frac{d}{2} - k_0 \right)}} \\
	&=
	\begin{cases}\displaystyle
 2 T_0 A \frac{\errort}{1 - \errortd}  &\mbox{if } \sigma=\sqrt{d} \\[10pt]\displaystyle
2 T_0 A \frac{e^{-\frac{\pi d^2}{4 \sigma^2}(1-\beta)^2}}{1 - e^{-\frac{\pi d}{\sigma^2}(1-\beta)}}  &\mbox{otherwise}
\end{cases}
\end{align}

\centerline{\textit{2. Bounding} $\quad\epsilon_2^{co}$}
\begin{align}
	\ltwo{\ket{\epsilon_2^{co}}} &= \frac{T_0}{d\sqrt{d}} \sum_{n=0}^{d-1} \sum_{k \in \mathbb{Z}/\mathcal{S}_d(k_0)} \abs{k\psi(k_0;k)} \\
	&< \frac{T_0}{\sqrt{d}} \sum_{k \in \mathbb{Z}/\mathcal{S}_d(k_0)} \abs{k \psi(k_0;k)} & & \text{trivial sum w.r.t. $n$} \\
	&< \frac{2 T_0 A}{\sqrt{d}} \sum_{k=k_0+\frac{d}{2}}  k e^{-\frac{\pi}{\sigma^2}(k-k_0)^2} & & \text{picking the larger error (right side)} \\
	&< \frac{2 T_0 A}{\sqrt{d}} \sum_{k=k_0+\frac{d}{2}} \left[ (k-k_0) + k_0 \right] e^{-\frac{\pi}{\sigma^2}(k-k_0)^2} & & \text{splitting into two sums} \\
	&< \frac{2 T_0 A}{\sqrt{d}} \left[ \frac{d}{2} + \frac{\sigma^2}{2\pi} + \frac{k_0}{1 - e^{-\frac{2\pi}{\sigma^2}\frac{d}{2}}} \right] e^{-\frac{\pi}{\sigma^2} \left(\frac{d}{2}\right)^2} \\
	&< 
	\begin{cases}\displaystyle
 T_0 A \sqrt{d} \left( 1 + \frac{1}{\pi} + \frac{\beta}{1 - e^{-\pi}} \right) e^{-\frac{\pi d}{4}}  &\mbox{if } \sigma=\sqrt{d} \\[10pt]\displaystyle
T_0 A \sqrt{d} \left( 1 + \frac{\sigma^2}{\pi d} + \frac{\beta}{1 - e^{-\frac{\pi d}{\sigma^2}}} \right) e^{-\frac{\pi d^2}{4\sigma^2}}  &\mbox{otherwise}
\end{cases}
\end{align}

\centerline{\textit{3. Bounding} $\quad\epsilon_3^{co}$}
\begin{align}
	\ltwo{\ket{\epsilon_3^{co}}} &= \sqrt{d} \sum_{l \in \mathcal{S}_d(k_0)} \sum_{n=0}^{d-1} \sum_{m\in\mathbb{Z}} \frac{A\sigma}{\sqrt{d}} \abs{m \left( -\frac{2\pi\sigma^2}{d^2}(p-n_0) - i \frac{2\pi k_0}{d} \right) e^{-\frac{\pi\sigma^2}{d^2}(p-n_0)^2} }_{p=n+md} \ket{\theta_l}.
\end{align}
Trivial sum w.r.t. $l$, bound the error by twice the right side error (w.r.t. $m$), choose $n=-d/2$ for worst case scenario, then trivial sum over $n$, and express $n_0$ in terms of $\alpha$,
\begin{align}
	\ltwo{\ket{\epsilon_3^{co}}} &= 2 A d^2 \sigma 2\pi \sum_{m=1}^\infty \left[ \left( m - \frac{1+\alpha}{2} \right) + \frac{1+\alpha}{2} \right] \left( \frac{\sigma^2}{d} \left( m - \frac{1+\alpha}{2} \right) + \frac{\beta}{2} \right) e^{-\pi\sigma^2 \left( m - \frac{1+\alpha}{2} \right)^2}.
\end{align}
Applying all of the bounds on the tails of Gaussians on each of the 4 terms, 
\begin{align}
	\frac{\ltwo{\ket{\epsilon_3^{co}}}}{2\pi A d^2} &= 2 \sigma  \cdot \frac{\sigma^2}{d} \left[ \left( \frac{1-\alpha}{2} \right)^2 + \frac{1}{2\pi\sigma^2} \left( \frac{1-\alpha}{2} + 1 + \frac{1}{1 - e^{-\pi\sigma^2(1-\alpha)}} \right) \right] e^{-\frac{\pi\sigma^2}{4}(1-\alpha)^2} \\
	&+ 2 \sigma  \cdot \frac{\sigma^2}{d} \left( \frac{1+\alpha}{2} \right) \left( \frac{1-\alpha}{2} + \frac{1}{2\pi\sigma^2} \right) e^{-\frac{\pi\sigma^2}{4}(1-\alpha)^2} \\
	&+ 2 \sigma  \cdot \frac{\beta}{2} \left( \frac{1-\alpha}{2} + \frac{1}{2\pi\sigma^2} \right) e^{-\frac{\pi\sigma^2}{4}(1-\alpha)^2} \\
	&+ 2 \sigma  \cdot \left( \frac{1+\alpha}{2} \right) \frac{\beta}{2} \left( \frac{1}{1 - e^{-\pi\sigma^2(1-\alpha)}} \right) e^{-\frac{\pi\sigma^2}{4}(1-\alpha)^2} \\
	=&
	\begin{cases}\displaystyle
  \sqrt{d} \left( (1-\alpha) + \frac{1}{\pi d} \left( 2 + \frac{1}{1 - e^{-\pi d (1-\alpha)}} \right) + \frac{\beta}{2} \left( 1 - \alpha + \frac{1}{\pi d} + \frac{1+\alpha}{1 - e^{-\pi d(1-\alpha)}} \right) \right) e^{-\frac{\pi d}{4}(1-\alpha)^2} &\mbox{if } \sigma=\sqrt{d}\quad\quad\quad\quad\quad\quad\quad\quad\quad\quad\quad\quad\quad\quad\quad\quad\quad\quad\quad\quad \\[10pt]\displaystyle
 \sigma \left( \frac{\sigma^2}{d} (1-\alpha) + \frac{1}{\pi d} \left( 2 + \frac{1}{1 - e^{-\pi\sigma^2(1-\alpha)}} \right) + \frac{\beta}{2} \left( 1 - \alpha + \frac{1}{\pi\sigma^2} + \frac{1+\alpha}{1 - e^{-\pi\sigma^2(1-\alpha)}} \right) \right) e^{-\frac{\pi\sigma^2}{4}(1-\alpha)^2} &\mbox{otherwise}\quad\quad\quad\quad\quad\quad\quad\quad\quad\quad\quad\quad\quad\quad\quad\quad\quad\quad\quad\quad
\end{cases}
\end{align}

\centerline{\textit{4. Bounding} $\quad\epsilon_4^{co}$}
\begin{align}
	\ltwo{\ket{\epsilon_4^{co}}} &= \sum_{l \in \mathcal{S}_d(k_0)} \sum_{m\in\mathbb{Z}/\{0\}} \abs{ \frac{d}{dx} \left( x\psi(k_0;x) \right) }_{x=l+md} = \sum_{k\in\mathbb{Z}/\mathcal{S}_d(k_0)} \abs{ \frac{d}{dx} \left( x\psi(k_0;x) \right) }_{x=k}\\
	&= \sum_{k\in\mathbb{Z}/\mathcal{S}_d(k_0)} \abs{ \left( 1 - \frac{2\pi}{\sigma^2}x(x-k_0) + i \frac{2\pi n_0}{d} x \right) e^{-\frac{\pi}{\sigma^2}(x-k_0)^2} }_{x=k}
\end{align}
Once again, bounding the error by twice the right side error (w.r.t. $m$), and expressing $k_0$ in terms of $\beta$,
\begin{align}
	\ltwo{\ket{\epsilon_4^{co}}} &< 2d A \sum_{m=1}^\infty \abs{ 1 - \frac{2\pi}{\sigma^2}  (k-k_0)^2 -\frac{2\pi}{\sigma^2} k_0(k-k_0) + i \pi\alpha (k-k_0+k_0) } e^{-\frac{\pi}{\sigma^2} (k-k_0)^2}
\end{align}

Bounding each sum as a Gaussian tail, (first term is ignored as it will only make the total smaller, being of the opposite sign)
\begin{align}
	\ltwo{\ket{\epsilon_4^{co}}} &< 2d A \frac{2\pi}{\sigma^2} \left( \frac{d^2}{4} + \frac{\sigma^2}{2\pi} \left( \frac{d}{2} + 1 + \frac{1}{1 - e^{-\frac{\pi d}{\sigma^2}}} \right) \right) e^{-\frac{\pi d^2}{4\sigma^2}} \\
	&+ 2dA \frac{2\pi}{\sigma^2} \frac{\beta d}{2} \left( \frac{d}{2} + \frac{\sigma^2}{2\pi} \right) e^{-\frac{\pi d^2}{\sigma^2}} \\
	&+ 2dA \pi \alpha \left( \frac{d}{2} + \frac{\sigma^2}{2\pi} \right) e^{-\frac{\pi d^2}{\sigma^2}} \\
	&+ 2dA \pi \alpha \frac{\beta}{2} \frac{1}{1 - e^{-\frac{\pi d}{\sigma^2}}} e^{-\frac{\pi d^2}{\sigma^2}} \\
	&< 
	\begin{cases}\displaystyle
 dA \left( d (\pi+1) (1+\beta) + 2 + \frac{2}{1 - e^{-\pi}} + \alpha \left( d (\pi+1) + \frac{\pi\beta}{1 - e^{-\pi}} \right) \right) e^{-\frac{\pi d}{4}}  &\mbox{if } \sigma=\sqrt{d} \\[10pt]\displaystyle
 dA \left( \left( \frac{\pi d^2}{\sigma^2} + d \right)(1+\beta) + 2 + \frac{2}{1 - e^{-\frac{\pi d}{\sigma^2}}} + \alpha \left( \pi d + \sigma^2 + \frac{\pi\beta}{1 - e^{-\frac{\pi d}{\sigma^2}}} \right) \right) e^{-\frac{\pi d^2}{4\sigma^2}}  &\mbox{otherwise}
\end{cases}
\end{align}

\textit{ALTERNATE BOUND}
\begin{align}
	\ltwo{\ket{\epsilon_4^{co}}} &= \sum_{l \in \mathcal{S}_d(k_0)} \sum_{m\in\mathbb{Z}/\{0\}} \abs{ \frac{d}{dx} \left( x\psi(k_0;x) \right) }_{x=l+md} \\
	&= \sum_{l \in \mathcal{S}_d(k_0)} \sum_{m\in\mathbb{Z}/\{0\}} \abs{ \left( 1 - \frac{2\pi}{\sigma^2}x(x-k_0) + i \frac{2\pi n_0}{d} x \right) e^{-\frac{\pi}{\sigma^2}(x-k_0)^2} }_{x=l+md}
\end{align}
Once again, bounding the error by twice the right side error (w.r.t. $m$), and replacing $l$ by its worst case, and trivializing the sum w.r.t. $l$, and finally expressing $k_0$ in terms of $\beta$,
\begin{align}
	\ltwo{\ket{\epsilon_4^{co}}} &< 2d A \sum_{m=1}^\infty \abs{ \left[ 1 - \frac{2\pi d^2}{\sigma^2} \left( m - \frac{1}{2} \right)^2 -\frac{2\pi d^2}{\sigma^2} \frac{\beta}{2}\left( m - \frac{1}{2} \right) + i 2\pi d \frac{\alpha}{2} \left( m - \frac{1}{2} + \frac{\beta}{2} \right) \right] } e^{-\frac{\pi d^2}{\sigma^2} \left( m - \frac{1}{2} \right)^2}
\end{align}

Once again, bounding each sum as a Gaussian tail,
\begin{align}
	\ltwo{\ket{\epsilon_4^{co}}} &< 2d A \frac{2\pi d^2}{\sigma^2} \left( \frac{1}{4} + \frac{\sigma^2}{2\pi d^2} \left( \frac{1}{2} + 1 + \frac{1}{1 - e^{-\frac{\pi d^2}{\sigma^2}}} \right) \right) e^{-\frac{\pi d^2}{4\sigma^2}} \\
	&+ 2dA \frac{2\pi d^2}{\sigma^2} \frac{\beta}{2} \left( \frac{1}{2} + \frac{\sigma^2}{2\pi d^2} \right) e^{-\frac{\pi d^2}{\sigma^2}} \\
	&+ 2dA \pi d \alpha \left( \frac{1}{2} + \frac{\sigma^2}{2\pi d^2} \right) e^{-\frac{\pi d^2}{\sigma^2}} \\
	&+ 2dA \pi d \alpha \frac{\beta}{2} \frac{1}{1 - e^{-\frac{\pi d^2}{\sigma^2}}} e^{-\frac{\pi d^2}{\sigma^2}}.
\end{align}

\centerline{\textit{5. Bounding} $\quad\epsilon_5^{co}$}
\begin{align}
	\ltwo{\ket{\epsilon_5^{co}}} &= \frac{1}{\sqrt{d}} \sum_{n=0}^{d-1} \sum_{k \in \mathbb{Z}/\mathcal{S}_d(k_0)} \abs{\psi(k_0;k)}, \\
	&=\sqrt{d} \sum_{k \in \mathbb{Z}/\mathcal{S}_d(k_0)} e^{-\frac{\pi}{\sigma^2}(k-k_0)^2} & & \text{trivial sum w.r.t $n$} \\
	&< 2\sqrt{d} A \sum_{k-k_0=\frac{d}{2}}^\infty e^{-\frac{\pi}{\sigma^2}(k-k_0)^2} & & \text{bound error by right side} \\
	&<
	\begin{cases}\displaystyle
 2\sqrt{d} A \frac{ e^{-\frac{\pi d}{4}}}{1 - e^{-\pi}} &\mbox{if } \sigma=\sqrt{d} \\[10pt]\displaystyle
 2\sqrt{d} A \frac{ e^{-\frac{\pi d^2}{4\sigma^2}}}{1 - e^{-\frac{\pi d}{\sigma^2}}}  &\mbox{otherwise}
\end{cases}
\end{align}

\centerline{\textit{6. Bounding} $\quad\epsilon_6^{co}$}
\begin{align}
	\ltwo{\ket{\epsilon_6^{co}}} &< \frac{2 \pi}{T_0} \sqrt{d} \sum_{l \in \mathcal{S}_d(k_0)} \sum_{n=0}^{d-1} \sum_{m \in \mathbb{Z}} \abs{m \tilde{\psi}(k_0;n+md)}
\end{align}

Trivial sum w.r.t. $l$, bounding the error by the right side (w.r.t. $m$) and replacing $n$ by the worst case $n=-\frac{d}{2}$,
\begin{align}
	\ltwo{\ket{\epsilon_6^{co}}} &= \frac{2 \pi}{T_0} d^2 \sigma A \sum_{m=1}^\infty \left( \left( m - \frac{1+\alpha}{2} \right) + \frac{1+\alpha}{2} \right) e^{-\pi^2\sigma^2 \left( m - \frac{1+\alpha}{2} \right)^2} \\
	&< 
	\begin{cases}\displaystyle
 \frac{2\pi A d^2 \sqrt{d}}{T_0} \left( \frac{1-\alpha}{2} + \frac{1}{2\pi d} + \left( \frac{1+\alpha}{2} \right) \frac{1}{1 - e^{-\pi d(1-\alpha)}} \right) e^{-\frac{\pi d}{4}(1-\alpha)^2}  &\mbox{if } \sigma=\sqrt{d} \\[10pt]\displaystyle
 \frac{2\pi A d^2 \sigma}{T_0} \left( \frac{1-\alpha}{2} + \frac{1}{2\pi\sigma^2} + \left( \frac{1+\alpha}{2} \right) \frac{1}{1 - e^{-\pi\sigma^2(1-\alpha)}} \right) e^{-\frac{\pi\sigma^2}{4}(1-\alpha)^2}  &\mbox{otherwise}
\end{cases}
\end{align}

\centerline{\textit{7. Bounding} $\quad\epsilon_7^{co}$}
\begin{align}
	\ltwo{\ket{\epsilon_7^{co}}} &= \frac{d}{T_0} \sum_{l \in \mathcal{S}_d(k_0)} \sum_{m \in \mathbb{Z}/\{0\}} \abs{\frac{d}{dx} \psi(k_0;x)}_{x=l+md} = \frac{dA}{T_0} \sum_{k\in\mathbb{Z}/\mathcal{S}_d(k_0)} \abs{\frac{d}{dx} \psi(k_0;x)}_{x=k}\\
	&< \frac{2d}{T_0} \sum_{k-k_0=\frac{d}{2}}^\infty \abs{ \frac{2\pi}{\sigma^2}(k-k_0) + i \frac{2\pi n_0}{d} } e^{-\frac{\pi}{\sigma^2}(k-k_0)^2} \\
	&< 
	\begin{cases}\displaystyle
 \frac{2\pi}{T_0} d A \left( 1 + \frac{1}{\pi} + \frac{\alpha}{1 - e^{-\pi}} \right) e^{-\frac{\pi d}{4}}  &\mbox{if } \sigma=\sqrt{d} \\[10pt]\displaystyle
 \frac{2\pi}{T_0} d A \left( \frac{d}{\sigma^2} + \frac{1}{\pi} + \frac{\alpha}{1 - e^{-\frac{\pi d}{\sigma^2}}} \right) e^{-\frac{\pi d^2}{4\sigma^2}}  &\mbox{otherwise}
\end{cases}
\end{align}

\centerline{\textit{8. Bounding} $\quad\epsilon_8^{co}$}
\begin{align}
	\ltwo{\ket{\epsilon_8^{co}}} &< d A \sum_{k=\frac{d}{2}}^\infty \abs{ \frac{2\pi}{\sigma^2}(k-k_0) + i \frac{2\pi n_0}{d} } e^{-\frac{\pi}{\sigma^2}(k-k_0)^2} \\
	&<  
	\begin{cases}\displaystyle
\pi d A \left( 1 -\beta + \frac{1}{\pi} + \frac{\alpha}{1 - e^{-\pi(1-\beta)}} \right) e^{-\frac{\pi d}{4}(1-\beta)^2}  &\mbox{if } \sigma=\sqrt{d} \\[10pt]\displaystyle
\pi d A \left( \frac{d}{\sigma^2}(1-\beta) + \frac{1}{\pi} + \frac{\alpha}{1 - e^{-\frac{\pi d}{\sigma^2}(1-\beta)}} \right) e^{-\frac{\pi d^2}{4\sigma^2}(1-\beta)^2}  &\mbox{otherwise}
\end{cases}
\end{align}

\end{appendices} 



\end{document}